\documentclass[11pt]{article}

\usepackage{graphicx} 
\usepackage[colorlinks=true, allcolors=blue]{hyperref}
\usepackage{alltt}
\usepackage{amsfonts}
\usepackage{amsmath, amsthm, bm}
\usepackage{amssymb}
\usepackage{subcaption}
\usepackage{xcolor}
\usepackage{prodint}
\usepackage{setspace}

\usepackage[bb=dsserif]{mathalpha}
\usepackage[
    style=apa,
  ]{biblatex}
\usepackage{enumitem}
\newtheorem{assumption}{Assumption}

\newtheoremstyle{boldremark}
    {\dimexpr\topsep/2\relax} 
    {\dimexpr\topsep/2\relax} 
    {}          
    {}          
    {\bfseries} 
    {.}         
    {.5em}      
    {}          

\DeclareMathOperator{\E}{\text{E}}
\DeclareMathOperator{\var}{\text{var}}
\DeclareMathOperator{\cov}{\text{cov}}

\DeclareMathOperator*{\argmin}{arg\,min}
\newcommand\norm[1]{\left\lVert#1\right\rVert}
\newcommand\dd{\mathop{}\!\mathrm{d}}
\newtheorem{theorem}{Theorem}
\newtheorem{lemma}{Lemma}
\newtheorem{corollary}{Corollary}
\theoremstyle{boldremark}
\newtheorem{remark}{Remark}
\newtheorem{alemma}{Lemma}[section]

\title{Variable importance measures for heterogeneous treatment effects with survival outcome}
\author{Simon Christoffer Ziersen $\&$ Torben Martinussen \\ \\ \small \textit{Section of Biostatistics, University of Copenhagen, Denmark}}
\date{}

\begin{document}

\maketitle

\begin{abstract}
    Treatment effect heterogeneity plays an important role in many areas of causal inference and within recent years, estimation of the conditional average treatment effect (CATE) has received much attention in the statistical community. While accurate estimation of the CATE-function through flexible machine learning procedures provides a tool for prediction of the individual treatment effect, it does not provide further insight into the driving features of potential treatment effect heterogeneity. Recent papers have addressed this problem by providing variable importance measures for treatment effect heterogeneity. Most of the suggestions have been developed for continuous or binary outcome, while little attention has been given to censored time-to-event outcome.
    In this paper, we extend the treatment effect variable importance measure (TE-VIM) proposed in \textcite{hines} to the survival setting with censored outcome. We derive an estimator for the TE-VIM for two different CATE functions based on the survival function and RMST, respectively. Along with the TE-VIM, we propose a new measure of treatment effect heterogeneity based on the best partially linear projection of the CATE and
    suggest accompanying 
    estimators for that projection. All estimators are based on semiparametric efficiency theory, and we give conditions under which they are asymptotically linear. The finite sample performance of the derived estimators are investigated in a simulation study. Finally, the estimators are applied and contrasted in two real data examples.
\end{abstract}


\section{Introduction}
Understanding treatment effect heterogeneity is important for personalizing treatment plans as well as informing further pharmacological/medical research. The former point has received much attention in the causal inference community within the past decade, see for example \textcite{kennedy} and \textcite{wagerAthey}. The work has focused on the Conditional Average Treatment Effect (CATE) given by the difference $\tau(x) = \E(Y^1 - Y^0\mid X=x)$, where $Y^1$ and $Y^0$ are the counterfactual outcomes under treatment and no treatment, respectively, and $X$ denotes a set of covariates. Under standard assumptions from the causal inference literature, including the assumption of no unmeasured confounding, the CATE can be identified from the observed data $\mathcal{O} = (Y_i, A_i, X_i)_{i=1}^n$, where $Y_i, A_i, X_i$ correspond to the outcome, treatment and covariates of individual $i$, as $\tau(x) = \E(Y\mid A=1, X=x) - \E(Y\mid A=0, X=x)$. 

 Considering counterfactual survival times $T^1$ and $T^0$, and letting $Y^a(t) = \mathbb{1}(T^a\geq t)$, $a = 0,1$, the CATE in the survival setting may be defined as $\tau(x;t) = \E\{Y^1(t) - Y^0(t)\mid X=x\}$, which, under the same causal assumptions, is identified by the observed data as $\tau(x;t) = S(t\mid A=1, X=x) - S(t\mid A=0, X=x)$ for a specific time horizon $[0,t]$, and where $S(t\mid A=a, X=x)$  denotes the conditional survival function.
 Survival analysis is often complicated by the fact that one does not observe the full data, but only a censored version given by $\mathcal{O} = (\tilde{T}_i, \Delta_i, A_i, X_i)_{i=1}^n$ where $\tilde{T}_i = T_i \wedge C_i$ for a given censoring time $C_i$ and $\Delta_i = \mathbb{1}(T_i \leq C_i)$. Under the additional assumption of conditional independent censoring given  $A$ and $X$, the CATE is still identified from the observed data as the above difference in conditional survival functions. Estimation of the CATE in the survival context has received some attention in the recent years: \textcite{cui} extend the work of \textcite{wagerAthey} to a survival setting, \textcite{hu} compares different machine learning methods for estimating the CATE in a survival setting and \textcite{xu} discuss the use of different meta-learners in combination with arbitrary machine learning methods.  

CATE estimation provides a tool for prediction of the individual treatment effect, but as the methods of obtaining such estimates are often based on machine learning, it provides little information as to which features are driving the observed heterogeneity (if any at all). As such, \textcite{levy} derives a measure of overall treatment effect heterogeneity as the variance of the treatment effect (VTE), given by $\var\{\tau(X)\}$, \textcite{wei} derives an estimator for sub-group treatment effects, and \textcite{boileau} constructs a general framework for identification of treatment effect modifiers, as a weighted covariance of individual covariates and the CATE, which they also extend to a survival setting. Their approach can be viewed in terms of the \textit{best linear projection} of the CATE-function, an approach also discussed in \textcite{vanLaan} and \textcite{semenova}, but where the projection is used to approximate a target function (such as the CATE-function) rather than summary statistics of the CATE itself. Finally, \textcite{hines} develop a treatment effect variable importance measure (TE-VIM), which measures the amount of the VTE explained by a given subset of covariates. Their derived estimand has the interpretation of a non-parametric ANOVA and can employ arbitrary machine learning methods for nuisance parameter estimation. 

In this paper, we extend firstly the TE-VIM of \textcite{hines} for two different CATE functions for survival data. The derived estimator is based on semiparametric efficiency theory, and the efficient influence function (EIF) corresponding to the TE-VIM with censored data is seen to share some structure to the one proposed by \textcite{hines}. This connection is found to hold for essentially all $\tau(x)$, when the EIF corresponding to the ATE, $\E\{\tau(X)\}$, is linear in the ATE. Secondly, we derive a new measure of treatment effect heterogeneity inspired by the assumption lean inference approach (\cite{stijn}) and derive an estimator based on its corresponding efficient influence function. The new measure is derived as the \textit{best partially linear projection} of the CATE and it can be interpreted as a regression coefficient, expressing the association between the CATE and a single covariate of interest. Other authors have suggested a similar approach (\cite{boileau}, \cite{cui}) for treatment effect variable importance, using the best linear projection of the CATE as a measure of heterogeneity. However, as we discuss in the Appendix, the error made by the projecting the CATE onto the linear model is larger compared to the projection onto the partially linear model, thus showing that our approach captures more of the heterogeneity through a single covariate compared to the best linear projection. Furthermore, the derived parameter is seen to provide a natural interpretation of the association between the CATE and a given covariate when the partially linear model does not hold for the CATE function, as it is given as weighted average of the conditional covariance of the CATE and the covariate in question.

We give assumptions under which the proposed estimators are asymptotically normal and locally efficient and investigate their finite sample performance in a simulation study, using random survival forests (\cite{ishwaran}) for nuisance parameter estimation. Finally, we illustrate and contrast the two approaches in two data examples. The first example is also studied in \textcite{cui} and \textcite{hines}. The second example considers data from the LEADER study, \textcite{marso}. 

In Section 2 we state the notation and setup used in the paper and in Section 3 we define two target parameters, each being a measure of treatment effect heterogeneity. In Section 4, we derive the efficient influence functions for the two target parameters and utilize these to construct cross-fitted one-step estimators. The asymptotic distributions of the estimators are then proved under high level assumptions on the nuisance parameter estimates. In Section 5, the finite sample performance of the derived estimators are investigated in a simulation study and in Section 6 we apply the estimators to two data examples. Section 7 concludes the paper with some final remarks.

\section{Notation and Setup}
Let $T$ and $C$ denote the survival and censoring time, respectively. Due to censoring, we do not observe $T$, but rather $\tilde{T} = T \wedge C$ together with the event indicator $\Delta = \mathbb{1}\{T \leq C\}$. Let $A \in \{0,1\}$ denote the  treatment indicator  at baseline and let $X = (X_1, ..., X_d)$ denote baseline covariates. With $O_i=(\tilde{T}_i, \Delta_i, A_i, W_i)$,
the observed data  $\mathcal{O}$  consist of $O_1, ..., O_n$ that are assumed  i.i.d. with distribution $P_0 \in \mathcal{M}$, where $\mathcal{M}$ is the set of all probability measures corresponding to a non-parametric model. 
Let $N(t) = \mathbb{1}\{\tilde{T} \leq t, \Delta = 1\}$ be the observed counting process for the event of interest and let $\lambda(t|a,x)$, $\lambda_c(t|a,x)$ denote the conditional hazard for the survival and censoring distribution, respectively, and let $\Lambda(t|a,x)$, $\Lambda_c(t|a,x)$ denote the corresponding cumulative hazard functions. Furthermore we denote the conditional survival function by $S(t|a,x)$, and let $S_c(t|a,x) $ denote the conditional survival function of the censoring distribution. We let $\pi(a|x) = P(A=a|X=x)$ denote the propensity score and $\mu$ is the distribution of $X$. Throughout, $dM(t\mid A, X) = dN(t) - \mathbb{1}(\tilde{T} \geq t) d\Lambda(t|A, X)$ is the counting process martingale increment given $A$ and $X$.

\noindent
To define causal parameters, we introduce the variable $Y(t) = \mathbb{1}\{T \geq t \} $ and define $Y^a(t)$ as the counterfactual outcome, that is, the outcome of a person if he or she, possibly contrary to the fact, had received treatment $a$. Let 
$$
\tau(x;t) = \E\{Y^1(t) - Y^0(t) | X = x\}
$$ 
be the CATE function, i.e. the average treatment effect conditional on  $X=x$ for some fixed time-horizon $t$, which is left out from the notation throughout the paper, so we write $\tau(x) = \tau(x;t)$. Under standard causal assumptions we can identify $\tau$ through the observed data as 
\begin{equation}\label{cateSurv}
\tau(x) = S(t|A=1, X=x) - S(t|A=0, X=x).    
\end{equation}
An alternative $\tau(x)$ is
\begin{equation}\label{cateRMST}
\tau(x) = E\left(T^{1}\wedge t-T^{0}\wedge t \mid X=x\right)=\int_0^{t} S(u\mid 1,x)\, du-\int_0^{t} S(u\mid 0,x)\, du.
\end{equation}
We will consider both in what follows, where we will refer to the first as the \textit{survival function setting} and to the second as the \textit{restricted mean survival time setting (RMST)}. 
Furthermore we define 
\begin{align*}
\tau_l(x) = \E\{\tau(X)|X_{-l} = x_{-l}\}
\end{align*}
as the conditional expectation of the CATE-function, where we fix all variables except $X_l$, $l \subseteq \{1, ..., d\}$, as we define $X_{-l}$ to be the covariates with an index not contained in $l$. We will use the notation $\tau_d = \E\{\tau(X)\}$ to denote the average treatment effect. 
Furthermore, we introduce the nuisance parameter $\nu = (\Lambda, \Lambda_c, \tau_l, \mu)$. Finally, we let $\norm{\cdot}$ denote the $L_2(P)$-norm, unless otherwise specified, such that $\norm{f} = \left(\int f^2 \dd P\right)^{\frac{1}{2}}$.

\section{Target parameter}\label{sec:targetParameter}
\subsection{Treatment effect variable importance measure}
As in \textcite{hines} we define
\begin{align*}
\Theta_l \equiv \E[\var\{\tau(X) \ | \ X_{-l} = x_{-l}\}] = \var\{\tau(X)\} - \var\{\tau_l(X)\} \geq 0.
\end{align*}
With a slight abuse of notation we denote the VTE: $\Theta_d = \var\{\tau(X)\}$. We note that $\Theta_l$ can be interpreted as the amount of heterogeneity not already explained by $X_{-l}$, as $\Theta_l$ is large when a large amount of the VTE is explained by $X_l$. The proposed treatment effect variable importance measure (TE-VIM) is defined by re-scaling $\Theta_l$ by the VTE:
\begin{equation}\label{psi}
\Psi_l \equiv \frac{\Theta_l}{\Theta_d} = 1 - \frac{\var\{\tau_l(X)\}}{\var\{\tau(X)\}} \nonumber
\end{equation}
with values in $[0,1)$. We can interpret $\Psi_l$ as a nonparametric analog of an ANOVA statistic, which is close to one when a large amount of the VTE is explained by $X_l$ and close to zero when a small amount of the VTE is explained by $X_l$. 

\subsection{Best partially linear projection}
Along with the TE-VIM, we consider an alternative target parameter inspired by \textcite{vansteelandt}, which is given by
\begin{equation}\label{gamma}
    \Omega_j \equiv \frac{\Gamma_j}{\chi_j} = \frac{\E\left [\cov\{X_j, \tau(X) \mid X_{-j}\} \right]}{\E\{\var (X_j\mid X_{-j})\}} \nonumber
\end{equation}
for a single covariate $X_j$.
It is seen that the parameter may depend on the scale of the covariate of interest, $X_j$, and as such, we propose that the variable importance of $X_j$ is based on the corresponding test-statistic   concerning  the hypothesis that $\Omega_j=0$. 
In contrast to $\Psi_l$, the parameter $\Omega_j$ measures the heterogeneity explained by a single covariate $X_j, \ j \in \{1,\cdots, d \}$, whereas $\Psi_l$ may be used to determine the heterogeneity explained by, possibly non-singular, sets of covariates. This makes $\Psi_l$ potentially better suited for incorporating subject matter knowledge, where naturally correlated covariates can be grouped together, where $\Omega_j$ serves as a variable importance measure to be used for single covariates.

The estimand $\Omega_j$ can be expressed by the linear term in the projection of $\tau$ onto the space of partially linear functions. To elaborate, let $\beta \in \mathbb{R}$ and let $w$ be some measurable function of $X_{-j}$ with finite variance. Without loss of generality, define
$$
\tau(x) = \beta x_{j} + w(x_{-j}) + R(x_j, x_{-j})
$$
for some function $R$ and let
$$
(\beta^*, w^*) = \argmin_{\beta, w}\E\{ R(X_j, X_{-j})^2 \} = \argmin_{\beta, w}\E\{ [\tau(X) - \beta X_j - w(X_{-j})]^2 \}
$$
be the least squares projection of $\tau$ onto the partially linear model. Then $\Omega_j = \beta^*$. 
If $R = 0$, then for a given level of $x_{-j}$, the parameter $\beta$ denotes the treatment effect modification given by $x_j$. When $R \neq 0$, $\Omega_j$ is the treatment effect modification parameter that minimizes the error made by summarizing the effect of $x_j$ in a single value.
We note that other authors have looked at the \textit{best linear projection} as a treatment effect variable importance measure (\cite{semenova}, \cite{boileau}, \cite{cui}, \cite{vanLaan}). In our setting, this corresponds to the least squares projection of $\tau$ onto the space of linear models. In Appendix \ref{App_A}, we give a discussion of partially linear versus linear projections of $\tau$ in terms of the size of the error given by the remainder term $R$. In a given application, the partially linear model may fail to hold for $\tau$, but $\Omega_j$ still provides an interpretable measure of heterogeneity as a weighted average of the conditional covariance of $\tau$ and $X_j$.

\section{Estimation}\label{sec:estimation}
The estimation of $\Psi_l$ goes through estimation of $\Theta_l$ and $\Theta_d$, separately. Likewise, an estimator of $\Omega_j$ is obtained from estimators of $\Gamma_j$ and $\chi_j$. The estimation of individual parameters is based on semiparametric efficiency theory. For an introduction to this methodology see for instance (\cite{kennedydouble}, \cite{hinesdemystifying}, \cite{robinslaan}, \cite{vaart}, Ch. 25). The theory revolves around the so-called efficient influence function (EIF), which characterizes the lower bound on the asymptotic variance of any regular estimator of a pathwise differentiable parameter in a non-parametric setting. The EIF is related to the target parameter and the model $\mathcal{M}$, and it can be calculated without reference to any estimator. Once it is known, it can be leveraged to construct an estimator that is asymptotically linear with the EIF as its influence function. Several techniques exist for constructing such estimators, and they all share the convenient property that it is possible to use data-adaptive nuisance parameter estimators (under some conditions), while still obtaining parametric-like inference on the target parameter.
Hence, estimation of $\Psi_l$ and $\Omega_j$ will follow the same pattern, where the corresponding EIF is calculated initially and then used to construct  an estimator for the specific target parameter in question.

\subsection{Estimation of $\Psi_l$}
\subsubsection{Efficient influence function}
The two target parameters $\Theta_l$ and $\Psi_l$ are functions of $\var\{\tau(X)\}$ and $\var\{\tau_l(X)\}$ so their efficient influence functions can be derived from the EIFs of $\var\{\tau(X)\}$ and $\var\{\tau_l(X)\}$ using the chain rule (cf. \cite{vaart}, Ch. 25.7.). Define 
$$H(u,t\mid a, x) = \int_u^t S(v\mid a, x)\dd v$$
and 
$$g(A,X) = \left(\frac{\mathbb{1}(A=1)}{\pi(1\mid X)} - \frac{\mathbb{1}(A=0)}{\pi(0\mid X)} \right).$$
We have the following result.
\begin{theorem}\label{EIFs}
Let $\tau(x)$ be given by \eqref{cateSurv}. The efficient influence functions of $\var\{\tau(X)\}$ and $\var(\tau_l(X))$ are given by $\tilde{\psi}_{\var\{\tau(X)\}}$ and $\tilde{\psi}_{\var\{\tau_l(X)\}}$, respectively, where
\begin{align}
\tilde{\psi}_{\var\{\tau(X)\}} =& [\tau(X) - \E\{\tau(X)\}]^2 - \var\{\tau(X)\} - 2[\tau(X) - \E\{\tau(X)\}] \nonumber  \\
& \phantom{\tau(X) - \E\{\tau(X)\}]^2 } \times g(A,X)\int_0^t \frac{S(t\mid A, X)}{S(u\mid A, x)S_c(u\mid A, X)} \dd M(u\mid A, X), \nonumber \\
\tilde{\psi}_{\var\{\tau_l(X)\}} =& [\tau_l(X) - \E\{\tau_l(X)\}]^2 - \var\{\tau(X)\} - 2[\tau_l(X) - \E\{\tau_l(X)\}] \nonumber \\
& \times \left(\tau_l(X) - \tau(X) + g(A,X)\int_0^t \frac{S(t\mid A, X)}{S(u\mid A, X)S_c(u\mid A, X)} \dd M(u\mid A, X) \right).  \nonumber
\end{align}
For $\tau(x)$ given by \eqref{cateRMST} we have
\begin{align}
\tilde{\psi}_{\var\{\tau(X)\}} =& [\tau(X) - \E\{\tau(X)\}]^2 - \var\{\tau(X)\} - 2[\tau(X) - \E\{\tau(X)\}] \nonumber  \\
& \times g(A,X)\int_0^t \frac{H(u,t, A, X)}{S(u\mid A, X)S_c(u\mid A, X)} \dd M(u\mid A, X), \nonumber \\
\tilde{\psi}_{\var\{\tau_l(X)\}} =& [\tau_l(X) - \E\{\tau_l(X)\}]^2 - \var\{\tau(X)\} - 2[\tau_l(X) - \E\{\tau_l(X)\}] \nonumber \\
&\times \left(\tau_l(X) - \tau(X) + g(A,X)\int_0^t \frac{H(u,t, A, X)}{S(u\mid A, X)S_c(u\mid A, X)} \dd M(u\mid A, X) \right).  \nonumber
\end{align}
In both the survival function and RMST setting the EIF's corresponding to $\Theta_l$ and $\Psi_l$ are given by $\tilde{\psi}_{\Theta_l}$ and $\tilde{\psi}_{\Psi_l}$, respectively, where
\begin{align*}
\tilde{\psi}_{\Theta_l} &= \tilde{\psi}_{\var(\tau(X))} - \tilde{\psi}_{\var(\tau_l(X))}, \\
\tilde{\psi}_{\Psi_l} &= \frac{1}{\var(\tau(X))}\left(\tilde{\psi}_{\Theta_l}(O) - \Psi_l \tilde{\psi}_{\var(\tau(X))}(O) \right). \\
\end{align*}
\end{theorem}
\begin{proof}
    See Appendix B.
\end{proof}
Before moving to estimation of $\Theta_l$ (and $\Theta_d$) we state some results from ATE-estimation. 
The ATE  $\tau_d$ has an EIF known from the literature in the survival function setting (e.g. \cite{heleneFrank} and \cite{westling}), and we may write it in terms of the uncentered EIF, $\varphi$, defined as:
\begin{equation}\label{pseudo}
\varphi(O) - \tau_d = \varphi_{1}(O) - \varphi_{0}(O) - \tau_d
\end{equation}
with 
\begin{equation}\label{eq:varphi_surv_a}
\varphi_{a}(O) = S(t \mid A=a,X) - \frac{\mathbb{1}(A=a)}{\pi(a \mid X)}  \int_0^t\frac{S(t \mid A, X)}{S(u- \mid A, X)S_C(u- \mid A, X)} dM(u \mid A, X).    
\end{equation}
The Gateaux derivative of $\tau(x)$ in the RMST setting is
$$
\frac{\mathbb{1}(X=x)}{f(x)}\frac{\mathbb{1}(A=a)}{\pi(a \mid X)}  \int_0^t\frac{-H(u,t\mid A, X)}{S(u- \mid A, X)S_C(u- \mid A, X)} dM(u \mid A, X),  
$$
see Lemma \ref{gat:tau} in Appendix B, 
from which it follows that $\tau_d$ has efficient influence function given by 
\begin{equation}
\varphi(O) - \tau_d = \varphi_{1}(O) - \varphi_{0}(O) - \tau_d \nonumber
\end{equation}
with 
\begin{equation}\label{eq:varphi_RMST_a}
\varphi_{a}(O) = \int_0^tS(u\mid a, X) \dd u - \frac{\mathbb{1}(A=a)}{\pi(a \mid X)}  \int_0^t\frac{H(u,t\mid A, X))}{S(u- \mid A, X)S_C(u- \mid A, X)} dM(u \mid A, X), 
\end{equation}
analogous to the survival function setting. The uncentered EIF, $\varphi$, can be parameterized by parts of the nuisance parameter, $(\pi, \Lambda, \Lambda_c)$, and we write $\varphi(\pi, \Lambda, \Lambda_c)$ when we want to be explicit about the nuisance parameters considered, which will be the case when we consider estimators for $\varphi$, where $\hat{\varphi} = \varphi(\hat{\pi}, \hat{\Lambda}, \hat{\Lambda}_c)$ is an obvious candidate.
We can now restate the EIF's given in Theorem \ref{EIFs} so that the structure is similar to that given in \textcite{hines}, but with the $\varphi$, $\tau$ and $\tau_l$ having different expressions depending on the specific estimand of interest. 
\begin{corollary}\label{hinesEIF}
The EIF of $\Theta_l$ and $\Theta_d$ is given by $\tilde{\psi}_{\Theta_l}$ and $\tilde{\psi}_{\Theta_d}$, respectively, where
\begin{align}
 & \quad \tilde{\psi}_{\Theta_l} = \{\varphi(O)-\tau_l(X)\}^2 - \{\varphi(O)-\tau(X)\}^2 - \Theta_l \label{eq:eifThetas} \\
 & \quad \tilde{\psi}_{\Theta_d} = (\varphi(O)-\tau_d)^2 - \{\varphi(O)-\tau(X)\}^2 - \Theta_d   \label{eq:eifThetad} \\
 & \quad \tilde{\psi}_{\Psi_l} = \frac{1}{\Theta_d}\left(\tilde{\psi}_{\Theta_l} - \Psi_l \tilde{\psi}_{\Theta_d} \right) \label{eq:eifPsis}
\end{align}
\end{corollary}
\begin{proof}
    See Appendix B.
\end{proof}
Note that the above EIFs have the same structure whether or not we are in the survival function setting or in the RMST setting, but with $\varphi$ having a different expression. For the rest of the paper we will use the form of the EIFs given in Corollary \ref{hinesEIF}. 
\begin{remark}
The fact that the structure of the EIFs is identical to the one in \textcite{hines}, stems from the definition of EIFs as derivatives for which the chain-rule apply. From the derivations of the EIFs in the Appendix, it is seen that for any function $\tau(x)$ with Gateaux derivative given by $\mathbb{1}(X=x)g(z)/f(x)$, for some function $g$ and some variable $z$, the EIFs will have the same structure as in \eqref{eq:eifThetas} and \eqref{eq:eifThetad} with $\varphi$ being the uncentered EIF of $\E\{\tau(X)\}$ where $\varphi = \tau + g$. Thus, the framework of \textcite{hines} can readily be extended to any data setting by calculating the EIF of the ATE in that setting and denoting the uncentered version $\varphi$. The properties of estimators derived by this approach will have to be studied case by case, though, as will be apparent in the following.   
\end{remark}

\subsubsection{Cross-fitted one-step estimators}

The EIFs of $\Theta_l$ and $\Theta_d$ are used to construct estimators that are asymptotically linear with influence function given by the EIFs above, from which  they are seen to be locally asymptotically efficient and asymptotically normal distributed. Given two such estimators, $\hat{\Theta}_l$ and $\hat{\Theta}_d$, an application of the delta method gives that $\hat{\Psi}_l = \hat{\Theta}_l/\hat{\Theta}_d$ is asymptotically linear with influence function given by \eqref{eq:eifPsis} (see \cite{vaart}, Ch. 25.7). For readability and ease of notation we only consider construction of an estimator for $\Theta_l$ in the following, but since the EIFs of $\Theta_l$ and $\Theta_d$ have a similar structure, the derived estimators will be the same with $l$ replaced by $d$.

There are different ways of constructing such estimators; one-step estimators, estimating equation based, and targeted minimum loss-based estimators (TMLE). All of them require that the nuisance parameters are estimated fast enough such that the resulting remainder term and empirical process term (see Section \ref{sec:appAN} in the Appendix) converge at rate $n^{-1/2}$. We will focus on the estimating equation based estimator, which is given as the solution to $\mathbb{P}_n \tilde{\psi}_{n,\Theta_l} = 0$ in $\Theta_l$, where $\tilde{\psi}_{n,\Theta_l}$ denotes the EIF with estimated nuisance parameters. Because the EIF is linear in $\Theta_l$, this also corresponds also to the one-step estimator, where $\mathbb{P}_n\tilde{\psi}_{n,\Theta_l}$ is added to a plug-in estimate of $\Theta_l$:
\begin{equation}\label{one-step}
\hat{\Theta}_l = \mathbb{P}_n (\hat{\varphi}-\hat{\tau}_l)^2 - (\hat{\varphi}-\hat{\tau})^2, \nonumber
\end{equation}
and analogously
$$
\hat{\Theta}_d = \mathbb{P}_n (\hat{\varphi}-\hat{\tau}_d)^2 - (\hat{\varphi}-\hat{\tau})^2,
$$
where $\hat{\varphi} = \varphi(\hat{\pi}, \hat{\Lambda}, \hat{\Lambda}_c)$. Note that the estimation of $\tau_l$ can be obtained as a regression of $\hat{\tau}(X)$ onto $X_{-l}$, while estimation of $\tau_d = \E \{\tau(X)\}$ can be obtained by the mean of $\hat{\tau}(X)$, i.e., the marginal distribution, $\mu$ is estimated with the empirical measure $\mathbb{P}_n$. Or, as $\tau_d$ is itself a differentiable parameter, more sophisticated methods can be used in constructing estimators $\hat{\tau}_d$ (see Section \ref{sec:nuis}). The $n^{-1/2}$-convergence of the empirical process term related to the one-step estimator depends on the flexibility of the nuisance estimators, in the sense that, e.g., working parametric models ensure $n^{-1/2}$-convergence, which is not the case for some data-adaptive estimators. More specifically, if the nuisance estimators falls in a Donsker class which also contains the true nuisance parameter, then $n^{-1/2}$-convergence of the empirical process term is obtained. To alleviate the Donsker class condition, we employ a type of sample splitting (coined cross-fitting, \cite{cher}) that  ensures the desired convergence as long as the nuisance estimators are $L_2(P)$-consistent. In Appendix \ref{App_Sample_splitting}, we detail the sample splitting, but note that this is a general construction of cross-fitted one-step estimators (\cite{kennedydouble}). 
We denote the cross-fitted one-step estimator  by $\hat{\Theta}_l^{CF}$. The following result gives the asymptotic behavior of $\hat{\Theta}_l^{CF}$ with the needed assumptions specified in  Appendix \ref{App_Sample_splitting}

\begin{theorem}\label{asTheta}
    Assume that assumption \ref{assA} hold for the nuisance parameter estimators in each data split 
    and assume $\Theta_l > 0$. Then $\hat{\Theta}^{CF}_l$ is asymptotically linear with influence function given by $\tilde{\psi}_{\Theta_l}$, \eqref{eq:eifThetas}, and 
    \begin{equation}\label{eq:anTheta}
        n^{1/2}(\hat{\Theta}^{CF}_l - \Theta_l) \xrightarrow[]{D} \mathcal{N}(0, P\tilde{\psi}_{\Theta_l}^2). \nonumber
    \end{equation}
\end{theorem}
\begin{proof}
    See Appendix \ref{sec:appAN}.
\end{proof}

Finally, we state a distribution result for an estimate of $\Psi_l$ based on the cross-fitted estimators $\hat{\Theta}_l^{CF}$ and $\hat{\Theta}_d^{CF}$
\begin{corollary} \label{cor:psiCF}
Let $\hat{\Psi}_l^{CF} = \frac{\hat{\Theta}_l^{CF}}{\hat{\Theta}_d^{CF}}$. Under assumption \ref{assA} and $\norm{\hat{\tau}_d - \tau_d}_{L_2(P)} = o_p(n^{-\frac{1}{4}})$, $\hat{\Psi}_l^{CF}$ is asymptotically efficient with influence function given by \eqref{eq:eifPsis} and 
$$
n^{1/2}(\hat{\Psi}^{CF}_l - \Theta) \xrightarrow[]{D} \mathcal{N}(0, P\tilde{\psi}_{\Psi_l}^2).
$$
\end{corollary}
To estimate the variance of $\hat{\Psi}^{CF}_l$, we define the cross-fitted plug-in estimator $\hat{\sigma}^{2,CF}_{\Psi_l}$ of $\sigma^2_{\Psi_l} = P\tilde{\psi}_{\Psi_l}^2$. 
This is detailed in Appendix \ref{App_Sample_splitting}.

\begin{remark}
    The target parameter $\Psi_l$ is restricted to $[0,1)$ while  the estimator $\hat{\Psi}_l^{CF}$ is unrestricted. This  can result in parameter estimates that are outside the range $[0,1)$, or confidence intervals that contain either 0 or 1. One way to obtain estimates strictly in $[0,1)$ is to construct a cross-fitted one-step estimator of the transformed parameter $\text{logit}(\Psi_l)$ and then use an expit-transformation to obtain an estimate and confidence interval of $\Psi_l$ in the range $[0,1)$. This approach is detailed in Appendix \ref{App-logit}.
\end{remark}

\subsection{Estimation of $\Omega_j$}
As $\Omega_j$ is a ratio of two $\Gamma_j$ and $\chi_j$, the construction of an estimator will follow the same procedure as with $\Psi_l$.
\subsubsection{Efficient influence function}
We derive the EIF's of $\Gamma_j$ and $\chi_j$ separately from which the EIF of $\Omega_j$ is obtained. This is summarized in the following theorem.
\begin{theorem}\label{EIFomega}
    Let $\varphi$ be given as in \eqref{eq:varphi_surv_a} for the survival function setting, and as in \eqref{eq:varphi_RMST_a} for the RMST setting. Then the EIF's of $\Gamma_j$ and $\chi_j$ are given by $\tilde{\psi}_{\Gamma_j}$ and $\tilde{\psi}_{\chi_j}$, respectively, where
    \begin{align}\label{eq:EIFGamma}
        \tilde{\psi}_{\Gamma_j} = \{ \varphi(O) - \tau_j(X) \}\{ X_j - \E(X_j\mid X_{-j}) \} - \Gamma_j
    \end{align}
    and
    \begin{align}\label{eq:EIFchi}
        \tilde{\psi}_{\chi_j} = \{X_j - \E(X_j\mid X_{-j})\}^2 - \chi_j.
    \end{align}
    The EIF of $\Omega_j$ is given by
    \begin{align}\label{eq:EIFomega}
        \tilde{\psi}_{\Omega_j} = \frac{1}{\chi_j}\left(\tilde{\psi}_{\Gamma_j} - \Omega_j \tilde{\psi}_{\chi_j} \right).
    \end{align}
\end{theorem}
\begin{proof}
    See Appendix  \ref{AppB}.
\end{proof}

\begin{remark}
    The EIF of $\Gamma_j$ is stated in terms the $\varphi$, which is the uncentered EIF of the ATE. Thus, the above EIF's are readily  extended to other data settings with different CATE functions, $\tau$. 
\end{remark}

\subsubsection{Cross-fitted one-step estimators}
As with $\Psi_l$, we construct cross-fitted one-step estimators for $\Gamma_j$ and $\chi_j$ based on the EIF's in Theorem \ref{EIFomega}. The needed assumptions (Assumptions A and B) are given in  Appendix 
\ref{App_Sample_splitting}.

\begin{theorem}\label{asGammaAndChi}
    Under assumption \ref{assA} with \ref{ass2} replaced by assumption \ref{assB}, $\hat{\Gamma}_j^{CF}$ is asymptotically linear with influence function given by the EIF \eqref{eq:EIFGamma}, and hence locally asymptotically efficient. Thus
    \begin{align}
        n^{1/2}(\hat{\Gamma}_j^{CF} - \Gamma_j) \xrightarrow[]{D} \mathcal{N}(0, P \tilde{\psi}_{\Gamma_j}^2). \nonumber 
    \end{align}
    Furthermore, under assumption \ref{assB}, $\hat{\chi}_j^{CF}$ is asymptotically linear with influence function given by the EIF \eqref{eq:EIFchi}, and hence locally asymptotically efficient. Furthermore:
    \begin{align}
        n^{1/2}(\hat{\chi}_j^{CF} - \chi_j) \xrightarrow[]{D} \mathcal{N}(0, P \tilde{\psi}_{\chi_j}^2).  \nonumber
    \end{align}
\end{theorem}
\begin{proof}
    See Appendix  \ref{sec:appAN}.
\end{proof}
And finally, a simple application of the delta method gives the main result for our estimator $\hat{\Omega}_j^{CF}=\hat{\Gamma}_j^{CF}/\hat{\chi}_j^{CF}$.
\begin{corollary}\label{cor:omega}
    Under assumption \ref{assA} with \ref{ass2} replaced by assumption \ref{assB}, $\hat{\Omega}_j^{CF}$ is asymptotically linear with influence function given by the EIF \eqref{eq:EIFomega}, and hence locally asymptotically efficient. Furthermore:
    \begin{align}
        n^{1/2}(\hat{\Omega}_j^{CF} - \Omega) \xrightarrow[]{D} \mathcal{N}(0, P \tilde{\psi}_{\Omega_j}^2).
    \end{align}
\end{corollary}

In Appendix \ref{App_Sample_splitting} we device how to 
estimate the variance $P\tilde{\psi}_{\Omega_j}^2$ using cross-fitting. We denote this estimator $\hat{\sigma}^{2,CF}_{\Omega_j}$. Lemma \ref{lem:SECF} in Appendix \ref{App_Sample_splitting} further gives the consistency of the cross-fitted variance estimators considered in this article.

As mentioned in Section \ref{sec:targetParameter}, the target parameter $\Omega_j$ may be scale sensitive, and rather than comparing the magnitude of $\Omega_j$ across different $j$'s, the variable importance is based on the corresponding test-statistic for the hypothesis $H:\Omega_j = 0$. Using Lemma \ref{lem:SECF} in Appendix \ref{App_Sample_splitting} and Slutsky's Theorem, we have the following result:
\begin{corollary}
Under the same setup as in Corollary \ref{cor:omega}, we have under the null-hypothesis, $H_0:\Omega_j=0$, that 
$$
\frac{\hat{\Omega}_j^{CF}}{\sqrt{\hat{\sigma}_{\Omega_j}^{2,CF}/n}} \overset{D}{\longrightarrow} \mathcal{N}(0,1).
$$
\end{corollary}

\subsection{Choice of nuisance parameter estimators}\label{sec:nuis}
In the construction of $\hat{\Psi}_l^{CF}$ and $\hat{\Omega}_j^{CF}$ we need estimators of the nuisance parameters $\Lambda, \Lambda_c, \pi, \tau_l$. For estimation of $\hat{\Psi}_l^{CF}$ we also need an estimator of $\tau_d$, and for estimation of $\hat{\Omega}_j^{CF}$ we further need an estimator $\hat{E}_n^j$ of $\E(X_j\mid X_{-j} = x_{-j})$. Estimators $\hat{\Lambda}, \hat{\Lambda}_c, \hat{\pi}, \hat{E}_n^j$ can be chosen by any machine learning methods of choice, where we will use $\hat{\Lambda}$ to construct $\hat{\tau}(x) = e^{-\hat{\Lambda}(t\mid, 1, x)} - e^{-\hat{\Lambda}(t\mid, 0, x)}$. The estimator $\hat{\tau}$ uses $\hat{\Lambda}$, estimated from the entire data, to predict $S(t\mid a,x)$ for $a=0,1$ and it is termed the S-learner in the literature. Other methods of estimating $\hat{\tau}$ from initial estimates of $\Lambda$ of $\pi$ are possible and an overview of such meta-learners are given in \textcite{xu}. 

For estimation of $\tau_l$, we consider a plug-in estimator utilizing the definition of the parameter: 
$$
\hat{\tau}_l(x) = \hat{E}_n\{\hat{\tau}(X)\mid X_{-l} = x_{-l}\},
$$
where $\hat{E}_n$ is a regression of the predicted $\hat{\tau}(X)$'s onto $X_{-l}$. Another possibility is to use the meta-learner given by $\hat{\tau}_l(x) = \hat{E}_n\{\hat{\varphi}(X)\mid X_{-l} = x_{-l}\}$, since $\tau_l(x) = \E\{\varphi(X)\mid X_{-l} = x_{-l}\}$. This meta-learner is a version of the DR-learner (\cite{kennedy}), if $\hat{E}_n$ and $\hat{\varphi}$ are estimated on different samples. Extending the analysis in \textcite{kennedy} to a survival setting might provide theoretical guaranties of $\hat{\tau}_l$ based on the DR-learner, as opposed to the plug-in estimator, but in testing we found that the plug-in estimator performed better. This is in line with the recommendations given in \cite{hines}. \\ \\
Estimation of $\tau_d = \E\{\tau(X)\}$ is a well-studied problem in the survival setting, and it is possible to construct an estimator $\hat{\tau}_d$ with further theoretical guaranties compared to $\hat{\tau}_l$. For a thorough analysis of $\tau_d$-estimation see \textcite{heleneFrank} and \textcite{westling}, who construct estimators based on the EIF \eqref{pseudo} and derive properties under which such estimators are asymptotically linear with influence function given by the EIF. 
We summarize the construction given in \textcite{westling} and give some further details in Appendix \ref{App_Sample_splitting}.

\section{Simulation Study}\label{sec:sim}
\subsection{Setup}
We performed a simulation study to investigate the finite sample properties of the estimators of $\Psi_l$ and $\Omega_j$ with and without cross-fitting in the survival function setting. We generated data from the following models:
\begin{align*}
    \Lambda(t\mid A, X) &= 2 t^{0.0025} \exp \{-X_1 - X_2 -0.3X_3 + 0.1X_4 - A(2 - 0.5X_1 - 0.3X_2)\}, \\
    \Lambda_c(t\mid A, X) &= 2 t^{0.00025} \exp(-0.3X_1), \\
    \pi(1\mid X) &= \frac{1}{1 + \exp(-0.3X_1 - 0.3X_2)},
\end{align*}
where $X_j, \ j=1,\ldots, 4,$ are i.i.d standard normally distributed. Note that $\Lambda$ and $\Lambda_c$ follow Cox-models with baseline hazards given by Weibull hazards. The time-horizon is chosen as $t=10$ with the true values of $\Psi_1$ and $\Omega_1$ being approximately 0.6907 and -0.1518, respectively, in the survival setting. 

The nuisance parameter estimators are chosen in different combinations listed below and the performance of the target parameter estimators is compared between the different choices of nuisance parameter estimators. In the following, we use a naming convention of the nuisance choices in the form \textbf{\textit{A-B}} where \textbf{\textit{A}} corresponds to the choice of $\hat{\Lambda}$, $\hat{\Lambda}_c$ and $\hat{\pi}$, and \textbf{\textit{B}} corresponds to $\hat{E}_n$ (and $\hat{E}_n^1$  for the $\Omega_1$-estimation). When generalized additive models (GAMs) are used, we use the implementation given in the R-package \verb|mgcv| with smoothing function given by the default setting, thin plate regression splines. When random forests are used, we use  the implementation given in the R-package \verb|rfsrc| with hyperparameters given by the default settings, see \cite{rfsrc}. 
Specifically, we refer to the different  nuisance parameter estimator as follows.

\begin{description}
    \item[\textit{correct-GAM}] $\hat{\Lambda}$ and $\hat{\Lambda}_c$ are given by Breslow estimates based on correctly specified Cox models, the propensity score is estimated by a correctly specified GLM and $\hat{E}_n$ (and $\hat{E}_n^1$) is given by a GAM including all interactions. 
    \item[\textit{correct-RF}] $\hat{\Lambda}$ and $\hat{\Lambda}_c$ are given by Breslow estimates based on correctly specified Cox models, the propensity score is estimated by a correctly specified GLM and $\hat{E}_n$ (and $\hat{E}_n^1$) is estimated by a random forest. 
    \item[\textit{RF-RF}] $\hat{\Lambda}$, $\hat{\Lambda}_c$, $\hat{\pi}$ and $\hat{E}_n$ (and $\hat{E}_n^1$) are all estimated by random (survival) forests.   
    \item[\textit{RF-GAM}] $\hat{\Lambda}$, $\hat{\Lambda}_c$ and $\hat{\pi}$ are estimated by random (survival) forest and $\hat{E}_n$ (and $\hat{E}_n^1$) is estimated by a GAM.
\end{description}
Each 
of the above settings are used to estimate $\hat{\Psi}_1$, $\hat{\Psi}_1^{CF}$, $\hat{\Omega}_1$ and $\hat{\Omega}_1^{CF}$, respectively. We used $K=10$ folds for the cross-fitted estimators. To separate the cross-fitted and non-cross-fitted estimators we extend the nuisance parameter naming with a suffix \textbf{\textit{C}} such that, e.g., \textbf{\textit{RF-RF-CF}} corresponds to the estimate $\hat{\Psi}_l^{CF}$ or $\hat{\Omega}_1^{CF}$ using the nuisance estimators \textbf{\textit{RF-RF}}. For estimation of $\Psi_1$, we consider four different sample sizes, $n=1000, 2000, 3000, 4000$, and for each setting we run $N=1000$ simulations and calculate the corresponding estimators for each simulation. For estimation of $\Omega_1$ we consider four different sample sizes, $n=250, 500, 750, 1000$.  The results are presented in the next subsections.

\subsection{Results for $\Psi_1$}
\subsubsection{Correctly specified $\lambda$, $\Lambda_c$ and $\pi$}
In Figure \ref{fig:simPsiCorrect} we see the sampling distribution of the estimators of $\Psi_1$ under the different choices for the nuisance estimator $\hat{E}_n$ with and without cross-fitting. All other nuisances estimators are correctly specified according to \textbf{\textit{correct-GAM}} and \textbf{\textit{correct-RF}} in the previous subsection. From Corollary \ref{cor:psiCF} we know that the estimators based on cross-fitting should be unbiased and asymptotically normal, even when using flexible nuisance estimators, which can not be guaranteed for the corresponding estimators without cross-fitting. Indeed we see that estimators based on GAM seem to follow a normal distribution around the true value, whereas the estimator based on RF is severely biased without cross-fitting but much less so with. The results of the simulations are presented in Table \ref{tab:resPsiCorrect}. When GAM is used to estimate $\hat{E}_n$, the estimator seems to perform satisfactory according to Corollary \ref{cor:psiCF} both with and without cross-fitting. When RF is used for $\hat{E}_n$, we get a large bias without sampling splitting, as also seen in Figure \ref{fig:simPsiCorrect}, but when cross-fitting is used, there still seem to be some non-vanishing bias inherent from the RF-estimation. From the standard error of the simulations corresponding to \textbf{\textit{correct-RF-CF}}, it seems as if the estimator is converging on $n^{1/2}$-rate, which, with the non-vanishing bias, results in coverage decreasing with sample size. Since only $\hat{\tau}_l$ is based on random forest, with all other nuisance estimators being based on correctly specified parametric models, this result suggests that assumption \ref{ass2} is not fulfilled for $\hat{\tau}_l$, which again can possibly be attributed to the choice of hyperparameters used in the random forest.

\begin{figure}[h!] 
    \centering
    \begin{subfigure}[c]{0.9\linewidth}
    \centering
    \includegraphics[scale = 0.65]{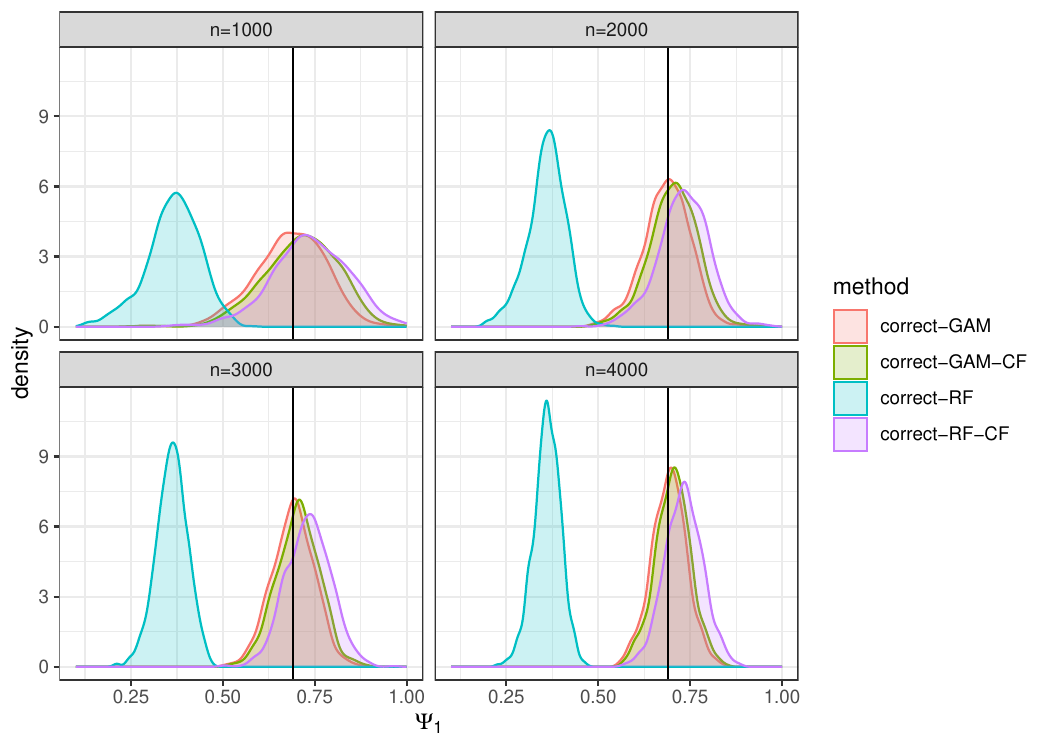}
    \caption{Correctly specified $\Lambda$, $\Lambda_c$ and $\pi$}
    \label{fig:simPsiCorrect}
    \end{subfigure}
    
    \bigskip
    
    \begin{subfigure}[c]{0.9\linewidth}
    \centering
    \includegraphics[scale = 0.65]{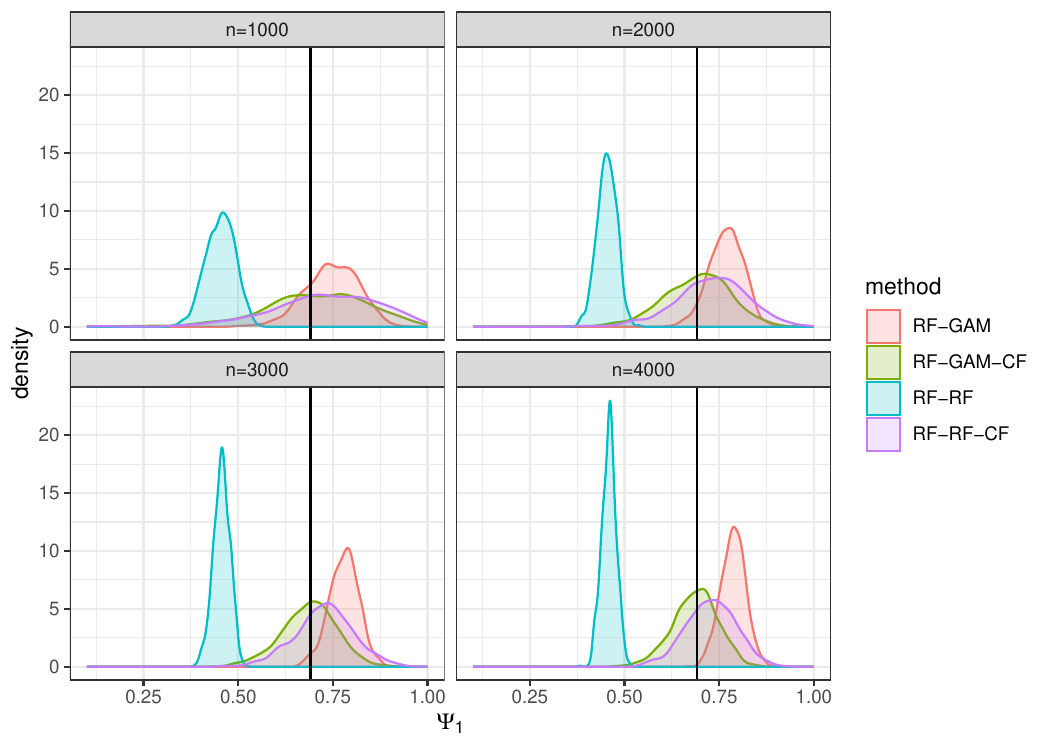}
    \caption{Flexible estimation of $\Lambda$, $\Lambda_c$ and $\pi$}
    \label{fig:simPsiRF}
    \end{subfigure}
    \caption{Sampling distribution of estimators of $\Psi_1$ in the survival function setting with varying nuisance estimators, with and without cross-fitting, across sample sizes $n=1000,2000,3000,4000$. The abbreviation of the methods should be read as follows: A-B-C, where A corresponds to the nuisance estimators $\Lambda$, $\Lambda_c$ and $\pi$, B corresponds to the nuisance estimator $\hat{E}_n$, and C corresponds to whether or not cross-fitting was used. Here, \textit{correct} corresponds to correctly specified Cox and logistic regression, RF corresponds to Random Forest, and GAM corresponds to a generalized additive model}
\end{figure}

\subsubsection{$\lambda$, $\Lambda_c$ and $\pi$ estimated by Random Forest}
In figure \ref{fig:simPsiRF} we see the sampling distribution of the estimators for $\Psi_l$, where the nuisance parameters $\lambda$, $\Lambda_c$ and $\pi$ are all estimated flexibly via \textbf{RF}. The estimators without cross-fitting are seen to be severely biased as was the case with correctly specified $\lambda$, $\Lambda_c$ and $\pi$. Table \ref{tab:resPsiRF} summarizes the results for the estimators using \textbf{RF}. Using \textbf{RF} to estimate $\hat{E}_n$ is seen to introduce some non-vanishing bias, as in Table \ref{tab:resPsiCorrect}, resulting in decreasing coverage. In the case of \textbf{GAM}-estimation for $\hat{E}_n$, cross-fitting is able to correctly debias the estimator giving approximately nominal coverage. 

\begin{table}[h!]
\centering
\begin{subtable}{0.8\textwidth}
\centering
\resizebox{0.9\textwidth}{!}{%
\begin{tabular}{rlrrrrr}
  \hline
  n & method & bias $\Psi_1$ & coverage & SD & mean SE & MSE \\ 
  \hline
    1000 & correct-GAM & -0.0161 & 0.9620 & 0.0963 & 0.0958 & 0.0095 \\ 
    2000 &  & -0.0120 & 0.9620 & 0.0640 & 0.0659 & 0.0042 \\ 
    3000 &  & -0.0068 & 0.9340 & 0.0576 & 0.0535 & 0.0034 \\ 
    4000 &  & -0.0048 & 0.9320 & 0.0478 & 0.0464 & 0.0023 \\ \hline
    1000 & correct-GAM-CF & 0.0230 & 0.9530 & 0.1030 & 0.1031 & 0.0111 \\ 
    2000 &  & 0.0107 & 0.9620 & 0.0657 & 0.0681 & 0.0044 \\ 
    3000 &  & 0.0094 & 0.9260 & 0.0588 & 0.0548 & 0.0035 \\ 
    4000 &  & 0.0082 & 0.9360 & 0.0484 & 0.0472 & 0.0024 \\ \hline
    1000 & correct-RF & -0.3349 & 0.0030 & 0.0759 & 0.0805 & 0.1179 \\ 
    2000 &  & -0.3382 & 0.0000 & 0.0526 & 0.0548 & 0.1172 \\ 
    3000 &  & -0.3364 & 0.0000 & 0.0422 & 0.0439 & 0.1149 \\ 
    4000 &  & -0.3366 & 0.0000 & 0.0369 & 0.0380 & 0.1146 \\ \hline
    1000 & correct-RF-CF & 0.0483 & 0.9200 & 0.1068 & 0.1066 & 0.0137 \\ 
    2000 &  & 0.0381 & 0.9240 & 0.0698 & 0.0716 & 0.0063 \\ 
    3000 &  & 0.0381 & 0.8840 & 0.0613 & 0.0577 & 0.0052 \\ 
    4000 &  & 0.0371 & 0.8700 & 0.0518 & 0.0498 & 0.0041 \\ 
   \hline
\end{tabular}
}
\subcaption{Correctly specified $\Lambda$, $\Lambda_c$ and $\pi$}
\label{tab:resPsiCorrect}
\end{subtable}
\bigskip

\begin{subtable}{0.8\textwidth}
\centering
\resizebox{0.9\textwidth}{!}{%
\begin{tabular}{rlrrrrr}
  \hline
  n & method & bias $\Psi_1$ & coverage & SD & mean SE & MSE \\ 
  \hline
  1000 & RF-RF & -0.2438 & 0.0000 & 0.0384 & 0.0332 & 0.0609 \\ 
    2000 &  & -0.2429 & 0.0000 & 0.0256 & 0.0235 & 0.0597 \\ 
    3000 &  & -0.2404 & 0.0000 & 0.0220 & 0.0191 & 0.0583 \\ 
    4000 &  & -0.2377 & 0.0000 & 0.0189 & 0.0165 & 0.0568 \\ \hline
    1000 & RF-RF-CF & 0.0469 & 0.9140 & 0.1575 & 0.1492 & 0.0270 \\ 
    2000 &  & 0.0353 & 0.9340 & 0.0967 & 0.0976 & 0.0106 \\ 
    3000 &  & 0.0304 & 0.9230 & 0.0780 & 0.0765 & 0.0070 \\ 
    4000 &  & 0.0318 & 0.9050 & 0.0661 & 0.0649 & 0.0054 \\  \hline
    1000 & RF-GAM & 0.0503 & 0.6360 & 0.0698 & 0.0453 & 0.0074 \\ 
    2000 &  & 0.0719 & 0.4260 & 0.0453 & 0.0324 & 0.0072 \\ 
    3000 &  & 0.0826 & 0.2310 & 0.0391 & 0.0263 & 0.0083 \\ 
    4000 &  & 0.0901 & 0.0980 & 0.0328 & 0.0227 & 0.0092 \\ \hline
    1000 & RF-GAM-CF & 0.0087 & 0.9460 & 0.1601 & 0.1521 & 0.0257 \\ 
    2000 &  & -0.0051 & 0.9500 & 0.0883 & 0.0945 & 0.0078 \\ 
    3000 &  & -0.0088 & 0.9440 & 0.0726 & 0.0739 & 0.0053 \\ 
    4000 &  & -0.0097 & 0.9560 & 0.0601 & 0.0623 & 0.0037 \\ \hline
\end{tabular}
}
\subcaption{Flexible estimation of $\Lambda$, $\Lambda_c$ and $\pi$}
\label{tab:resPsiRF}
\end{subtable}
\caption{Results of 1000 simulations of $\hat{\Psi}_1$ in the survival function setting with varying nuisance estimators, with and without cross-fitting, across sample sizes $n=1000,2000,3000,4000$. The abbreviation of the methods should be read as follows: A-B-C, where A corresponds to the nuisance estimators $\Lambda$, $\Lambda_c$ and $\pi$, B corresponds to the nuisance estimator $\hat{E}_n$, and C corresponds to whether or not cross-fitting was used. Here, \textit{correct} corresponds to correctly specified Cox and logistic regression, RF corresponds to Random Forest, and GAM corresponds to a generalized additive model. The tables shows the bias, coverage, empirical standard deviation (SD), mean estimated standard error (mean SE), and the mean squared error (MSE).}
\label{tab:resPsi}
\end{table}

\subsection{Results for $\Omega_1$}
\subsubsection{Correctly specified $\lambda$, $\Lambda_c$ and $\pi$}
Figure \ref{fig:simOmegaCorrect} presents the sampling distribution of the estimators of $\Omega_1$ with correctly specified $\Lambda$, $\Lambda_c$ and $\pi$ for different choices for estimation of $\hat{E}_n$ and $\hat{E}_n^1$, with and without cross-fitting. Compared to $\Psi_1$-estimation, there is no severe shift in sample distribution between the estimators. The estimator with \textbf{RF} and no cross-fitting is seen to have a slightly wider distribution than the others, for which there is no noticeable difference. Table \ref{tab:resOmegaCorrect} gives the results of the simulations, where the \textbf{RF} without cross-fitting is seen to generally have a slightly larger bias and MSE than the others. Remarkably, compared to the $\Psi_1$-estimation, far fewer observations are needed for reliable estimation of $\Omega_1$. One thing to note is that Corollary \ref{cor:omega} is only guaranteed to hold for the cross-fitted estimator, but since the estimator is given as the ratio of two other estimators, it can happen that the bias introduced by employing \textbf{RF} without cross-fitting roughly cancels in the ratio, which may explain why the \textbf{RF} without cross-fitting is seen to perform reliably with approximately nominal coverage. Indeed this is the case for the simulation study conducted here, as seen in Figure \ref{fig:simGammaCorrect} and \ref{fig:simChiCorrect} in the Appendix, where the sample distribution of the estimators of $\Gamma_1$ and $\chi_1$ are shown. Hence, we cannot recommend \textbf{RF} without cross-fitting, as the cancellation of biases in the ratio of $\Gamma_1$ and $\chi_1$ is unlikely to occur generally. 

\begin{figure}[h!] 
    \centering
    \begin{subfigure}[c]{0.9\linewidth}
    \centering
    \includegraphics[scale = 0.65]{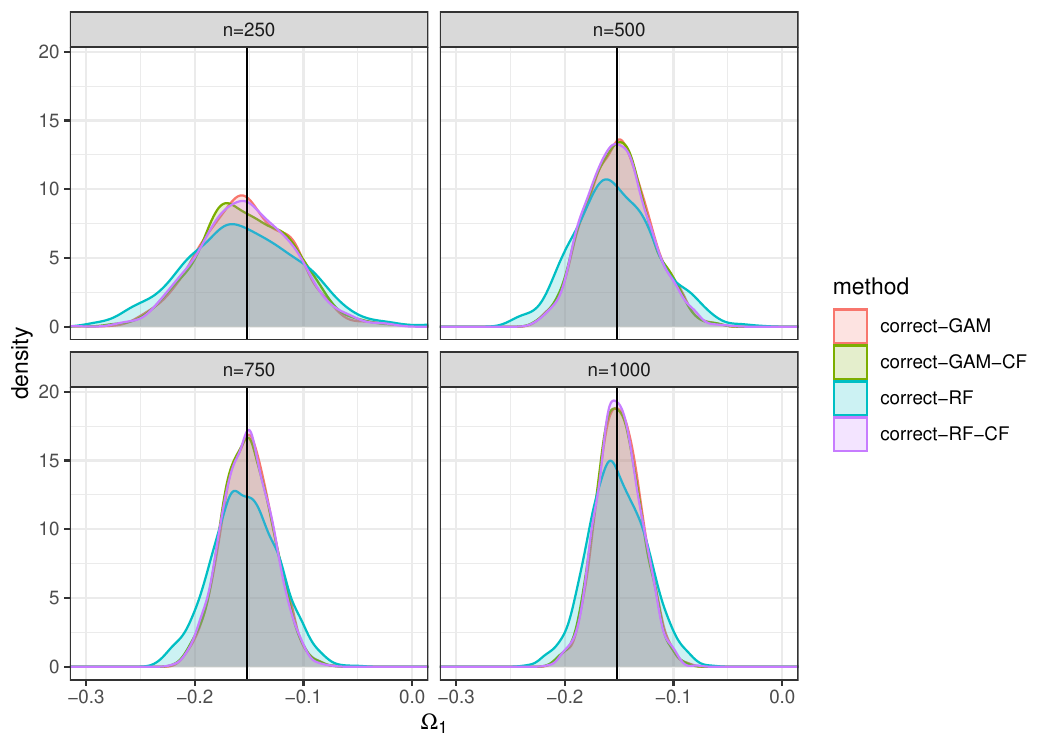}
    \caption{Correctly specified $\Lambda$, $\Lambda_c$ and $\pi$}
    \label{fig:simOmegaCorrect}
    \end{subfigure}

    \begin{subfigure}[c]{0.9\linewidth}
    \centering
    \includegraphics[scale = 0.65]{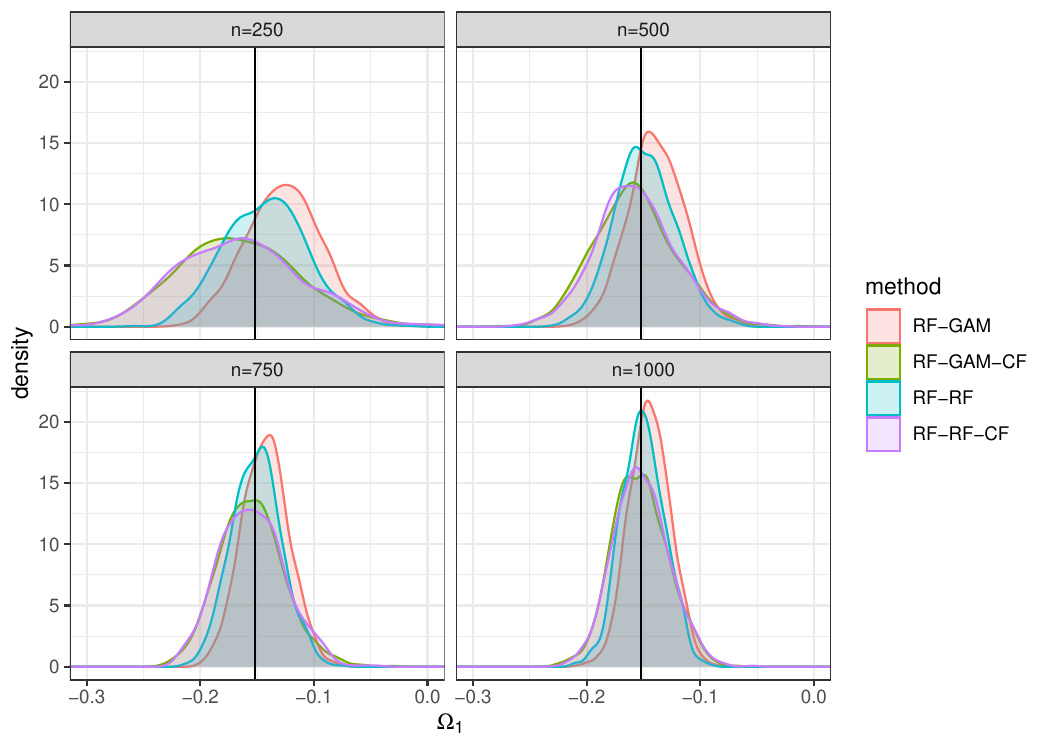}
    \caption{Flexible estimation of $\Lambda$, $\Lambda_c$ and $\pi$}
    \label{fig:simOmegaRF}
    \end{subfigure}
    \caption{Sampling distribution of estimators of $\Omega_1$ in the survival function setting with varying nuisance estimators, with and without cross-fitting, across sample sizes $n=250,500,750,1000$. The abbreviations of the methods are read as follows: A-B-C, where A corresponds to the nuisance estimators $\Lambda$, $\Lambda_c$ and $\pi$, B corresponds to the nuisance estimators $\hat{E}_n$ and $\hat{E}^j_n$, and C corresponds to whether or not cross-fitting was used. Here, \textit{correct} corresponds to correctly specified Cox and logistic regression, RF corresponds to Random Forest, and GAM corresponds to a generalized additive model.}
\end{figure}

\subsubsection{$\lambda$, $\Lambda_c$ and $\pi$ estimated by Random Forest}
Figure \ref{fig:simOmegaRF} shows the sampling distribution of the estimators of $\Omega_1$ when RF is used for estimation of $\lambda$, $\Lambda_c$ and $\pi$. Here, the difference between the cross-fitted and non-cross-fitted estimators are more noticeable. Interestingly, using cross-fitting seem to produce similar distributions, regardless of whether \textbf{RF} og \textbf{GAM} was used for estimation of $\hat{E}_n$ and $\hat{E}_n^1$. Table \ref{tab:resOmegaRF} presents the results of the simulation study. Generally, the bias seem to vanish with the sample size (again, in the case of \textbf{RF-RF}, this might be a coincidence), but the coverage for the non-cross-fitted estimators are far off, whereas cross-fitting seem to provide approximately nominal coverage, even in relatively small samples.

\begin{table}[h!]
\centering
\begin{subtable}{0.8\textwidth}
\centering
\resizebox{0.9\textwidth}{!}{%
\begin{tabular}{rlrrrrr}
  \hline
  n & method & bias $\Omega_1$ & coverage & SD & mean SE & MSE \\ 
  \hline
    250 & correct-GAM & - 0.0011 & 0.934 & 0.0422 & 0.0408 & 0.0018 \\ 
  500 &  & 0.0011 & 0.945 & 0.0291 & 0.0283 & 0.0009 \\ 
  750 &  & 0.0009 & 0.952 & 0.0230 & 0.0231 & 0.0005 \\ 
  1000 &  & 0.0021 & 0.947 & 0.0203 & 0.0201 & 0.0004 \\ \hline
    250 & correct-GAM-CF & -0.0004 & 0.932 & 0.0438 & 0.0419 & 0.0019 \\ 
  500 &  & 0.0006 & 0.948 & 0.0291 & 0.0287 & 0.0008 \\ 
  750 &  & 0.0001 & 0.950 & 0.0231 & 0.0233 & 0.0005 \\ 
  1000 &  & 0.0012 & 0.948 & 0.0204 & 0.0202 & 0.0004 \\ \hline
    250 & correct-RF & -0.0041 & 0.946 & 0.0532 & 0.0544 & 0.0028 \\ 
  500 &  & -0.0015 & 0.953 & 0.0378 & 0.0387 & 0.0014 \\ 
  750 &  & -0.0020 & 0.955 & 0.0299 & 0.0317 & 0.0009 \\ 
  1000 &  & 0.0003 & 0.956 & 0.0265 & 0.0277 & 0.0007 \\ \hline
    250 & correct-RF-CF & -0.0012 & 0.936 & 0.0428 & 0.0414 & 0.0018 \\ 
  500 &  & 0.0006 & 0.945 & 0.0291 & 0.0283 & 0.0008 \\ 
  750 &  & -0.0001 & 0.946 & 0.0229 & 0.0229 & 0.0005 \\ 
  1000 &  & 0.0015 & 0.947 & 0.0202 & 0.0199 & 0.0004 \\ 
   \hline
\end{tabular}
}
\subcaption{Correctly specified $\Lambda$, $\Lambda_c$ and $\pi$}
\label{tab:resOmegaCorrect}
\end{subtable}
\bigskip

\begin{subtable}{0.8\textwidth}
\centering
\resizebox{0.9\textwidth}{!}{%
\begin{tabular}{rlrrrrr}
  \hline
  n & method & bias $\Omega_1$ & coverage & SD & mean SE & MSE \\ 
  \hline
  250 & RF-RF & 0.0090 & 0.915 & 0.0363 & 0.0324 & 0.0014 \\ 
  500 &  & 0.0031 & 0.904 & 0.0265 & 0.0224 & 0.0007 \\ 
  750 &  & 0.0014 & 0.891 & 0.0217 & 0.0183 & 0.0005 \\ 
  1000 &  & 0.0013 & 0.897 & 0.0190 & 0.0158 & 0.0004 \\ \hline
    250 & RF-RF-CF & -0.0127 & 0.935 & 0.0548 & 0.0539 & 0.0032 \\ 
  500 &  & -0.0047 & 0.946 & 0.0353 & 0.0355 & 0.0013 \\ 
  750 &  & -0.0027 & 0.948 & 0.0287 & 0.0285 & 0.0008 \\ 
  1000 &  & -0.0001 & 0.952 & 0.0242 & 0.0246 & 0.0006 \\  \hline
    250 & RF-GAM & 0.0244 & 0.740 & 0.0331 & 0.0245 & 0.0017 \\ 
  500 &  & 0.0135 & 0.747 & 0.0249 & 0.0168 & 0.0008 \\ 
  750 &  & 0.0092 & 0.775 & 0.0205 & 0.0137 & 0.0005 \\ 
  1000 &  & 0.0074 & 0.759 & 0.0181 & 0.0118 & 0.0004 \\ \hline
    250 & RF-GAM-CF & -0.0129 & 0.941 & 0.0558 & 0.0545 & 0.0033 \\ 
  500 &  & -0.0060 & 0.949 & 0.0358 & 0.0365 & 0.0013 \\ 
  750 &  & -0.0032 & 0.952 & 0.0285 & 0.0293 & 0.0008 \\ 
  1000 &  & -0.0011 & 0.961 & 0.0241 & 0.0253 & 0.0006 \\ \hline
\end{tabular}
}
\subcaption{Flexible estimation of $\Lambda$, $\Lambda_c$ and $\pi$}
\label{tab:resOmegaRF}
\end{subtable}
\caption{Results of 1000 simulations of $\hat{\Omega}_1$ in the survival function setting with varying nuisance estimators, with and without cross-fitting, across sample sizes $n=250,500,750,1000$. The abbreviations of the methods are read as follows: A-B-C, where A corresponds to the nuisance estimators $\Lambda$, $\Lambda_c$ and $\pi$, B corresponds to the nuisance estimators $\hat{E}_n$ and $\hat{E}^j_n$, and C corresponds to whether or not cross-fitting was used. Here, \textit{correct} corresponds to correctly specified Cox and logistic regression, RF corresponds to Random Forest, and GAM corresponds to a generalized additive model. The tables shows the bias, coverage, empirical standard deviation (SD), mean estimated standard error (mean SE), and the mean squared error (MSE).}
\label{tab:resOmega}
\end{table}

\section{Application to real data}
 To showcase the two approaches and their differences, we analyze two data different data sets using the estimands $\Psi_l$ and $\Omega_j$ to assess potential treatment effect heterogeneity. In the first example, the two estimands are used analogously to estimate the variable importance of each individual covariate, whereas in the second example $\Psi_l$ is used to analyze the importance of a group of variables, while $\Omega_j$ is estimated for individual covariates.

\subsection{Application to HIV data}
We apply the methods described in the previous sections to the AIDS Clinical Trial Group Study 175 (\cite{hammer}). The data can be found in the \verb|R|-package \verb|ACTG175| and consist of 2139 HIV patients who were randomized to one of four treatments: (1) zidovudine (ZDV)(n=532), (2) zidovudine + didanosine (ZDV+ddI)(n=522), (3) zidovudine + zalcitabine (ZDV+ZAL)(n=524), and (4) didanosine(n=561). Patients were followed from treatment initiation until end of follow-up or a composite event consisting of a decline in CD4 cell count greater than $50\%$, or disease progression to AIDS, or death. We define our event of interest as the composite event. Patients were censored at end-of-follow-up. In line with \textcite{cui}, we define the treatment effect (comparing two treatments) to be given by the RMST at 1000 days after treatment initiation, and we consider 12 baseline covariates for which we will analyse the possible treatment effect heterogeneity explained by each of them. The covariates consist of 5 continuous variables, age, CD4 cell count, CD8 cell count, weight (kg), Karnofsky score, and 7 binary variables, gender, race, hemophilia, homosexual activity, antiretroviral history, symptomatic status, intravenous drug use history.

The aim of the study was to compare monotherapy (ZDV or ddI) with combination therapy (ZDV+ddI or ZDV + ZAL),
and as an illustration we consider  the 
comparison ZDV vs ZDV+ZAL.
We applied the cross-fitted TE-VIM and the best partially linear projection estimators, with $K=10$, described in Section \ref{sec:estimation}, with all nuisance parameters estimated by Random Forests as implemented in the \verb|R|-package \verb|RandomForestSRC|. For the TE-VIM, we used the logit-transformed $\hat{\zeta}_l^{CF}$ (see Appendix \ref{App-logit}) with an expit-back-transformation to obtain $\Psi_l$-estimates and confidence intervals that respect the boundary $\Psi_l \in [0,1)$. 

In Table \ref{tab:zdv+zal} we see the results based on the comparison of ZDV vs ZDV+ZAl. Here the TE-VIM estimates are ranging from 0.02 to 0.992 with cd8 having the largest estimate, but all with a confidence interval ranging from 0 to 1. The results based on the best partially linear projection give p-values in the range 0.007 to 0.953 with 3 significant p-values for cd8, karnof, and cd4, respectively. Both measures ranks CD8 cell count as the most "important" in terms of explaining treatment heterogeneity, but with the TE-VIM having a confidence interval of [0,1]. The results suggest that CD8 cell count, Karnofsky score and CD4 cell count may be important in explaining the treatment effect heterogeneity of ZDV vs ZDV+ZAl on RMST at $t=1000$ days after treatment initiation. In order to interpret the direction of the estimates, we use CD8 as an example. The estimate of $\Omega_j$ is $0.093$ with a p-value of $0.007$ and by looking at the definition of $\Omega_j$ as the least squares projection onto the partially linear model, this corresponds to an increase in treatment effect of 0.093 days in terms of RMST for every increase in CD8 of 1. As noted, the partially linear model might fail to hold, but the interpretation shows that the treatment effect is increased for larger values of CD8.

\begin{table}[h!]
    \begin{subtable}[c]{0.5\textwidth}
    \label{tab:psi-zdv+zal}
    \centering
    \begin{tabular}{lrrr}
  \hline
  covariate & $\Psi_j$ & lower & upper \\ 
  \hline
 cd8 & 0.992 & 0.000 & 1.000 \\ 
   karnof & 0.368 & 0.000 & 1.000 \\ 
   symptom & 0.348 & 0.000 & 1.000 \\ 
   gender & 0.197 & 0.001 & 0.982 \\ 
   str2 & 0.117 & 0.000 & 1.000 \\ 
   age & 0.088 & 0.000 & 1.000 \\ 
   race & 0.052 & 0.000 & 1.000 \\ 
   homo & 0.037 & 0.000 & 1.000 \\ 
   hemo & 0.029 & 0.000 & 1.000 \\ 
   wtkg & 0.026 & 0.000 & 1.000 \\ 
   cd4 & 0.023 & 0.000 & 1.000 \\ 
   drugs & 0.020 & 0.000 & 1.000 \\ 
   \hline
\end{tabular}
\subcaption{Heterogeneity explained by $\Psi_j$}
    \end{subtable}
    \begin{subtable}[c]{0.5\textwidth}
    \label{tab:omega-zdv+zal}
    \centering
    \begin{tabular}{rlrrr}
  \hline
  covariate & $\Omega_j$ & SE & p-value \\ 
  \hline
 cd80 & 0.093 & 0.035 & 0.007 \\ 
   karnof & 6.446 & 3.024 & 0.033 \\ 
   cd40 & -0.288 & 0.136 & 0.034 \\ 
   symptom & -55.853 & 47.467 & 0.239 \\ 
   drugs & 48.365 & 50.376 & 0.337 \\ 
   wtkg & 0.971 & 1.092 & 0.374 \\ 
   age & -0.985 & 1.702 & 0.563 \\ 
   race & -18.520 & 37.165 & 0.618 \\ 
   gender & 24.795 & 54.105 & 0.647 \\ 
   homo & -19.473 & 52.341 & 0.710 \\ 
   hemo & 5.463 & 77.456 & 0.944 \\ 
   str2 & -1.719 & 29.241 & 0.953 \\ 
   \hline
\end{tabular}
\subcaption{Heterogeneity explained by $\Omega_j$}

\end{subtable}
\caption{Heterogeneity in the effect of zidovudine + zalcitabine (ZDV+zal) vs zidovudine (ZDV) on RMST. Estimation of variable importance on the treatment effect of ZDV+zal vs ZDV on RMST. The data is from the study \textcite{hammer} and the outcome is time to an event consisting of a decline in CD4 cell count greater than $50\%$, disease progression to AIDS, or death. The treatment effect is defined as the difference in RMST at 1000 days after treatment initiation between ZDV+zal and ZDV. In table (a), the variable importance is estimated by the $expit$-transformation of $\hat{\zeta}_l^{CF}$, for $l$ ranging over the single covariates, with corresponding confidence intervals. In table (b), the variable importance is estimated by $\hat{\Omega}_j^{CF}$, for $j$ ranging over the single covariates.}
\label{tab:zdv+zal}
\end{table}

In comparing the results related to $\Psi_l$ and $\Omega_j$, respectively, we see the difference in sample sizes needed for providing meaningful estimates between the two measures, as indicated by the simulation study in Section \ref{sec:sim}. With confidence intervals of [0,1], the estimates of $\Psi_l$ do not give any inside into the potential heterogeneity in the effect of the treatments considered here,  whereas the estimates of $\Omega_j$ were able to find significant treatment effect modification for some of the covariates.

\subsection{LEADER data}
We consider data from the LEADER study, see \textcite{marso}, that investigates the effect of liraglutide in combination with standard care in patients with type 2 diabetes on cardiovascular outcomes. We focus on all-cause mortality as  event of interest to avoid competing risks issues. In the study, 9340 patients were randomized to either liraglutide $(n = 4668)$ or placebo $(n = 4672)$ with a maximum follow-up of 60 months. Patients were followed until an event with the outcome of interest or censoring at end of follow-up or exit from the study. For the sake of illustration, we consider a complete-case analysis leaving us with 4586 patients receiving liraglutide and 4576 patients receiving placebo at baseline. 

 \textcite{marso}  used Cox regression to analyse the effect of liraglutide on  all-cause mortality. 
 We focus on the difference in survival probability
at $t=3$ years of follow-up.
Since the study is randomized, all CATE-functions (for different covariate sets $X$) are identified from the observed data, and we can simply choose an appropriate set among the available covariates without violating the assumption of no unmeasured confounding. Thus, we choose the following 12 covariates for our analysis (all measured at baseline): sex, age, SYSBPBL (systolic blood pressure), DIABPBL (diastolic blood pressure), CHOL1BL (total cholesterol), HDLBL (HDL cholesterol), LDLBL (LDL cholesterol), TRIG1BL (Triglycerides), EGFREPB (eGFR), CREATBL (serum Creatinine), BMIBL (BMI), DIABDUR (diabetes duration). 

For the estimation of $\Omega_j$, we exclude DIABPL, HDL1BL, LDL1BL, and TRIG1Bl, to avoid collinearity. From the definition of $\Omega_j$ as the least squares projection onto the partially linear model, it is clear that collinearity in a group of covariates can prevent the detection of heterogeneity, as most of the effect could be captured in $w(X_{-j})$. Hence we choose CHOL1BL and SYSBPBL to represent cholesterol and blood pressure, respectively. Furthermore, we dichotomize CHOL1BL into $\text{CHOL}\_ \text{GRP}$, which is 1 if CHOL1BL is above 5 and 0 if CHOL1BL is less than or equal to 5, representing a high versus normal cholesterol.  

For estimation of $\Psi_l$, we choose two groups and compare their importance in explaining the potential heterogeneity in $\tau$. We let Cholesterol be the set comprised of HDL1BL, LDL1BL, CHOL1BL, TRIG1BL and let Bloodpressure be the set comprised of SYSBPBL and DIABPBL. Furthermore, we use the estimator $\zeta_l^{CF}$ to estimate the logit-transformation of $\Psi_l$ (see Appendix \ref{App-logit}) and obtain estimates and confidence intervals of $\Psi_l$ in the range $[0,1)$. 

We use random forests for estimation of the all nuisance parameters for both $\Omega_j$ and $\Psi_l$ with default hyperparameter settings as given in the \verb|R|-package \verb|randomForestSRC| and $K=5$ folds for the cross-fitting procedure. 

Results for $\Psi_l$ are presented in Table \ref{tab:leader-psil-surv}. 
The estimates and corresponding 95\% confidence-intervals for 
Cholesterol and Bloodpressure are given by
 0.99 [0, 1] and 
0.33 [0.02, 0.91], respectively. 
The point estimate associated with Cholesterol indicates that it is  more important  than Bloodpressure in explaining the heterogeneity. However,
the corresponding confidence-intervals are spanning the entire range in both cases and therefore little can be concluded from the estimates. Additionally, we estimated $\Theta_d$ by an exponential transformation of a cross-fitted one-step estimator of the log-transformed $\Theta_d$, analogously to the logit-transformation of $\Psi_l$, to obtain a $\Theta_d$-estimate above 0. We got an estimate of $5\cdot 10^{-5}$ $[1.5 \cdot 10^{-6}, 1.7 \cdot 10^{-3}]$. Hence, caution is warranted in the interpretation of the estimates of $\Psi_l$ (aside from the range of the confidence intervals), as the estimate of $\Theta_d$ suggests no heterogeneity in the treatment effect $\tau$.

Results for $\Omega_j$ are presented in Table \ref{tab:leader-omega-res}, where it is seen that all estimates are insignificant except for the  estimate 0.025 associated with  \text{CHOL\_GRP}, p=0.042. By the least squares definition of $\Omega_j$, this means an increase in the treatment effect of $2.5\%$ for patients with high cholesterol compared to patients with normal cholesterol. In Figure \ref{fig:leader-kaplan-meier}, we see the Kaplan-Meier curves for the high and normal cholesterol groups. The figure illustrates the potential heterogeneity in relation to cholesterol found by $\Omega_j$.

\begin{table}[h!]
\centering
\begin{tabular}{llrr}
  \hline
group & $\Psi_l$ & lower & upper \\ 
  \hline
Cholesterole & 0.99 & 0.00 & 1.00 \\ 
  BloodPressure & 0.326 & 0.02 & 0.91 \\ 
   \hline
\end{tabular}
\caption{Results from the LEADER data on treatment effect heterogeneity in liraglutide vs. placebo via $\Psi_l$. Estimation of $\Psi_l$ was based on $\tau$ defined in terms of survival probabilities. 
The group Cholesterol is comprised of the covariates HDL1BL, LDL1BL,
CHOL1BL, TRIG1BL, and Bloodpressure is comprised of SYSBPBL and DIABPBL. The outcome is time to all-cause death and the time horizon for the treatment effect $\tau$ is set at $t=3$ years. All nuisance parameters were estimated with random forests and $K=5$ folds were used for cross-fitting.}
\label{tab:leader-psil-surv}
\end{table}

\begin{table}[h!]
\centering
\begin{tabular}{lrrr}
  \hline
covariate & $\Omega_j$ & SE & p-value \\ 
  \hline
SEX & -0.0246 & 0.0306 & 0.4210 \\ 
  AGE & 0.0023 & 0.0015 & 0.1362 \\ 
  SYSBPBL & 0.0004 & 0.0003 & 0.2076 \\ 
  CHOL\_GRP & 0.0246 & 0.0121 & 0.0421 \\ 
  EGFREPB & -0.0004 & 0.0014 & 0.7856 \\ 
  CREATBL & -0.0008 & 0.0009 & 0.4263 \\ 
  BMIBL & 0.0001 & 0.0009 & 0.9152 \\ 
  DIABDUR & -0.0013 & 0.0007 & 0.0821 \\ 
   \hline
\end{tabular}
\caption{Results from the LEADER data on treatment effect heterogeneity in liraglutide vs. placebo via $\Omega_j$. Estimation of $\Omega_j$ was based on $\tau$ defined in terms of survival probabilities. 
The outcome is time to all-cause death and the time horizon for the treatment effect $\tau$ is set at $t=3$ years. All nuisance parameters were estimated with random forests and $K=5$ folds were used for cross-fitting.} 
\label{tab:leader-omega-res}
\end{table}

\begin{figure}[h!]
    \centering
    \begin{subfigure}{0.5\linewidth}
        \centering
        \includegraphics[scale = 0.6]{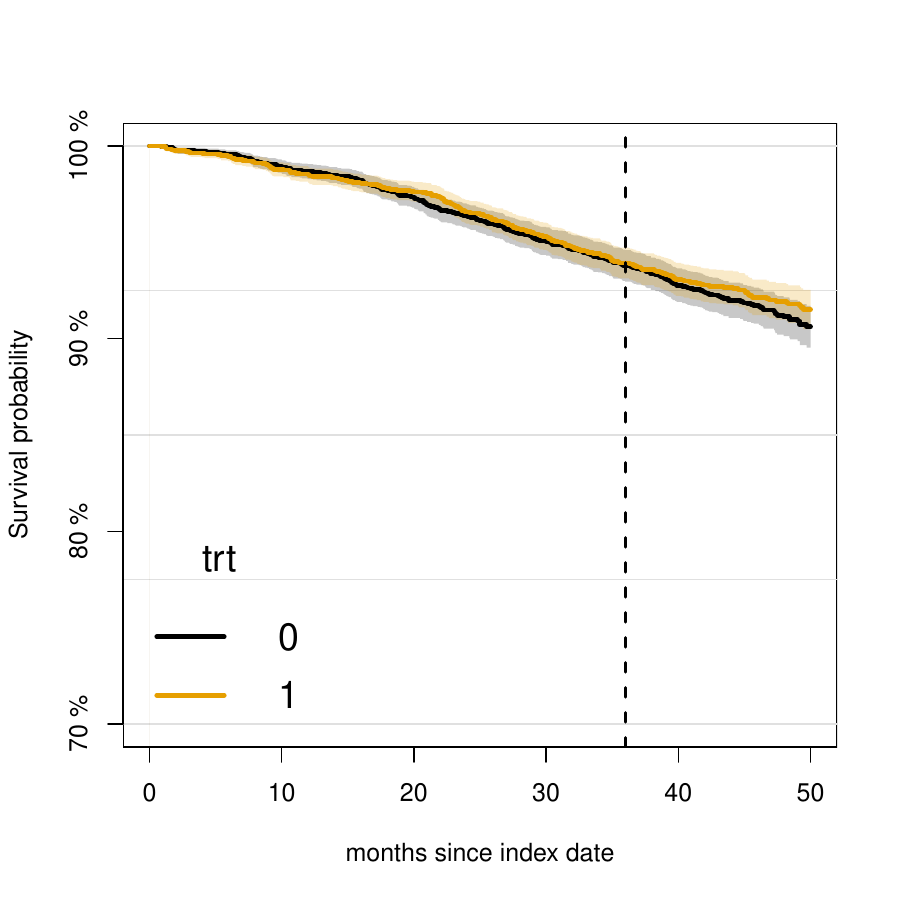}
        \caption{Cholestorol less than or equal to 5}
        \label{fig:leader-cholestorolLessThan5}
    \end{subfigure}
    \begin{subfigure}{0.5\linewidth}
        \centering
        \includegraphics[scale = 0.6]{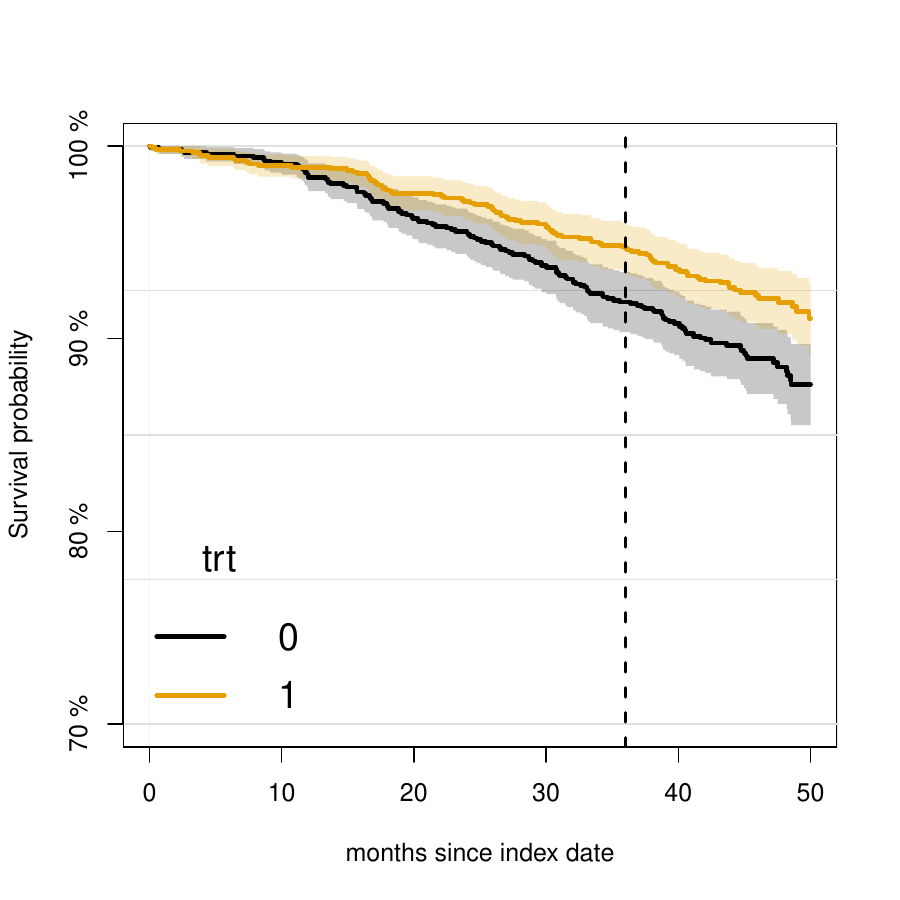}
        \caption{Cholestorol greater than 5}
        \label{fig:leader-cholestorolGreaterThan5}
    \end{subfigure}
    \caption{Kaplan-Meier plots for LEADER data. The yellow curve, corresponding to trt=1, shows the estimated survival probability for patients receiving liraglutide. The grey curve, corresponding to trt=0, shows the survival probability for patients receiving placebo. The dotted line indicates $t=3$ years after baseline. (a) shows the estimated curves for patients with  normal cholesterol at baseline, and (b) shows the curves for patients with  high cholesterol at baseline.}
    \label{fig:leader-kaplan-meier}
\end{figure}

As with the HIV data, the estimates of $\Psi_l$ failed to provide insight into the potential treatment effect heterogeneity with confidence intervals spanning the entire range $[0,1]$, while the estimates of $\Omega_j$ indicated at potential heterogeneity through cholesterol. 

\section{Closing remarks}
In this paper, we have extended the treatment effect variable importance measures introduced by \textcite{hines} to a time-to-event setting allowing for censored data. We have constructed estimators for the TE-VIMs $\Theta_l$ and $\Psi_l$ using two different CATE functions and given assumptions under which they are seen to be asymptotically normal and locally efficient. The assumptions require that the nuisance estimators $\hat{\tau}$ and $\hat{\tau}_l$ are both consistent at $n^{-1/4}$-rate, allowing for the use of machine-learning to estimate the nuisance parameters.  The simulation study showed that the estimators without cross-fitting were heavily biased when using data adaptive nuisance estimators, such as random forest, but that the cross-fitting was mostly able to correct the bias introduced by the flexible nuisance estimation. Importantly, it seems that the main challenge lies in choosing $\hat{E}_n$ appropriately, since using random forest (with default hyperparameters) was seen to introduce a non-vanishing bias, whereas using GAM for $\hat{E}_n$-estimation gave correct coverage even when other nuisance parameters were estimated with random forest. 
One possible avenue to leverage the choice of RF-hyperparameters could be to replace the current cross-fitted one-step estimators with targeted-maximum-likelihood (TMLE). \textcite{li} constructed a TMLE for the TE-VIMs of \textcite{hines}, and though the remainder term still calls for initial $\hat{\tau}_l$ estimators that are consistent at $n^{-1/4}$ rate, the estimators were seen to have better finite sample performance compared to the one-step estimator. Thus, one may pursue a TMLE based on the EIF's derived in this paper with the same asymptotic properties as $\hat{\Theta}_l^{CF}$ under assumption \ref{assA}, to possibly achieve better finite sample performance. We leave this for future work. 

Furthermore, we have proposed a new variable importance measure based on the ideas from \textcite{stijn} as a best partially linear projection of the CATE-function. The estimand has the interpretation of the real parameter in a partially linear model of the CATE function, but it continues to serve as a measure of heterogeneity when the model fails to hold. One consequence, though, is that it could happen that $\Omega_j=0$ even when $X_j$ explains some of the treatment effect, as seen by plugging $\beta=0$ into
$$
\tau(x) = \beta x_j + w(x_{-j}) + R(x_j, x_{-j}).
$$
In contrast to the estimators for $\Psi_l$, the estimators of $\Omega_j$ was seen to perform well in relatively small sample sizes compared to the sample sizes needed for reliable estimation of $\Psi_l$, even when using Random Forest for all nuisance parameter estimation. This was also evident in the practical examples, where the estimates of $\Psi_l$ all had corresponding confidence intervals form 0 to 1, essentially rendering them useless as measures of variable importance, but where the p-values associated with the hypothesis $H:\Gamma_j=0$ provided some significant findings. 
\clearpage
\nocite{*}
\printbibliography

@article{hines,
  title={Variable importance measures for heterogeneous causal effects},
  author={Hines, Oliver and Diaz-Ordaz, Karla and Vansteelandt, Stijn},
  journal={arXiv preprint arXiv:2204.06030},
  year={2022}
}

@article{kennedy,
  title={Towards optimal doubly robust estimation of heterogeneous causal effects},
  author={Kennedy, Edward H},
  journal={arXiv preprint arXiv:2004.14497},
  year={2022}
}

@article{kennedydouble,
  title={Semiparametric doubly robust targeted double machine learning: a review},
  author={Kennedy, Edward H},
  journal={arXiv preprint arXiv:2203.06469},
  year={2022}
}

@article{hinesdemystifying,
  title={Demystifying statistical learning based on efficient influence functions},
  author={Hines, Oliver and Dukes, Oliver and Diaz-Ordaz, Karla and Vansteelandt, Stijn},
  journal={The American Statistician},
  volume={76},
  number={3},
  pages={292--304},
  year={2022},
  publisher={Taylor \& Francis}
}

@article{heleneFrank,
  title={Estimation of time-specific intervention effects on continuously distributed time-to-event outcomes by targeted maximum likelihood estimation},
  author={Rytgaard, Helene CW and Eriksson, Frank and van der Laan, Mark J},
  journal={Biometrics},
  year={2023},
  publisher={Wiley Online Library}
}

@article{westling,
  title={Inference for treatment-specific survival curves using machine learning},
  author={Westling, Ted and Luedtke, Alex and Gilbert, Peter B and Carone, Marco},
  journal={Journal of the American Statistical Association},
  number={just-accepted},
  pages={1--26},
  year={2023},
  publisher={Taylor \& Francis}
}

@book{vaart,
  title={Asymptotic statistics},
  author={van der Vaart, Aad W},
  volume={3},
  year={2000},
  publisher={Cambridge university press}
}

@book{robinslaan,
  title={Unified methods for censored longitudinal data and causality},
  author={van der Laan, Mark J and Robins, James M},
  year={2003},
  publisher={Springer}
}

@article{martSten,
  title={Estimation of separable direct and indirect effects in continuous time},
  author={Martinussen, Torben and Stensrud, Mats Julius},
  journal={Biometrics},
  volume={79},
  number={1},
  pages={127--139},
  year={2023},
  publisher={Wiley Online Library}
}

@article{gill,
  title={A survey of product-integration with a view toward application in survival analysis},
  author={Gill, Richard D and Johansen, Soren},
  journal={The annals of statistics},
  volume={18},
  number={4},
  pages={1501--1555},
  year={1990},
  publisher={Institute of Mathematical Statistics}
}

@article{cher,
  title={Double/debiased machine learning for treatment and structural parameters},
  author={Chernozhukov, Victor and Chetverikov, Denis and Demirer, Mert and Duflo, Esther and Hansen, Christian and Newey, Whitney and Robins, James},
  year={2018},
  publisher={Oxford University Press Oxford, UK}
}

@article{rfsrc,
  title={Package ‘randomForestSRC’},
  author={Ishwaran, Hemant and Kogalur, Udaya B and Kogalur, Maintainer Udaya B},
  journal={breast},
  volume={6},
  number={1},
  pages={854},
  year={2023}
}

@article{wagerAthey,
  title={Estimation and inference of heterogeneous treatment effects using random forests},
  author={Wager, Stefan and Athey, Susan},
  journal={Journal of the American Statistical Association},
  volume={113},
  number={523},
  pages={1228--1242},
  year={2018},
  publisher={Taylor \& Francis}
}

@article{cui,
  title={Estimating heterogeneous treatment effects with right-censored data via causal survival forests},
  author={Cui, Yifan and Kosorok, Michael R and Sverdrup, Erik and Wager, Stefan and Zhu, Ruoqing},
  journal={Journal of the Royal Statistical Society Series B: Statistical Methodology},
  volume={85},
  number={2},
  pages={179--211},
  year={2023},
  publisher={Oxford University Press US}
}

@article{hu,
  title={Estimating heterogeneous survival treatment effect in observational data using machine learning},
  author={Hu, Liangyuan and Ji, Jiayi and Li, Fan},
  journal={Statistics in medicine},
  volume={40},
  number={21},
  pages={4691--4713},
  year={2021},
  publisher={Wiley Online Library}
}

@incollection{xu,
  title={Treatment heterogeneity with survival outcomes},
  author={Xu, Yizhe and Ignatiadis, Nikolaos and Sverdrup, Erik and Fleming, Scott and Wager, Stefan and Shah, Nigam},
  booktitle={Handbook of Matching and Weighting Adjustments for Causal Inference},
  pages={445--482},
  year={2023},
  publisher={Chapman and Hall/CRC}
}

@article{levy,
  title={A fundamental measure of treatment effect heterogeneity},
  author={Levy, Jonathan and van der Laan, Mark and Hubbard, Alan and Pirracchio, Romain},
  journal={Journal of Causal Inference},
  volume={9},
  number={1},
  pages={83--108},
  year={2021},
  publisher={De Gruyter}
}

@article{wei,
  title={Efficient targeted learning of heterogeneous treatment effects for multiple subgroups},
  author={Wei, Waverly and Petersen, Maya and van der Laan, Mark J and Zheng, Zeyu and Wu, Chong and Wang, Jingshen},
  journal={Biometrics},
  volume={79},
  number={3},
  pages={1934--1946},
  year={2023},
  publisher={Wiley Online Library}
}

@article{boileau,
  title={A nonparametric framework for treatment effect modifier discovery in high dimensions},
  author={Boileau, Philippe and Leng, Ning and Hejazi, Nima S and van der Laan, Mark and Dudoit, Sandrine},
  journal={arXiv preprint arXiv:2304.05323},
  year={2023}
}

@article{ishwaran,
  title={Random survival forests},
  author={Ishwaran, Hemant and Kogalur, Udaya B and Blackstone, Eugene H and Lauer, Michael S},
  year={2008}
}

@article{li,
  title={Targeted Learning on Variable Importance Measure for Heterogeneous Treatment Effect},
  author={Li, Haodong and Hubbard, Alan and van der Laan, Mark},
  journal={arXiv preprint arXiv:2309.13324},
  year={2023}
}

@article{stijn,
  title={Assumption-lean inference for generalised linear model parameters},
  author={Vansteelandt, Stijn and Dukes, Oliver},
  journal={Journal of the Royal Statistical Society Series B: Statistical Methodology},
  volume={84},
  number={3},
  pages={657--685},
  year={2022},
  publisher={Oxford University Press}
}

@article{hammer,
  title={A trial comparing nucleoside monotherapy with combination therapy in HIV-infected adults with CD4 cell counts from 200 to 500 per cubic millimeter},
  author={Hammer, Scott M and Katzenstein, David A and Hughes, Michael D and Gundacker, Holly and Schooley, Robert T and Haubrich, Richard H and Henry, W Keith and Lederman, Michael M and Phair, John P and Niu, Manette and others},
  journal={New England Journal of Medicine},
  volume={335},
  number={15},
  pages={1081--1090},
  year={1996},
  publisher={Mass Medical Soc}
}

@article{vansteelandt,
  title={Assumption-lean inference for generalised linear model parameters},
  author={Vansteelandt, Stijn and Dukes, Oliver},
  journal={Journal of the Royal Statistical Society Series B: Statistical Methodology},
  volume={84},
  number={3},
  pages={657--685},
  year={2022},
  publisher={Oxford University Press}
}

@article{semenova,
  title={Debiased machine learning of conditional average treatment effects and other causal functions},
  author={Semenova, Vira and Chernozhukov, Victor},
  journal={The Econometrics Journal},
  volume={24},
  number={2},
  pages={264--289},
  year={2021},
  publisher={Oxford University Press}
}

@article{vanLaan,
  title={Statistical inference for variable importance},
  author={Van der Laan, Mark J},
  journal={The International Journal of Biostatistics},
  volume={2},
  number={1},
  year={2006},
  publisher={De Gruyter}
}

@book{tsiatis,
  title={Semiparametric theory and missing data},
  author={Tsiatis, Anastasios A},
  volume={4},
  year={2006},
  publisher={Springer}
}

@article{kennedy2020,
  title={Sharp instruments for classifying compliers and generalizing causal effects},
  author={Kennedy, Edward H and Balakrishnan, Sivaraman and G’Sell, Max},
  year={2020}
}

@article{marso,
  title={Liraglutide and cardiovascular outcomes in type 2 diabetes},
  author={Marso, Steven P and Daniels, Gilbert H and Brown-Frandsen, Kirstine and Kristensen, Peter and Mann, Johannes FE and Nauck, Michael A and Nissen, Steven E and Pocock, Stuart and Poulter, Neil R and Ravn, Lasse S and others},
  journal={New England Journal of Medicine},
  volume={375},
  number={4},
  pages={311--322},
  year={2016},
  publisher={Mass Medical Soc}
}
\newpage
\appendix
\section{On projection parameters}
\label{App_A}

In related research, other authors have studied a variable importance measure that is closely related to our projection parameter, namely the so-called "best linear predictor/projection" (\cite{semenova}, \cite{vanLaan}, \cite{cui}, \cite{boileau}). In the case of continuous outcome, \textcite{semenova} provide theoretical results via debiased machine-learning (see e.g. \textcite{cher}) and \textcite{vanLaan} provide theoretical results based on semiparametric efficiency theory. In both cases they consider a target function/parameter, which is then approximated by a projection of the target function onto a working model indexed by a Euclidean parameter. This approach is different from ours in that we seek interpretable summary statistics of our target function (CATE) through a projection, rather than estimating an interpretable target parameter through a projection. The approaches given by \textcite{cui} and \textcite{boileau} are more akin to ours. In the former, they provide a procedure for estimation of the best linear projection of $\tau$ in a survival context but without theoretical results, where the latter considers the best linear projection of $\tau$ onto a linear model, in the case where $\E X_i=0$, $i=1, \ldots , d$, and derives an explicit parameter that is similar to ours, for which they provide an estimation procedure based on semiparametric efficiency theory. All of them consider a projection of $\tau$ onto a working model index by a Euclidean parameter. In contrast, our projection parameter is defined through a projection of $\tau$ onto a subspace indexed by $(\beta, w)$, where $\beta$ is a real-valued parameter and $w$ is a measurable function of $d-1$ variables. Since the space of functions indexed by a Euclidean parameter is a subspace of the space we consider, the error made from projecting $\tau$ onto a subspace is smaller in our setting compared to the best linear projection.

The above discussion will be clarified below. First, we state a result showing that our projection parameter is in fact the desired projection. \\ \\
Let $\mathcal{H}$ be the Hilbert space of measurable functions of $x$ with finite variance endowed with the covariance inner product. We have the following result: 
\begin{alemma}\label{lem:a.projection}
    The projection of $\tau \in \mathcal{H}$ onto the subspace $\mathcal{U} = \{u\in \mathcal{H} : u(x) = \beta x_j + w(x_{-j}), \ \beta \in \mathbb{R}, \  Pw^2 < \infty \}$ is given by
    $$
    \Pi(\tau \mid \mathcal{U}) = \beta^*X_{j} + w^*(X_{-j}),
    $$
    where 
    $$
    \beta^* = \Omega_j \quad \text{and} \quad w^*(X_{-j}) = E(\tau(X)\mid X_{-j}) - \Omega_j E(X_j\mid X_{-j}).
    $$
\end{alemma}

\begin{proof}
    We want to find $\beta^*$ and $w^*$ such that 
    $$
    (\beta^*, w^*) = \argmin_{\beta, w} \E\{\tau(X) - \beta X_j - w(X_{-j})\}^2.
    $$
    Observe that for any two measurable functions $a:\mathcal{X} \rightarrow \mathbb{R}$ and $b:\mathcal{X}_{-j} \rightarrow \mathbb{R}$ with finite variance, we have
    \begin{align*}
    E([a(X) - b(X_{-j})]^2\mid X_{-j}) = \left[E(a(X)\mid X_{-j}) - b(X_{-j}) \right]^2 + \var(a(X)\mid X_{-j}).
    \end{align*}
    Hence
    \begin{align*}
        \E\{\tau(X) - \beta X_j - w(x_{-j})\}^2 =& \E\left\{ [E(\tau(X) - \beta X_j \mid X_{-j}) - w(X_{-j})]^2 \right\} \\
        &+ \E\{ \var(\tau(X) - \beta X_j \mid X_{-j})\}.
    \end{align*}
    The second term on the right hand side of the latter display does not depend on $w$, and since the integrand in the first term is positive, the expression is minimized in $w$ when the first term is equal to zero. This implies
    $$
    w^*(x_{-j}) = E(\tau(X) - \beta X_{j}\mid X_{-j} = x_{-j}).
    $$
    For $\beta$, observe that
    $$
    \frac{\dd}{\dd \beta}  \E\{\tau(X) - \beta X_j - w(X_{-j})\}^2 = -2 \E\{X_j[\tau(X) - \beta X_j - w(x_{-j})]\} = 0, 
    $$
    together with $w^*$ implies
    \begin{align*}
        \E\{X_j[\tau(X) - \beta X_j - E(\tau(X) - \beta X_{j}\mid X_{-j})]\} = 0,          
    \end{align*}
    and therefore 
    $$
    \beta^* = \frac{\E\{ \cov(\tau(X), X_j\mid X_{-j}) \}}{\E\{ \var(X_j\mid X_{-j}) \}} = \Omega_j
    $$
    and
    $$
    w^*(X_{-j}) = E(\tau(X)\mid X_{-j}) - \Omega_j E(X_j\mid X_{-j}).
    $$
\end{proof}
Next, we consider the best linear projection of $\tau$ and contrast it with the best partially linear projection. To that end, we write the CATE function as
$$
\tau(x) = \alpha + \gamma^T x + R_{\alpha, \gamma}(x)
$$
for $\alpha \in \mathbb{R}$ and $\gamma \in \mathbb{R^d}$ and let
$$
(\alpha^*, \gamma^*) = \argmin_{\alpha, \gamma} \E\{R_{\alpha, \gamma}(X)^2\} = \argmin_{\alpha, \gamma} \E\{[\tau(X) - \alpha - \gamma^T X]^2\}.
$$
We define the remainder corresponding to the best partially linear projection as $R_{\beta, w}(x)$ such that 
$$
(\beta^*, w^*) = \argmin_{\beta, w}\E\{ R_{\beta, w}(X_j, X_{-j})^2 \} = \argmin_{\beta, w}\E\{ [\tau(X) - \beta X_j - w(X_{-j})]^2 \}.
$$
Since the space of linear functions is a subspace of the space of partially linear functions, $\mathcal{U}$, we have that $\alpha^* + \gamma^{*T}x \in \mathcal{U}$. By the projection theorem for Hilbert spaces (see e.g. \cite{tsiatis}, Theorem 2.1), the distance between $\tau$ and the projection onto $\mathcal{U}$ is smaller than the distance between $\tau$ and any other function in $\mathcal{U}$. Hence, by Lemma \ref{lem:a.projection}, 
$$
\norm{R_{\beta^*, w^*}} \leq \norm{R_{\alpha^*,\gamma^*}}.
$$
The result shows that the error made from model misspecification is smaller in the best partially linear projection compared to the best linear projection. As the dimension of $X$ grows, the difference in the errors become larger (since the linear restriction of $X_{-j}$ becomes increasingly strict compared to $w(X_{-j})$), and the best partially linear projection is thus better suited as a measure of importance of a single covariate.

\section{Derivation of efficient influence functions}
\label{AppB}
Here we derive the efficient influence functions from theorem \ref{EIFs} and \ref{EIFomega}. We will first derive the Gateaux derivatives of $\tau$, $\tau_l$ and $\E^j(x)$, respectively, which are then used in the derivations of the EIF's through the chain rule. In the following we use subscripts $P$ to underline the dependence on $P$. Let the parametric submodel be given by $P_\epsilon = Q\epsilon + (1-\epsilon)P \in \mathcal{M}$, where $Q$ is the Dirac measure in a single observation $O$, and define the operator $\partial_\epsilon = \left.\frac{d}{d\epsilon}\right\vert_{\epsilon=0}$ such that $\partial_\epsilon \psi(P_\epsilon)= \left.\frac{d}{d\epsilon} \psi(P_\epsilon) \right\vert_{\epsilon=0}$ for some mapping $\psi:\mathcal{M}\rightarrow \mathbb{R}$. In the following let 
$$
g(A,X) = \left(\frac{\mathbb{1}(A=1)}{\pi(1\mid X)} - \frac{\mathbb{1}(A=0)}{\pi(0\mid X)} \right)
$$
and
$$H(u,t,a,x) = \int_u^t S(s\mid a, x)\dd s,$$
and we write $\tau = \tau_P$ to denote $\tau$ under distribution $P$.

\begin{alemma}\label{gat:tau}
    The Gateaux derivative of $\tau(x)$ is given by 
    \begin{align*}
    \partial_\epsilon S_{P_\epsilon}(t\mid 1, x) - S_{P_\epsilon}(t\mid 0, x) = \frac{\mathbb{1}(X=x)}{f(x)}g(A,x)\int_0^t \frac{-S(t\mid A, x)}{S(s\mid A, x)S_c(s\mid A, x)} \dd M(s\mid A, x)
    \end{align*}
    in the survival functions setting and
    \begin{align*}
    &\partial_\epsilon \int_0^{t^*} S_{P_\epsilon}(t\mid 1, x) - \int_0^{t} S_{P_\epsilon}(t\mid 0, x) \dd t \\
    =& \frac{\mathbb{1}(X=x)}{f(x)}g(A,x)\int_0^{t} \frac{-H(u, t, a, x)}{S(u\mid a, x)S_c(u\mid a, x)} \dd M(u\mid a,x),
    \end{align*}
    in the RMST setting.
\end{alemma}

\begin{proof}
    We start by calculating the Gateaux derivative of the conditional cumulative hazard function $\Lambda(t\mid a, x)$, which is also given in, e.g, the supplementary material of \textcite{martSten}, but included here for completeness. Let $P(\tilde{T}\geq s, a, x) = \sum_{\delta=0,1} \int_s^{\infty}P(\dd s, \delta, a, x)$ and note that $P(\tilde{T}\geq s \mid a, x) = S(s\mid a, x)S_c(s\mid a, x)$ because of independent censoring. Then 
\begin{align}
&\partial_\epsilon \Lambda_\epsilon(t\mid a, x) \nonumber \\
=& \int_0^t \partial_\epsilon \frac{P_\epsilon (\dd s,\Delta = 1, a, x)}{P_\epsilon( \tilde{T} \geq s, a, x)} \nonumber \\
=& \int_0^t \frac{Q(\dd s, \Delta = 1, a, x) - P(\dd s, \Delta = 1, a, x)}{P(\tilde{T}\geq s, a, x)} \nonumber \\
&- \int_0^t \left( \sum_{\delta = 0,1} \mathbb{1}(\tilde{T}\geq s, \delta, a, x) - P(\tilde{T}\geq s, \delta, a, x) \right) \frac{P(\dd s, \Delta = 1, a, x)}{P(\tilde{T} \geq s, a, x)^2} \nonumber \\
=& \frac{\mathbb{1}(A=a)\mathbb{1}(X=x)}{\pi(a\mid x) f(x)}\left\{\int_0^t \frac{1}{P(\tilde{T}\geq s \mid a, x)} \dd N(s) - \int_0^t \frac{\mathbb{1}(\tilde{T}\geq s)}{P(\tilde{T}\geq s \mid a, x)} \dd \Lambda(s\mid a, x) \right\} \nonumber \\
&= \frac{\mathbb{1}(A=a)\mathbb{1}(X=x)}{\pi(a\mid x) f(x)}\int_0^t \frac{1}{S(s\mid a, x)S_c(s\mid a, x)} \dd M(s\mid a, x). \nonumber \label{eq:eifLambda}
\end{align}

Consider the survival function setting $\tau(x) = e^{-\Lambda(t\mid, 1,x)} - e^{-\Lambda(t\mid, 0,x)}.$ A simple application of the chain rule gives
\begin{align*}
    \partial_\epsilon \tau_{P_\epsilon}(x) &= \frac{\mathbb{1}(X=x)}{f(x)}g(A,x)\int_0^t \frac{-S(t\mid A, x)}{S(s\mid A, x)S_c(s\mid A, x)} \dd M(s\mid A, x),
\end{align*}
which gives the first claim.
Now, consider the RMST setting, 
$$\tau(x) = \int_0^{t} S(u\mid 1,x)\, du-\int_0^{t} S(u\mid 0,x)\, du.$$ 
Again, the chain rule gives 
\begin{align*}
&\partial_\epsilon \int_0^{t} S_{P_\epsilon}(s\mid 1, x) \dd s - \int_0^{t} S_{P_\epsilon}(t\mid 0, x) \dd s \\
=& \frac{\mathbb{1}(X=x)}{f(x)}g(A,x)\int_0^{t} \frac{-H(u, t, a, x)}{S(u\mid a, x)S_c(u\mid a, x)} \dd M(u\mid a,x)  
\end{align*}
where
$$H(u,t,a,x) = \int_u^t S(s\mid a, x)\dd s, $$
which gives the second claim.
\end{proof}

The next result is from \textcite{hines}. The result is stated in equation (4) in their Appendix.
\begin{alemma}[\cite{hines}]\label{hines}
    Let $g_P(X)$ denote some functional of $P$. Then
    \begin{align*}
        &\partial_\epsilon \E_{P_\epsilon}(g_{P_\epsilon}(X)\mid X_{-l} = x_{-l}) \\
        =& \frac{1(X_{-l} = x_{-l})}{f(x_{-l})}\left[g_P(x) - \E(g_P(X)\mid X_{-l} = x_{-l}) \right] + \E_P(\partial_\epsilon g_{P_\epsilon}(X)\mid X_{-l} = x_{-l}). 
    \end{align*}
\end{alemma}

\begin{alemma}\label{gat:taus}
    The Gateaux derivative of $\tau_l(x)$ is given by 
    \begin{align*}
        &\partial_\epsilon \E_{P_\epsilon}(\tau_{P_\epsilon}(X) \mid X_{-l} = x_{-l}) \\
        =&\frac{\mathbb{1}(X_{-l} = x_{-l})}{f_{x_{-l}}(x_{-l})}\left(\tau(x) - \tau_l(x) + g(A,X)\int_0^t \frac{-S(t\mid A, x)}{S(s\mid A, x)S_c(s\mid A, x)} \dd M(s\mid A, x) \right)
    \end{align*}
    in the survival setting and 
    \begin{align*}
        &\partial_\epsilon \E_{P_\epsilon}(\tau_{P_\epsilon}(X) \mid X_{-l} = x_{-l}) \\
        =&\frac{\mathbb{1}(X_{-l} = x_{-l})}{f_{x_{-l}}(x_{-l})}\left(\tau(x) - \tau_l(x) + g(A,X)\int_0^t \frac{-H(u, t, a, x)}{S(s\mid A, x)S_c(s\mid A, x)} \dd M(s\mid A, x) \right)
    \end{align*}
    in the RMST setting.
\end{alemma}

\begin{proof}
We note that for any functional $g_P(X)$ with Gateaux derivative $\frac{\mathbb{1}(X=x)}{f(x)}v(O)$ for some function $v:\mathcal{O}\rightarrow \mathbb{R}$ we have 
$$
\E_P(\partial_\epsilon g_{P_\epsilon}(X)\mid X_{-s} = x_{-s}) = \frac{\mathbb{1}(X_{-s} = x_{-s})}{f(x_{-s})}v(O).
$$
Let $g_P(x) = \tau(x)$. An application of lemma \ref{hines} followed by an application of lemma \ref{gat:tau} gives the result.
\end{proof}

\subsection{Proof of Theorem \ref{EIFs}}
\textit{EIF in survival setting} \\\\
Consider the survival function setting, $\tau(x) = \exp(-\Lambda(t\mid A=1, x)) - \exp(-\Lambda(t\mid A=0, x))$. Lemma \ref{gat:tau} gives the Gateaux derivative of $\tau$ from which we can calculate the EIF's of $\var(\tau(X))$ and $\var(\tau_l(X))$ by simple applications of the chain rule.
\begin{align*}
&\partial_\epsilon \var_{P_\epsilon}(\tau_{P_\epsilon}(X)) \\
=& \int \partial_\epsilon (\tau_{P_\epsilon}(X) - \E \tau_{P_\epsilon}(X))^2\dd P_\epsilon \nonumber \\
=&  (\tau_{P}(X) - \E \tau_{P}(X))^2 - \int (\tau_{P}(X) - \E \tau_{P}(X))^2\dd P + \int \partial_\epsilon (\tau_{P_\epsilon}(X) - \E \tau_{P_\epsilon}(X))^2\dd P \nonumber \\
=& (\tau_{P}(X) - \E \tau_{P}(X))^2 - \var(\tau(X)) + \int  2(\tau_{P}(X) - \E \tau_{P}(X)) \partial_\epsilon (\tau_{P_\epsilon}(X) - \E \tau_{P_\epsilon}(X))\dd P \nonumber \\
=& (\tau_{P}(X) - \E \tau_{P}(X))^2 - \var(\tau(X)) + \int  2(\tau_{P}(X) - \E \tau_{P}(X)) \partial_\epsilon \tau_{P_\epsilon}(X)\dd P \nonumber \\
=& (\tau_P(X) - \E\tau_P(X))^2 - \var(\tau_P(X))   \\
&+ 2(\tau_P(X) - \E\tau_P(X))g(A,X)\int_0^t \frac{-S(t\mid A, x)}{S(s\mid A, x)S_c(s\mid A, x)} \dd M(s\mid A, x) \nonumber \\
=& \ \tilde{\psi}_{\var(\tau(X))}
\end{align*}
Analogously we find the EIF of $\var(\tau_l(X))$ by use of the Gateaux derivative of $\tau_l$ from lemma \ref{gat:taus}:
\begin{align*}
&\partial_\epsilon \var_{P_\epsilon}(\tau_{l,P_\epsilon}(X)) \\
=& \int \partial_\epsilon (\tau_{l,P_\epsilon}(X) - \E \tau_{s,P_\epsilon}(X))^2\dd P_\epsilon \nonumber \\
=& (\tau_l(X) - \E(\tau(X)))^2 - \var(\tau(X)) \\
&+ 2(\tau_l(X) - \E(\tau(X)))\left(\tau(x) - \tau_l(x) + g(A,X)\int_0^t \frac{-S(t\mid A, x)}{S(s\mid A, x)S_c(s\mid A, x)} \dd M(s\mid A, x) \right)  \nonumber \\
=& \tilde{\psi}_{\var(\tau_l(X))}
\end{align*}
noting that $\E\tau_l(X) = \E\tau(X)$. From the two EIF's we have that the EIF of $\Theta_l$ is given by their difference:
\begin{equation}\label{eq:eifThetaApp}
    \tilde{\psi}_{\Theta_l} = \tilde{\psi}_{\var(\tau(X))} - \tilde{\psi}_{\var(\tau_l(X))}
\end{equation}
and the EIF of $\Psi_l$ is given by 
\begin{equation}\label{eq:eifPsiApp}
  \Phi_l(O) = \frac{1}{\var(\tau(X))}\left(\tilde{\psi}_{\Theta_l}(O) - \Psi_l \tilde{\psi}_{\var(\tau(X))}(O) \right).  
\end{equation} \\ \\
\textit{EIF in restricted mean setting} \\ \\
Let $\tau(x) = \int_0^{t^*} S(t\mid 1,x) \dd t-\int_0^{t^*} S(t\mid 0,x) \dd t$. Since the structure of the Gateaux derivatives of $\tau$ and $\tau_l$, from lemma \ref{gat:tau} and \ref{gat:taus}, is identical to the survival case with $H(u,t,a,x)$ replacing $S(t\mid a, x)$, the calculations from the survival function setting apply and we have the EIF's of $\Theta_l$ and $\Psi_l$ are given by \eqref{eq:eifThetaApp} and \eqref{eq:eifPsiApp}, respectively with

\begin{align*}
\tilde{\psi}_{\var(\tau(X))} =& (\tau_P(X) - \E\tau_P(X))^2 - \var(\tau_P(X))   \\
&+ 2(\tau_P(X) - \E\tau_P(X))g(A,X)\int_0^t \frac{-H(u,t, A, x)}{S(u\mid A, x)S_c(u\mid A, x)} \dd M(u\mid A, x) \nonumber
\end{align*}
and
\begin{align*}
 &\tilde{\psi}_{\var(\tau_l(X))} \\
 =&  (\tau_l(X) - \E(\tau_l(X)))^2 - \var(\tau(X)) \\
&+ 2(\tau_l(X) - \E(\tau_l(X)))\left(\tau(x) - \tau_l(x) + g(A,X)\int_0^t \frac{-H(u,t, A, x)}{S(u\mid A, x)S_c(u\mid A, x)} \dd M(u\mid A, x) \right)  \nonumber
\end{align*}

\subsection{Proof of Corollary \ref{hinesEIF}}

Note that in both the survival function and RMST setting, the EIFs of $\var\{\tau(X)\}$ and $\var\{\tau_l(X)\}$ can be written as
$$
\tilde{\psi}_{\var\{\tau(X)\}} = [\tau(X) - \E\{\tau(X)\}]^2 + 2[\tau(X) - \E\{\tau(X)\}][\varphi(O) - \tau(X)] - \var\{\tau(X)\}
$$
and
$$
\tilde{\psi}_{\var\{\tau_l(X)\}} = [\tau(X) - \E\{\tau(X)\}]^2 + 2[\tau(X) - \E\{\tau(X)\}][\varphi(O) - \tau_l(X)] - \var\{\tau_l(X)\}.
$$
A simple rewriting of the above EIFs gives
\begin{align*}
\tilde{\psi}_{\var\{\tau(X)\}} =& [\tau(X) - \E\{\tau(X)\}]^2 - 2\tau(X)^2 + 2\tau(X)\E\{\tau(X)\} \nonumber \\
&+ 2[\tau(X) - \E\{\tau(X)\}]\varphi(O) - \var\{\tau(X)\} \nonumber\\
=& \E\{\tau(X)\}^2 - \tau(X)^2 + 2[\tau(X) - \E\{\tau(X)\}]\varphi(O) - \var\{\tau(X)\} \nonumber \\
=& [\varphi(O) - \tau_d]^2 - [\varphi(O) - \tau(X)]^2 - \var\{\tau(X)\} \nonumber 
\end{align*}
and analogously for $\var\{\tau_l(X)\}$:
$$
\tilde{\psi}_{\var\{\tau(X)\}} = [\varphi(O) - \tau_d]^2 - [\varphi(O) - \tau_l(X)]^2 - \var\{\tau_l(X)\}.
$$
Subtracting the two gives the EIF for $\Theta_l$ and the chain rule gives the EIF for $\Psi_l$.

\subsection{Proof of Theorem \ref{EIFomega}}
Let $g_P(X) = X_j$. Lemma \ref{hines} then gives
\begin{align}\label{gat:ej}
    \partial_\epsilon \E_{P_\epsilon}(X_j\mid X_{-j}) = \frac{\mathbb{1}(X_{-j} = x_{-j})}{f(x_{-j})}(X_j - \E(X_j\mid X_{-j})). \nonumber
\end{align}
It follows immediately that the EIF of $\chi_j$ is given by
\begin{align*}
    &\partial_\epsilon \chi_j(P_\epsilon) \\
    =& \partial_\epsilon \E_{P_\epsilon}\{X_j - \E_{P_\epsilon}(X_j\mid X_{-j}) \}^2 \nonumber \\
    =& (X_j - \E(X_j\mid X_{-j}))^2 - \chi_j(P) \nonumber \\
    &- 2\int [x_j - \E(X_j\mid X_{-j} = x_{-j})] \frac{\mathbb{1}(X_{-j} = x_{-j})}{f(x_{-j})}[X_j - \E(X_j\mid X_{-j} = x_{-j})] P_{X_j, X_{-j}}\left(\dd (x_j, x_{-j}) \right) \nonumber \\
    =& (X_j - \E(X_j\mid X_{-j}))^2 - \chi_j(P) \nonumber \\
    &- 2[X_j - \E(X_j\mid X_{-j})]\int [x_j - \E(X_j\mid X_{-j})]  P_{X_j \mid X_{-j} = x_{-j}}\left(\dd x_j \right) \nonumber \\
    =& (X_j - \E(X_j\mid X_{-j}))^2 - \chi_j(P) \nonumber \\
    =& \tilde{\psi}_{\chi_j}. \nonumber
\end{align*}
For the derivation of the EIF of $\Gamma_j$ we let $\varphi$ denote the uncentered EIF of the ATE regardless of whether we consider the survival function setting or the RMST setting. Hence, by lemma \ref{gat:tau}, we write the Gateaux derivative of the CATE function as
$$
\partial_\epsilon \tau_{P_\epsilon}(x) = \frac{\mathbb{1}(X=x)}{f(x)}(\varphi(O) - \tau(x)),
$$
and, by lemma \ref{gat:taus}, the Gateaux derivative of $\tau_{\{j\}}$ as
$$
\partial_\epsilon \tau_{\{j\},P_\epsilon}(x) = \frac{\mathbb{1}(X_{-j}=x_{-j})}{f(x_{-j})}(\varphi(O) - \E(\tau(X)\mid X_{-j})).
$$
Observe that 
\begin{equation}
\setlength{\jot}{10pt}
\begin{aligned}
    \Gamma_j =& \E\{ \cov(\tau(X), X_j \mid X_{-j}) \} \nonumber \\
    =& \E\{E([\tau(X) - \E(\tau(X)\mid X_{-j})] [X_j - \E(X_j \mid X_{-j})] \mid X_{-j}) \} \nonumber \\
    =& \E\{E(\tau(X)[X_j - \E(X_j \mid X_{-j})]\mid X_{-j}) - \E(\tau(X)\mid X_{-j}) \E(X_j - \E(X_j \mid X_{-j}) \mid X_{-j}) \} \nonumber \\
    =& \E\{\tau(X)X_j\} - \E\{ \E(\tau(X)\mid X_{-j}) \E(X_j\mid X_{-j}) \}, \nonumber
\end{aligned}
\end{equation}
and that the EIF of $\Gamma_j$ is given by
\begin{equation}\nonumber
\setlength{\jot}{10pt}
\begin{aligned}
    \partial_\epsilon \Gamma_j(P_\epsilon) =& \partial_\epsilon \big\{\E_{P_\epsilon}\{\tau_{P_\epsilon}(X)X_j\} - \E_{P_\epsilon}\{ \E_{P_\epsilon}(\tau_{P_\epsilon}(X)\mid X_{-j}) \E_{P_\epsilon}(X_j\mid X_{-j}) \} \big\} \\
    =& \tau(X) X_j - \E\{\tau(X)X_j\} + X_j[\varphi(O) - \tau(X)] \\
    &- \E(\tau(X)\mid X_{-j}) \E(X_j\mid X_{-j}) + \E\{ \E(\tau(X)\mid X_{-j}) \E(X_j\mid X_{-j}) \} \\
    &- \partial_\epsilon \E\{ \E_{P_\epsilon}(\tau_{P_\epsilon}(X)\mid X_{-j}) \E(X_j\mid X_{-j}) \} - \partial_\epsilon \E\{ \E(\tau(X)\mid X_{-j}) \E_{P_\epsilon}(X_j\mid X_{-j}) \} \\
    =& \tau(X) X_j - \E\{\tau(X)X_j\} + X_j[\varphi(O) - \tau(X)] \\
    &- \E(\tau(X)\mid X_{-j}) \E(X_j\mid X_{-j}) + \E\{ \E(\tau(X)\mid X_{-j}) \E(X_j\mid X_{-j}) \} \\
    &-  [\varphi(O) - \E(\tau(X)\mid X_{-j})] \E(X_j\mid X_{-j}) -  \E(\tau(X)\mid X_{-j}) [X_j - (X_j\mid X_{-j})] \\
    =& \varphi(O)[X_j - \E(X_j\mid X_{-j}] - \E(\tau(X)\mid X_{-j})[X_j - \E(X_j\mid X_{-j})] - \Gamma_j \\
    =& [\varphi(O) - \E(\tau(X)\mid X_{-j})][X_j - \E(X_j\mid X_{-j})] - \Gamma_j \\
    =& \tilde{\psi}_{\Gamma_j}
\end{aligned}    
\end{equation}
The EIF of $\Omega_j$ follows by an application of the chain rule.

\section{Sample splitting}
\label{App_Sample_splitting}
Split the index set $\{1, \ldots, n\}$ (uniformly at random) into K disjoint sets $\mathcal{T}_1, \mathcal{T}_2, \ldots, \mathcal{T}_K$, such that $\{1, \ldots, n\} = \dot{\cup}_{k=1}^k \mathcal{T}_k$. Let $\mathcal{V}_k$ denote the subset of the observed data corresponding to the $k$'th index set, $\mathcal{T}_k$, i.e., $\mathcal{V}_k = \{O_i: \ i \in \mathcal{T}_k\}$, such that $\mathcal{O} = \dot{\cup}_{k=1}^K \mathcal{V}_k.$ Let $\mathbb{P}_n^k$ be the empirical measure in the sample $\mathcal{V}_k$ and let $\phi_{\Theta_l} = (\varphi-\tau_l)^2 - (\varphi-\tau)^2$ denote the uncentered EIF of $\Theta_l$ with $\hat{\phi}_{\Theta_l}$ being an estimate obtained by plugging in estimated nuisance parameters in the expression for $\phi$. Let $\hat{\phi}_{\Theta_l, -k}$ be the estimate of $\phi_{\Theta_l}$ based on data in $\mathcal{V}_{-k} = \cup_{i \neq k} \mathcal{V}_i$. The cross-fitted one-step estimator is then given by
\begin{equation}
    \hat{\Theta}_l^{CF} = \sum_{i = k}^K \frac{n_k}{n} \mathbb{P}_n^k \hat{\phi}_{\Theta_l, -k} = \frac{1}{n}\sum_{k = 1}^K \sum_{i \in \mathcal{T}_k} \left\{ (\hat{\varphi}_{-k}(O_i)-\hat{\tau}_{l, -k}(X_i))^2 - (\hat{\varphi}_{-k}(O_i)-\hat{\tau}_{-k}(X_i))^2 \right\} \nonumber
\end{equation}
where $n_k$ is the number of observations in $\mathcal{V}_k$ and $\hat{\varphi}_{-k} = \varphi(\hat{\Lambda}_{-k}, \hat{\Lambda}_{C, -k}, \hat{\pi}_{-k})$ with $(\hat{\Lambda}_{-k}, \hat{\Lambda}_{C, -k}, \hat{\pi}_{-k})$ being the nuisance estimators obtained from the sample $\mathcal{V}_{-k}$ . We let $\hat{\Theta}_d^{CF}$ be defined in the same way with $d$ instead of $l$ in the above formula and note that $\hat{\tau}_{d,-k}$ corresponds to the ATE estimate based on the data in $\mathcal{V}_{-k}$ (see \ref{sec:nuis}).

Next we state a set of conditions from which the asymptotic distribution of $\hat{\Theta}_l^{CF}$ is obtained.
\begin{NoHyper}
\begin{assumption}[Nuisance parameters]\label{assA}
Let $g(s \mid a,x) = \pi(a\mid x)S_c(s\mid a,x)$. For nuisance estimates $\hat{\pi}, \hat{\Lambda}$, $\hat{\Lambda}_C$ define
\begin{itemize}
    \item $\hat{\tau}(x) = e^{-\hat{\Lambda}(t\mid A=1, x)} - e^{-\hat{\Lambda}(t\mid A=0, x)}$
    \item $\hat{g}(s\mid a, x) = \hat{\pi}(a\mid x)e^{-\hat{\Lambda}_c(s\mid a, x)}$
    \item $\hat{\tau}_l(x) = \hat{E}_n(\hat{\tau}(X) \mid X_{-l} = x_{-l})$
\end{itemize}
where $\hat{E}_n$ is some regression of $\hat{\tau}(X)$ onto $X_{-l}$. Define $\hat{L}(s,t\mid a,x) = \frac{S(s\mid a, X)}{ \hat{S}(s\mid a,X)} \hat{S}(t\mid a,X)$ in the survival function setting and $\hat{L}(s,t\mid a,x) = \frac{\hat{H}(s,{t}\mid a, X) S(s\mid a, X)}{\hat{S}(s\mid a, X)}$ in the RMST setting. Assume that the nuisance parameters are chosen such that

\begin{enumerate}[label=\ref{assA}\arabic*]
    \item \label{ass1} 
    $\exists \eta > 0, \ s.t. \ \eta < \hat{g}(s \mid a, x) \ and \ \eta < e^{-\hat{\Lambda}(s \mid a, x)} \quad \forall (s,a,x) \in [0, t] \times \{0,1\} \times \mathcal{X}.$
    \item \label{ass2}
    $\norm{\hat{\tau}(x) - \tau(x)}_{L_2(P)} = o_p(n^{-\frac{1}{4}}).$
    \item \label{ass3}
    $\norm{\hat{\tau}_l(x) - \tau_l(x)}_{L_2(P)} = o_p(n^{-\frac{1}{4}}).$
    \item \label{ass4}
    $  \E \left\{ \int_0^t \left(1 - \frac{g(s\mid a, X)}{\hat{g}(s\mid a, X)} \right) \hat{L}(s,t\mid a,x) \dd \left[\Lambda(s\mid a, X) - \hat{\Lambda}(s\mid a, X)\right]   \right\} = o_p(n^{-\frac{1}{2}}).$
    \item \label{ass5}
    $\norm{\ \sup_{s<t}\left|\hat{g}(s\mid a, x) - g(s\mid a,x)\right|\ }_{L_2(P)} = o_p(1)$. \\
    $\norm{\ \sup_{s<t}\left|\hat{\Lambda}(s\mid a, x) - \Lambda(s\mid a,x)\right|\ }_{L_2(P)} = o_p(1).$
    \item \label{ass6}
    $|\hat{\tau}_l(x) - \hat{\tau}(x)| \leq \delta < \infty$ for almost all $x$.
\end{enumerate}
\end{assumption}
\end{NoHyper}

Before proceeding to estimation of $\Theta_d$, some comments about assumption \ref{assA} are in order. \ref{ass1} is a common positivity assumption, which states that all individuals have a positive probability of receiving treatment and being under observation for the entire time horizon. Assumption \ref{ass2} and \ref{ass3} refer to convergence rates of the CATE function estimate as well as the conditional CATE function estimate. For CATE estimates given by $\hat{\tau}(x) = e^{\hat{-\Lambda}(t\mid A=1, x)} - e^{-\hat{\Lambda}(t\mid A=0, x)}$, the assumption boils down to an assumption on the convergence rate of the survival function estimate, but this is seen to be a mild assumption for which many ML-methods concur (see e.g. the discussion in Section 4.3 in \cite{kennedydouble}). We note though, that assumption \ref{ass2} and \ref{ass3} imply that the estimator is not double robust in the sense that one only needs the outcome or the censoring and propensity to be correctly specified in order to obtain a consistent estimator, which is in contrast to other related estimands (e.g. the ATE, see \cite{westling}, theorem 2). From the structure of the remainder term (see the proof of Theorem \ref{asTheta} in Appendix \ref{sec:appAN}), it is seen that the estimator is consistent if $\Lambda$, $\tau$ and $\tau_l$ are consistent, but that it is not the case if only the censoring and propensity are consistent. Thus, it is important to employ flexible methods for obtaining nuisance estimators $\hat{\Lambda}$, $\hat{\tau}$, and $\hat{\tau}_l$. Furthermore, since $\tau_l(x) = \E(e^{-\Lambda(t\mid 1, X)} - e^{-\Lambda(t\mid 0, X)} \mid X=x)$ it will generally be a complicated function, even for a correctly specified $\hat{\Lambda}$ (e.g. as a Cox regression with a Breslow baseline hazard), emphasizing the need for a flexible estimator $\hat{\E}_n$.

The assumption \ref{ass4} corresponds to a bound on the aforementioned remainder term (see proof of Theorem \ref{asTheta} in Appendix \ref{sec:appAN}). In studies on related target parameters (e.g. the ATE) with uncensored outcome, the related bound on the remainder term is seen to have a product structure, in the sense that the product of the $L_2(P)$-norms of the outcome regression and the propensity estimator needs to be $o_p(n^{-1/2})$ (see \cite{kennedydouble}). This is then achieved if both estimators converge on $n^{-1/4}$ rate or, e.g., if one estimator is bounded in probability and the other converges on parametric rate. In our case $\hat{\Lambda}$ will often be a step function and one has to study \ref{ass4} in greater detail in order to obtain a product structure result analogous to the uncensored case. This is beyond the scope of this paper, but we will expect it to be the case in many settings. \ref{ass5} corresponds to uniform consistency of the time-to-event nuisance parameters. Assumption \ref{ass6} is a technical assumption, which we would expect to hold for most reasonable choices of $\hat{E}_n$.

Let $\hat{\tilde{\psi}}_{\Psi_l, -k}$ denote the estimate of the EIF $\tilde{\psi}_{\Psi_l}$ based on data from $\mathcal{V}_{-k}$. Define the variance estimator
$$
\hat{\sigma}^{2,CF}_{\Psi_l} = \sum_{i = k}^K \frac{n_k}{n} \mathbb{P}_n^k \hat{\tilde{\psi}}_{\Psi_l, -k}^2 =  \frac{1}{n}\sum_{k = 1}^K \sum_{i \in \mathcal{T}_k} \left[\frac{1}{\hat{\Theta}_d^{CF}} \left( \hat{\phi}_{\Theta_l, -k}(O_i) - \hat{\Theta}_l^{CF} - \hat{\Psi}_l ( \hat{\phi}_{\Theta_d}(O_i) - \hat{\Theta}_d^{CF} ) \right) \right]^2.
$$
Lemma \ref{lem:SECF} below  gives that the variance estimator above is consistent, and a confidence interval of $\hat{\Psi}_l^{CF}$ is then constructed as $\hat{\Psi}_l^{CF} \pm 1.96 \sqrt{\hat{\sigma}^{2,CF}_{\Psi_l} / n}$

As with $\Psi_l$, we construct cross-fitted one-step estimators for $\Gamma_j$ and $\chi_j$ based on the EIF's in Theorem \ref{EIFomega} where we let $\phi_{\Gamma_j} = [ \varphi(O) - \tau_j(X) ][ X_j - \E(X_j\mid X_{-j}) ]$ and $\phi_{\chi_j} = [X_j - \E(X_j\mid X_{-j})]^2$ denote the uncentered EIF's. Then, using the same sample splitting notation as in the construction of $\hat{\Theta}_l^{CF}$, we denote $\hat{E}^j_{n,-k}$ the regression of $X_j$ onto $X_{-j}$ in the sample $\mathcal{V}_{-k}$ and define the estimators
\begin{equation}
    \hat{\Gamma}_j^{CF} = \sum_{i = k}^K \frac{n_k}{n} \mathbb{P}_n^k \hat{\phi}_{\Gamma_j,-k} = \frac{1}{n}\sum_{k = 1}^K \sum_{i \in \mathcal{T}_k} \left\{ [ \hat{\varphi}_{-k}(O_i) - \hat{\tau}_{j,-k}(X) ][ X_{i,j} - \hat{E}^j_{n,-k}(X_{i, -j}) ] \right\} \nonumber
\end{equation}
and
\begin{equation}
    \hat{\chi}_j^{CF} = \sum_{i = k}^K \frac{n_k}{n} \mathbb{P}_n^k \hat{\phi}_{\chi_j,-k} = \frac{1}{n}\sum_{k = 1}^K \sum_{i \in \mathcal{T}_k} [X_{i, j} - \hat{E}^j_{n,-k}(X_{i, -j})]^2. \nonumber
\end{equation}
The two estimators are combined to create an estimator for $\Omega_j$:
\begin{equation}
    \hat{\Omega}_j^{CF} = \frac{\hat{\Gamma}_j^{CF}}{\hat{\chi}_j^{CF}}. \nonumber
\end{equation}

To state results on the asymptotic distribution of the  above estimators, we need a slight modification of assumption \ref{assA} as follows:

\begin{NoHyper}
\begin{assumption}\label{assB}\ \\  
\begin{enumerate}[label=\ref{assB}\arabic*]
    \item \label{assb1} 
    $\norm{\hat{\E}^j_n - \E(\cdot \mid X_{-j})}_{L_2(P)} = o_p(n^{-\frac{1}{4}}).$
    \item \label{assb2}
    $(X_j - \hat{E}_n^j(X_{-j}))^2 \leq \delta < \infty, \quad \delta > 0, \quad a.s. $
    \item \label{assb3}
    $ \var(X_j\mid X_{-j}) < \infty$ for all $j \in \{1, \ldots d\}$.
\end{enumerate}
\end{assumption}
\end{NoHyper}

Assumption \ref{assb1} relates to the convergence rate of $\hat{E}_n^j$, and it is similar to assumption \ref{ass2} and \ref{ass3}. Assumption \ref{assb2} is a technical assumption, which we would expect to hold for most reasonable estimators $\hat{E}_n^j$, and assumption \ref{assb3} assumes all the conditional distributions of $X_j$ given $X_{-j}$ to have second moment.

Let $\hat{\tilde{\psi}}_{\Omega_j, -k}$ be the estimate of $\tilde{\psi}_{\Omega_j}$ in the sample $\mathcal{V}_{-k}$. We define the cross-fitted plug-in variance estimator as
$$
\hat{\sigma}^{2,CF}_{\Omega_j} =  \sum_{i = k}^K \frac{n_k}{n} \mathbb{P}_n^k \hat{\tilde{\psi}}_{\Omega_j, -k}^2 =  \frac{1}{n}\sum_{k = 1}^K \sum_{i \in \mathcal{T}_k} \left[\frac{1}{\hat{\chi}_j^{CF}} \left( \hat{\phi}_{\Gamma_j, -k}(O_i) - \hat{\Gamma}_j^{CF} - \hat{\Omega}_j ( \hat{\phi}_{\chi_j, -k}(O_i) - \hat{\chi}_j^{CF} ) \right) \right]^2.
$$

\begin{lemma}\label{lem:SECF}
    Let $\psi_1$ and $\psi_2$ be two pathwise differentiable maps from $\mathcal{M}$ to the reals with EIF's given by $\tilde{\psi}_1$ and $\tilde{\psi}_2$, respectively, where $\tilde{\psi}_i(\psi_i, \nu_i) = \varphi_i(\nu_i) - \psi_i$, $i=1,2$, for some nuisance parameters $\nu_i$. Let $\hat{\psi}^{CF}_i$ denote the cross-fitted one-step estimator for $\psi_i$ and assume that $\norm{\varphi(\hat{\nu}_{i,-k}) - \varphi(\nu_i)} = o_p(1)$ for each $k$ and $i$. Furthermore, we assume that
    $$
    \hat{\psi}_i^{CF} - \psi_i = \mathbb{P}_n \tilde{\psi}_i + o_p(n^{-1/2}), \quad i = 1,2.
    $$
    Let $\Psi = \frac{\psi_1}{\psi_2}$ and denote the cross-fitted estimator $\hat{\Psi}^{CF} = \frac{\hat{\psi}_1^{CF}}{\hat{\psi}_2^{CF}}$. Define
    $$
    \tilde{\psi}(\psi_1, \psi_2, \nu_1, \nu_2) = \frac{1}{\psi_2}\left(\varphi_1(\nu_1) - \psi_1 - \frac{\psi_1}{\psi_2} (\varphi_2(\nu_2) - \psi_2 ) \right).
    $$
    Then
    \begin{align}\label{lem:SECF1}
    \hat{\Psi}^{CF} - \Psi = \mathbb{P}_n\tilde{\psi}(\psi_1, \psi_2, \nu_1, \nu_2) + o_p(n^{-1/2})    
    \end{align}
    and we have the following consistency results for the cross-fitted variance estimators:
    \begin{align}\label{lem:SECF2}
    \hat{\sigma}_{\psi_i}^{2,CF} &= \sum_{k=1}^K \frac{n_k}{n}\mathbb{P}_n^k \tilde{\psi}_i(\hat{\psi}_i^{CF}, \hat{\nu}_{i, -k})^2 \overset{P}{\longrightarrow} P\tilde{\psi}_i(\psi_i, \nu_i)^2 \\
    \label{lem:SECF3}
    \hat{\sigma}_{\Psi}^{2,CF} &= \sum_{k=1}^K \frac{n_k}{n}\mathbb{P}_n^k \tilde{\psi}(\hat{\psi}_1^{CF}, \hat{\psi}_2^{CF}, \hat{\nu}_{1, -k}, \hat{\nu}_{2, -k})^2 \overset{P}{\longrightarrow} P\tilde{\psi}(\psi_1, \psi_2, \nu_1, \nu_2)^2.
    \end{align}
\end{lemma}
\begin{proof}
    See Appendix \ref{app:consVar}.
\end{proof}

We will now briefly summarize the construction given in \textcite{westling} concerning the nuisance parameter estimator $\hat{\tau}_d$.

For estimators $\hat{\Lambda}, \hat{\Lambda}_c, \hat{\pi}$ following assumption \ref{assA}, we construct the following cross-fitted estimator 
\begin{equation}
    \hat{\tau}_d^{CF} = \sum_{i = k}^K \frac{n_k}{n} \mathbb{P}_n^k \hat{\varphi}_{-k} = \frac{1}{n}\sum_{k = 1}^K \sum_{i \in \mathcal{T}_k} \hat{\varphi}_{-k}(O_i). \nonumber
\end{equation}
Theorem 3 in \cite{westling} gives that $\hat{\tau}^{CF}_d$ is asymptotically linear with influence function given by $\tilde{\psi}_{\tau_d} = \varphi - \tau_d$. Now, let $X_n = n^{1/2}(\hat{\tau}_d^{CF} - \tau_d)$. Then $X_n \rightsquigarrow X$ with $X \sim \mathcal{N}(0,P\tilde{\psi}_{\tau_d}^2)$ and Prohorov's theorem (\cite{vaart} theorem 2.4) gives $\norm{X_n}=O_p(1)$, and hence, $\norm{\hat{\tau}_d^{CF} - \tau_d} = n^{-1/2}\norm{X_n} \overset{P}{\rightarrow} 0$. Take $\frac{1}{4} < \delta < \frac{1}{2}$. Then $n^\delta \norm{\hat{\tau}_d^{CF} - \tau_d} = n^{-\epsilon} \norm{X_n} \overset{P}{\rightarrow} 0$, $\epsilon > 0$, and hence $\norm{\hat{\tau}_d^{CF} - \tau_d} = o_p(n^{-\delta})$. Since $\delta > \frac{1}{4}$ the rate assumption on $\hat{\tau}_d$ in Corollary \ref{cor:psiCF} is fulfilled.

Notice, however, that the estimation of $\hat{\Theta}_d^{CF}$ requires estimation of $\hat{\tau}_{d,-k}$ for each split $k$. Thus, in order to use the convergence rate results above, we need to perform a nested type of cross-fitting, such that $\hat{\tau}_{d, -k}$ is the cross-fitted estimator above but estimated using data in $\mathcal{V}_{-k}$ instead of the entire data. For estimation of $\Theta_d$ we have the following procedure:

\begin{enumerate}
    \item Split the data uniformly at random into $K_1$ subsamples $\mathcal{V}_k$, $k=1,\ldots, K_1$, such that $\mathcal{O} = \dot{\cup}_{k=1}^{K_1} \mathcal{V}_k.$
    \item for each $k$ estimate $\hat{\Lambda}, \hat{\Lambda}_c, \hat{\pi}$ using data in $\mathcal{V}_{-k}$ and obtain $\hat{\tau}_{-k}$ and $\hat{\varphi}_{-k}$. For $\tau_{d,-k}$-estimation: 
    \begin{enumerate}
        \item Split $\mathcal{V}_{-k}$ into $K_2$ subsamples $\mathcal{V}^k_i$, $i=1,\ldots, K_2$, such that $\mathcal{V}_{-k} = \dot{\cup}_{i=1}^{K_2} \mathcal{V}^k_i.$
        \item For each $i = 1,\ldots, K_2$ estimate $\hat{\Lambda}, \hat{\Lambda}_c, \hat{\pi}$ using data in $\mathcal{V}^k_{-i}$ and obtain $\hat{\varphi}^k_{-i}$.
        \item Obtain the $k'th$ estimate $\hat{\tau}_{d,-k}^{CF} = \frac{1}{n_{-k}}\sum_{i=1}^{K_2} \sum_{O_j \in \mathcal{V}_i^k} \hat{\varphi}^k_{-i}(O_j)$, where $n_{-k}$ is the number of observations in $\mathcal{V}_{-k}$.
    \end{enumerate}
     \item Obtain the estimate 
     \begin{equation}
    \hat{\Theta}_d^{CF} = \frac{1}{n}\sum_{k = 1}^K \sum_{i \in \mathcal{T}_k} \left\{ (\hat{\varphi}_{-k}(O_i)-\hat{\tau}_{d, -k}^{CF})^2 - (\hat{\varphi}_{-k}(O_i)-\hat{\tau}_{-k}(X_i))^2 \right\} \nonumber
    \end{equation}
\end{enumerate}
Using this estimating scheme, the convergence rate assumption on $\hat{\tau}^{CF}_d$ is automatically fulfilled by assumption \ref{assA} and Corollary \ref{cor:psiCF} applies without further restrictions on $\hat{\tau}_d$.

\section{On estimation of logit transformation of $\Psi_l$}
\label{App-logit}
The target parameter $\Psi_l$ is restricted to $[0,1)$, but the estimator $\hat{\Psi}_l^{CF}$ is unrestricted, which in practice can result in parameter estimates that are outside the range $[0,1)$, or confidence intervals that contain either 0 or 1. To combat this issue, we construct a cross-fitted one-step estimator of the transformed parameter $\text{logit}(\Psi_l)$. We note that given initial estimators, $\hat{\Theta}_l^{CF}$, $\hat{\Theta}_d^{CF}$, $\hat{\Theta}_{l}^0$, $\hat{\Theta}_d^0$, where $\hat{\Theta}_{l}^0$ and $\hat{\Theta}_d^0$ are plug-in estimators, the construction is directly given in \textcite{hines} and in Theorem 4 in their Appendix, they give additional conditions on the plug-in estimators under which the estimator of $\text{logit}(\Psi_l)$ is asymptotically linear. Hence, we will only sketch the construction, and refer to \textcite{hines} for a derivation of the asymptotic results. \\ \\
Define the transformed target parameter $\zeta_l(P) \equiv \text{logit}(\Psi_l(P))$. The efficient influence function of $\zeta(P)$ is given by 
$$
\tilde{\psi}_{\zeta_l} = \frac{\tilde{\psi}_{\Psi_l}}{\Psi_l(1-\Psi_l)}
$$
by the chain rule. Given the plug-in estimators, define $\hat{\Psi}_l^0 = \frac{\hat{\Theta}_l^0}{\hat{\Theta}_d^0}$. The cross-fitted one-step estimator is then given by 
$$
\hat{\zeta}_l^{CF} = \sum_{k=1}^K \frac{n_k}{n} \hat{\zeta}_{l,k}
$$
where 
$$
\hat{\zeta}_{l,k} = \text{logit}(\hat{\Psi}_{l,-k}^0) + \mathbb{P}_n^k \hat{\tilde{\psi}}_{\zeta_l, -k} = \text{logit}(\hat{\Psi}_{l,-k}^0) + \mathbb{P}_n^k \frac{\hat{\tilde{\psi}}_{\Psi_l,-k}}{\hat{\Psi}_{l,-k}^0(1-\hat{\Psi}_{l,-k}^0)}
$$
with $\hat{\Psi}_{l,-k}^0$ being the plug-in estimator obtained from the sample $\mathcal{V}_{-k}$ and $\hat{\tilde{\psi}}_{\Psi_l,-k}$ being the estimator of the EIF $\tilde{\psi}_{\Psi_l}$ derived from nuisance estimators obtained from $\mathcal{V}_{-k}$. As noted in \textcite{hines}, the estimator, $\hat{\zeta}_l^{CF}$, can be written in terms of the already established estimators with
$$
\hat{\zeta}_{l,k} = \text{logit}(\hat{\Psi}_{l,-k}^0) + \frac{\hat{\Theta}_{d,-k}}{\hat{\Theta}_{d,-k}^{0}} \frac{(\hat{\Psi}_{l,-k} - \hat{\Psi}_{l,-k}^{0})}{\hat{\Psi}_{l,-k}^0(1-\hat{\Psi}_{l,-k}^0)},
$$
where $\hat{\Theta}_{d,-k}$ and $\hat{\Psi}_{l,-k}$ are the one-step estimators obtained from $\mathcal{V}_{-k}$.

Under the same assumptions as in Corollary \ref{cor:psiCF}, Theorem 4 in \textcite{hines} gives that $\hat{\zeta}_l^{CF}$ is asymptotically linear with $\tilde{\psi}_{\zeta_l}$ as its influence function. In practice, this can be leverage to obtain an estimate of $\Psi_l$, where the estimate and the confidence interval are restricted to the interval $[0,1)$, by an expit-transformation of $\hat{\zeta}_l^{CF}$ and the corresponding confidence interval. By the delta method, the back-transformed estimator then shares the same asymptotic properties as given in Corollary \ref{cor:psiCF}.


\section{Proofs of asymptotic results}\label{sec:appAN}
The proofs of Theorem \ref{asTheta} and \ref{asGammaAndChi} follow the same recipe. The strategy is based on an expansion of the target parameter estimator in question as described in \textcite{kennedydouble}, and we will give a short recap of the general idea. In the following, let $\psi(P)$ denote a generic target parameter with EIF given by $\tilde{\psi}_P = \varphi_P - \psi(P)$, i.e., the EIF is linear in $\psi(P)$. Define the corresponding cross-fitted one-step estimator
$$
\hat{\psi}^{CF} = \sum_{k=1}^K \frac{n_k}{n}\mathbb{P}_n^k \varphi_{\hat{P}_k},
$$
where $\hat{P}_{-k}$ is an estimate of $P$ obtained from $\mathcal{V}_{-k}$. Consider the following expansion of $\mathbb{P}_n^k \varphi_{\hat{P}_{-k}}$.
\begin{align*}
    \mathbb{P}_n^k \varphi_{\hat{P}_{-k}} = \mathbb{P}_n^k \tilde{\psi} + (\mathbb{P}_n^k - P) (\varphi_{\hat{P}_{-k}} - \varphi_P) + P\varphi_{\hat{P}_{-k}} - \psi(P). 
\end{align*}
Given the above expansion, we obtain the decomposition
\begin{align*}
    \hat{\psi}^{CF} - \psi(P) = \mathbb{P}_n \tilde{\psi} + \underbrace{\sum_{k=1}^K \frac{n_k}{n} \mathbb{P}_n^k (\mathbb{P}_n^k - P) (\varphi_{\hat{P}_{-k}} - \varphi_P)}_{\text{empirical process term}} + \underbrace{\sum_{k=1}^K \frac{n_k}{n} P(\varphi_{\hat{P}_{-k}} - \psi(P))}_{\text{remainder term}}.
\end{align*}
By Lemma 2 in the supplementary material in \textcite{kennedy2020}, the empirical process term is $o_p(n^{-1/2})$ if $\norm{\varphi_{\hat{P}_{-k}} - \varphi_P} = o_p(1)$ for each $k$. The remainder term is $o_p(n^{-1/2})$ if $P\varphi_{\hat{P}_{-k}} - \psi(P) = o_p(n^{-1/2})$ for each $k$ by the continuous mapping theorem, since $\frac{n_k}{n} \overset{P}{\rightarrow} \frac{1}{K}$. This is essentially Proposition 2 in \textcite{kennedydouble}. Hence, the statements in Theorem \ref{asTheta} and Theorem \ref{asGammaAndChi} follow, if we can show that $\norm{\varphi_{\hat{P}_{-k}} - \varphi_P} = o_p(1)$ and that $P(\varphi_{\hat{P}_{-k}}) - \psi(P) = o_p(n^{-1/2})$ for the corresponding estimators. In the following, we drop the dependence on $k$ to ease notation. \\ \\
We start by stating two results related to the empirical process term and remainder term for the ATE, which will come in handy in the proofs of Theorem \ref{asTheta} and Theorem \ref{asGammaAndChi}. The results are essentially found in \textcite{westling} for the survival function setting (albeit, stated slightly differently), but we repeat them here for completeness and extend them to the RMST setting. 

\begin{alemma}\label{rem:ate}
    Let $\varphi$ be given as in \eqref{eq:varphi_surv_a} for the survival function setting and as in \eqref{eq:varphi_RMST_a} in the RMST setting. Under assumption \ref{ass1} and \ref{ass4}, $P\{\varphi(\hat{\nu}) - \tau \} = o_p(n^{-1/2})$.
\end{alemma}

\begin{proof}
    The result for the survival function setting is proved in \textcite{westling} and \textcite{heleneFrank}. We include the computations for completeness and extend it to the RMST case.
    \\ \\
    \textit{Survival Case} \\ \\
Let $\tau_a(x) = S(t\mid A=a, X=x)$ such that 
$$
\varphi(\hat{\nu}) - \tau = \varphi_{1}(\hat{\nu}) - \tau_1 - (\varphi_{0}(\hat{\nu}) - \tau_0),
$$
where $\varphi_a(\hat{\nu})$ is given in \eqref{eq:varphi_surv_a}. Thus, to bound $P\{\varphi(\hat{\nu}) - \tau \}$, we only need to derive a bound for $P\{\varphi_a(\hat{\nu}) - \tau_a\}$. In the following we consider the nuisance estimates fixed.
\begin{align} \label{remainderSurv}
    &\E\{\varphi_a(\hat{\nu})(O) - \tau_a(X)\} \nonumber \\
    =& \E\{E(\varphi_a(\hat{\nu})(O) - \tau_a(X)\mid A, X)\} \nonumber \\
    =& \E\left\{\hat{S}(t\mid a, X) - S(t\mid a, X) - \frac{\mathbb{1}(A=a)\hat{S}(t\mid A, X)} {\hat{\pi}(a\mid X)} \right. \nonumber \\
    & \left. \times \left( \int_0^t \frac{\E(\dd N(s)\mid A=a, X)}{\hat{S}(s\mid A, X) \hat{S}_c(s\mid A,X)} - \int_0^t \frac{\E(\mathbb{1}(\tilde{T}\geq t)\mid A=a, X)\dd \hat{\Lambda}(s\mid A,X)}{\hat{S}(s\mid A, X) \hat{S}_c(s\mid A,X)}\right)  \right\} \nonumber \\
    =& \E\left\{\hat{S}(t\mid a, X) - S(t\mid a, X) \right. \nonumber \\
    &- \left. \frac{{\pi}(a\mid X)\hat{S}(t\mid a, X)}{\hat{\pi}(a\mid X)}  \int_0^t \frac{S(s\mid a, X) S_c(s\mid a,X)}{\hat{S}(s\mid a, X) \hat{S}_c(s\mid a,X)} \dd \left[\Lambda(s\mid a, X) - \hat{\Lambda}(s\mid a, X)\right]   \right\}.
\end{align}
Now, consider the survival function difference above. Using Duhamel's equation (\cite{gill}) we have
\begin{equation}
    \label{duhamel}
    \hat{S}(t\mid a, x) - S(t\mid a, x) = \int_0^t\frac{S(s\mid a,x)}{\hat{S}(s\mid a, x)}\dd \left[ \Lambda(s\mid a, x) - \hat{\Lambda}(s\mid a, x) \right] \hat{S}(t\mid a, x). \nonumber
\end{equation}
Plugging this into \eqref{remainderSurv} yields
\begin{align}
     &\E\{\hat{\varphi}_a(O) - \tau_a(X)\} \nonumber \\
     =& \E \left\{ \int_0^t \left(1 - \frac{\pi(a\mid X)S_c(s\mid a, X)}{\hat{\pi}(a\mid X)\hat{S}_c(s\mid a, X)} \right) \frac{S(s\mid a, X)}{ \hat{S}(s\mid a,X)} \hat{S}(t\mid a,X) \dd \left[\Lambda(s\mid a, X) - \hat{\Lambda}(s\mid a, X)\right]   \right\} \nonumber \\
     =& o_p(n^{-1/2}) \nonumber
\end{align}
by assumption \ref{ass4}. \\ \\
\textit{RMST case} \\ \\
Now consider the case where $\tau(x) = \int_0^{t} S(s\mid 1,x) \dd s-\int_0^{t} S(s\mid 0,x) \dd s$. As in the survival setting we define
$$
\tau_a(x) = \int_0^{t} S(s\mid a,x) \dd s
$$
and 
$$\varphi_a(O) = \tau_a(X) - \frac{\mathbb{1}(A=a)}{\pi(a\mid X)} \int_0^{t}\frac{H(u,t,A,X)}{S(s \mid A, X)S_C(s \mid A, X)} dM(s \mid A, X).
$$
By derivations analogous to \eqref{remainderSurv} we have
\begin{align*}
    &\E\{\hat{\varphi}_a(O) - \tau(X)\} \nonumber \\
    &= \E\{E(\hat{\varphi}_a(O) - \tau(X)\mid A, X)\} \nonumber \\
    &= \E\left\{ \int_0^{t} \hat{S}(s\mid a, X) - S(s\mid a, X)\dd s \right. \nonumber \\
    &- \left. \frac{\pi(a\mid X)}{\hat{\pi}(a\mid X)}  \int_0^{t} \frac{\hat{H}(s,t \ a, x) S(s\mid a, X) S_c(s\mid a,X)}{\hat{S}(s\mid a, X) \hat{S}_c(s\mid a,X)} \dd \left[\Lambda(s\mid a, X) - \hat{\Lambda}(s\mid a, X)\right]   \right\} \nonumber \\
    &= \E\left\{ \int_0^{t} \int_0^s \frac{S(u\mid a, X)}{\hat{S}(u\mid a, X)}\dd \left[\Lambda(u\mid a, X) - \hat{\Lambda}(u\mid a, X) \right] \hat{S}(s\mid a, X) \dd s \right. \nonumber \\
    &- \left. \frac{\pi(a\mid X)}{\hat{\pi}(a\mid X)}  \int_0^{t} \frac{\hat{H}(s,{t} \ a, x) S(s\mid a, X) S_c(s\mid a,X)}{\hat{S}(s\mid a, X) \hat{S}_c(s\mid a,X)} \dd \left[\Lambda(s\mid a, X) - \hat{\Lambda}(s\mid a, X)\right]   \right\} \nonumber \\
    &= \E\left\{ \int_0^{t}  \frac{\hat{H}(u,{t}\mid a, X) S(u\mid a, X)}{\hat{S}(u\mid a, X)}\dd \left[\Lambda(u\mid a, X) - \hat{\Lambda}(u\mid a, X) \right] \right. \nonumber \\
    &- \left. \frac{\pi(a\mid X)}{\hat{\pi}(a\mid X)}  \int_0^{t} \frac{\hat{H}(s,{t} \ a, x) S(s\mid a, X) S_c(s\mid a,X)}{\hat{S}(s\mid a, X) \hat{S}_c(s\mid a,X)} \dd \left[\Lambda(s\mid a, X) - \hat{\Lambda}(s\mid a, X)\right]   \right\} \nonumber \\
    &= \E\left\{ \int_0^{t}  \frac{\hat{H}(s,{t}\mid a, X) S(s\mid a, X)}{\hat{S}(s\mid a, X)}\left(1 - \frac{S_c(s\mid a, X)\pi(a\mid X)}{\hat{S}_c(s\mid a, X)\hat{\pi}(a\mid X)} \right) \dd \left[\Lambda(u\mid a, X) - \hat{\Lambda}(u\mid a, X) \right] \right\} \\
    &= o_p(n^{-1/2})
\end{align*}
by assumption \ref{ass4}. The third equality follows from Duhamel's equation. 
\end{proof}

Next we have a lemma, which is essentially given in \textcite{westling} (Lemma 3 in their supplementary material) in the survival function setting, though our assumptions are stated slightly different. We include the proof for completeness and extend the result to the RMST setting. 
\begin{alemma}\label{emp:ate}
    Let $\varphi$ be given as in \eqref{eq:varphi_surv_a} for the survival function setting and as in \eqref{eq:varphi_RMST_a} for the RMST setting. Under assumption \ref{ass1}, \ref{ass2} and \ref{ass5} it holds that $\norm{\varphi(\hat{\nu}) - \varphi(\nu)} = o_p(1)$.
\end{alemma}

\begin{proof}
    Observe that 
$$
\norm{\hat{\varphi} - \varphi} \leq \norm{\hat{\varphi}_1 - \varphi_1} + \norm{\hat{\varphi}_0 - \varphi_0}
$$
so that we only need to focus on $\norm{\hat{\varphi}_a - \varphi_a}$. We start deriving a bound in the survival setting and then proceed to the RMST setting. \\ \\
\textit{Survival function setting} \\ \\
Consider the decomposition
\begin{align} \label{eq:empSurvDecomp}
    &\hat{\varphi}_a(O) - \varphi_a(O) \nonumber \\
    =& (\hat{\tau}_a(X) - \tau_a(X)) - \left(\frac{\mathbb{1}(A=a)}{\hat{\pi}(a\mid X)}\int_0^t\frac{\hat{S}(t\mid A, X)}{\hat{S}(s\mid A, X)\hat{S}_c(s\mid A, X)} \dd \hat{M}(s\mid A, X) \right. \nonumber \\
    &\quad \quad \quad - \left. \frac{\mathbb{1}(A=a)}{\pi(a\mid X)}\int_0^t \frac{S(t\mid A, X)}{S(s\mid A, X)S_c(s\mid A, X)} \dd M(s\mid A, X) \right) \nonumber \\
    =& (\hat{\tau}_a(X) - \tau_a(X)) \nonumber \\
    &- \mathbb{1}(A=a)\left(\int_0^t \frac{\hat{S}(t\mid a, X)}{\hat{S}(s\mid a, X)\hat{g}(s\mid a, X)} - \frac{S(t\mid a, X)}{S(s\mid a, X)g(s\mid a, X)} \dd N(s)\right) \nonumber \\
    &- \mathbb{1}(A=a)\left(\int_0^t \frac{\hat{S}(t\mid a, X)\mathbb{1}(\tilde{T}\geq s)}{\hat{S}(s\mid a, X)\hat{g}(s\mid a, X)}\hat{\Lambda}(\dd s \mid a, X) - \int_0^t \frac{S(t\mid a, X)\mathbb{1}(\tilde{T}\geq s)}{S(s \mid a, X)g(s\mid a, X)} \Lambda(\dd s \mid a, X)\right).
\end{align}
We need to bound each term in the above expression (separated by parentheses) individually. For the first term, \ref{ass2} gives that $\norm{\hat{\tau} - \tau} = o_p(1)$. For the second term we have for $\text{almost all} \ x$
\begin{align*}
     &\mathbb{1}(A=a)\left(\int_0^t \frac{\hat{S}(t\mid a, x)}{\hat{S}(s\mid a, x)}\left( \frac{1}{\hat{g}(s\mid a, x)} - \frac{1}{g(s\mid a, x)} \right) \right. \\
     & \left. \phantom{- \frac{1}{g(s\mid a, x)}} - \frac{1}{g(s\mid a, x)}\left( \frac{\hat{S}(t\mid a, x)}{\hat{S}(s\mid a, x)} - \frac{S(t\mid a, x)}{S(s\mid a, x)}  \right) \dd N(s)\right)^2\\
    \leq & \ 2\left(\int_0^t  \frac{1}{\hat{g}(s\mid a, x)} - \frac{1}{g(s\mid a, x)} \dd N(s)\right)^2 +  2\eta^{-2}\left( \int_0^t  \frac{\hat{S}(t\mid a, x)}{\hat{S}(s\mid a, x)} - \frac{S(t\mid a, x)}{S(s\mid a, x)}  \dd N(s)\right)^2 \\
    \leq & \ 2 \eta^{-4} \left\{\sup_{s\leq t}\left| \hat{g}(s\mid a, x) - g(s\mid a, x)\right|\right\}^2 + 4\eta^{-2}\left(\int_0^t \frac{1}{\hat{S}(s\mid a,x)}\left(\hat{S}(t\mid a, x) - S(t\mid a, x) \right)\dd N(s) \right)^2 \\
    & \quad \quad \quad \quad \quad \quad \quad \quad + 4\eta^{-2} \left(\int_0^t S(t\mid a, x)\left(\frac{1}{\hat{S}(s\mid a, x)} - \frac{1}{S(s\mid a,x)} \right)\dd N(s) \right)^2 \\
    \leq & \ 2 \eta^{-4} \left\{\sup_{s\leq t}\left| \hat{g}(s\mid a, x) - g(s\mid a, x)\right|\right\}^2 + 4 \eta^{-4} \left(\hat{S}(t\mid a, x) - S(t\mid a, x) \right)^2 \\
    &+ 4\eta^{-4} \left\{\sup_{s\leq t}\left| \hat{S}(s\mid a, x) - S(s\mid a, x)\right|\right\}^2 \\
    \leq & \ 2\eta^{-4} \left\{\sup_{s\leq t}\left| \hat{g}(s\mid a, x) - g(s\mid a, x)\right|\right\}^2 + 8\eta^{-4} \left\{\sup_{s\leq t}\left| \hat{S}(s\mid a, x) - S(s\mid a, x)\right|\right\}^2,
\end{align*}
which, together with \ref{ass5}, shows that the $L_2(P)$ norm of the second term converges in probability to zero. For the third term in \eqref{eq:empSurvDecomp}, we use the same technique as described in the proof of Lemma 3 in \textcite{westling}. It is included here for completeness, and extended to the RMST setting in the following part of the proof. let $K_1(s,t\mid a,x) = \frac{S(t\mid a,x)}{S(s\mid a,x)}$ and let $\hat{K}_1$ be defined accordingly with $\hat{S}$ in place of $S$. The backwards equation (\cite{gill}, Theorem 5) gives that for almost all $x$
$$K_1(s,t\mid a, x) = 1 - \int_s^t \frac{S(t\mid a,x)}{S(w\mid a, x)}\Lambda(\dd w).$$
Hence, $K_1(\dd s, t \mid a, x) = \frac{S(t\mid a,x)}{S(s\mid a, x)}\Lambda(\dd s)$ and similarly, $\hat{K}_1(\dd s, t\mid a, x) = \frac{\hat{S}(t\mid a,x)}{\hat{S}(s\mid a, x)}\hat{\Lambda}(\dd s)$. This result allows us to write the third term in \eqref{eq:empSurvDecomp} for almost all $x$ as (dropping the indicator $\mathbb{1}(A=a)$ since it disappears in the bound anyway)
\begin{align}\label{eq:empSurvDecompThirdTerm}
    &\left(\int_0^t \frac{\hat{S}(t\mid a, X)\mathbb{1}(\tilde{T}\geq s)}{\hat{S}(s\mid a, X)\hat{g}(s\mid a, X)}\hat{\Lambda}(\dd s \mid a, X) - \int_0^t \frac{S(t\mid a, X)\mathbb{1}(\tilde{T}\geq s)}{S(s \mid a, X)g(s\mid a, X)} \Lambda(\dd s \mid a, X)\right)^2 \nonumber \\
    = & \left(\int_0^{t\wedge \tilde{T}} \frac{1}{\hat{g}(s\mid a, X)}\hat{K}_1(\dd s \mid a, X) - \int_0^{t\wedge \tilde{T}} \frac{1}{g(s\mid a, X)} K_1(\dd s \mid a, X)\right)^2 \nonumber \\
    = & \left(\int_0^{t\wedge \tilde{T}} \left( \frac{1}{\hat{g}(s\mid a, X)} - \frac{1}{g(s\mid a, X)} \right) \hat{K}_1(\dd s \mid a, X) \right. \nonumber \\
    & \left. \phantom{\int_0^{t\wedge \tilde{T}} \int_0^{t\wedge \tilde{T}} \int_0^{t\wedge \tilde{T}} } + \int_0^{t\wedge \tilde{T}} \frac{1}{g(s\mid a, X)} \left[ \hat{K}_1(\dd s \mid a, X) - K_1(\dd s \mid a, X) \right]\right)^2.
\end{align}
Thus, if we can show that the two integrals in the above display are consistent in $L_2(P)$-norm, it follows that \eqref{eq:empSurvDecomp} is is $o_p(1)$, which completes the proof. For the first integral in \eqref{eq:empSurvDecompThirdTerm}, we have for almost all $x$
\begin{align*}
    & \int_0^{t\wedge \tilde{T}} \left( \frac{1}{\hat{g}(s\mid a, X)} - \frac{1}{g(s\mid a, X)} \right) \hat{K}_1(\dd s \mid a, X) \\
    =& \int_0^{t\wedge \tilde{T}} \left( \frac{1}{\hat{g}(s\mid a, X)} - \frac{1}{g(s\mid a, X)} \right) \frac{\hat{S}(t\mid a, X)}{\hat{S}(s\mid a, X)} \hat{\Lambda}(\dd s \mid a, X) \\
   \leq & \sup_{s\leq t} \left| \frac{\hat{S}(t\mid a,x)}{\hat{S}(s\mid a, x)} \right| \sup_{s\leq t}\left| \frac{1}{\hat{g}(s\mid a, X)} - \frac{1}{g(s\mid a, X)} \right| \hat{\Lambda}(t) \\
   \leq & \left| \log \eta \right| \eta^{-2}  \sup_{s\leq t} \left| \hat{g}(s\mid a,x) - g(s\mid a,x) \right|,
\end{align*}
where we have used $\frac{\hat{S}(t\mid a,x)}{\hat{S}(s\mid a,x)}\leq 1$ together with assumption \ref{ass1}. Assumption \ref{ass5} then shows that the first integral in \eqref{eq:empSurvDecompThirdTerm} is $o_p(1)$ in $L_2(P)$-norm. Using integration by parts, we can bound the second integral in \eqref{eq:empSurvDecompThirdTerm}. For almost all $x$ we have
\begin{align*}
    &\int_0^{t\wedge \tilde{T}} \frac{1}{g(s\mid a, x)} \left[ \hat{K}_1(\dd s \mid a, x) - K_1(\dd s \mid a, x) \right] \\
    =& \frac{1}{g(t\mid a,x)}\left[\hat{K}_1(t\wedge \tilde{T}, t\mid a,x) -  K_1(t\wedge \tilde{T}, t\mid a,x) \right] - \frac{1}{g(0\mid a,x)}\left[\hat{K}_1(0\mid a,x) -  K_1(0, t\mid a,x) \right] \\
    & \phantom{\frac{1}{g(t\mid a,x)}} - \int_0^{t\wedge \tilde{T}} \left[\hat{K}_1(s\mid a,x) -  K_1(s, t\mid a,x) \right] \left(\frac{1}{g}\right)(\dd s \mid a,x) \\
    \leq & 3 \eta^{-1} \sup_{s\leq t} \left| \hat{K}_1(s\mid a,x) - K_1(s\mid a,x)\right| \\
    \leq & C \sup_{s\leq t} \left| \hat{\Lambda}(s\mid a,x) - \Lambda(s\mid a,x) \right| 
\end{align*}
for some $C>0$, where the last inequality follows from the mean value theorem. By \ref{ass5}, the above expression is $o_p(1)$ in $L_2(P)$-norm and it follows that \eqref{eq:empSurvDecompThirdTerm} is $o_p(1)$ in $L_2(P)$-norm, which completes the proof for the survival function setting.  \\ \\ 
\textit{RMST setting} \\ \\
Consider the decomposition
\begin{align}\label{eq:empRMSTDecomp}
    &\hat{\varphi}_a(O) - \varphi_a(O) \nonumber \\
    =& (\hat{\tau}_a(X) - \tau_a(X)) \nonumber \\
    &- \left(\frac{\mathbb{1}(A=a)}{\hat{\pi}(a\mid X)}\int_0^t\frac{\hat{H}(s, t\mid A, X)}{\hat{S}(s\mid A, X)\hat{S}_c(s\mid A, X)} \dd \hat{M}(s\mid A, X) \right. \nonumber \\
    &\quad \quad \quad - \left. \frac{\mathbb{1}(A=a)}{\pi(a\mid X)}\int_0^t \frac{H(s, t\mid A, x)}{S(s\mid A, X)S_c(s\mid A, X)} \dd M(s\mid A, X) \right) \nonumber \\
    =& \left( \int_0^t \hat{S}(s\mid a, X) - S(s\mid a, X) \dd u \right)  \nonumber \\
    &- \mathbb{1}(A=a)\left(\int_0^t \frac{\hat{H}(s, t\mid a, X)}{\hat{S}(s\mid a, X)\hat{g}(s\mid a, X)} - \frac{H(s,t\mid a, X)}{S(s\mid a, X)g(s\mid a, X)} \dd N(s)\right) \nonumber \\
    &- \mathbb{1}(A=a)\left(\int_0^t \frac{\hat{H}(s,t\mid a, X)\mathbb{1}(\tilde{T}\geq s)}{\hat{S}(s\mid a, X)\hat{g}(s\mid a, X)}\hat{\Lambda}(\dd s \mid a, X) - \int_0^t \frac{H(s,t\mid a, X)\mathbb{1}(\tilde{T}\geq s)}{S(s \mid a, X)g(s\mid a, X)} \Lambda(\dd s \mid a, X)\right).
\end{align}
Since the structure is similar to the survival function setting, the arguments will be similar too. We will again consider each term in turn. For first term we have 
\begin{align*}
    \int_0^t \hat{S}(s\mid a, X) - S(s\mid a, X) \dd u \ \leq \ t\sup_{s<t}\left| \hat{S}(s\mid a, X) - S(s\mid a, X) \right|,
\end{align*}
which is $o_p(1)$ in $L_2(P)$ by assumption \ref{ass5}. For the second term, note that $H(s,t\mid a, X) \leq t$, and
\begin{align*}
    \left| \hat{H}(s,t\mid a, X) - H(s,t\mid a, X) \right|  =&  \int_s^t \hat{S}(s\mid a, X) - S(s\mid a, X) \dd u \\
    \leq &  t\sup_{s<t}\left| \hat{S}(s\mid a, X) - S(s\mid a, X) \right|,
\end{align*}
and
\begin{align*}
    \frac{H(s,t\mid a, X)}{S(s\mid a, X)} = \int_s^t \frac{S(u\mid a, X)}{S(s\mid a, X)} \dd u \ \leq \ \int_s^t \dd u \ \leq \ t. 
\end{align*}
Hence, replacing $S(t\mid a,X)$ and $\hat{S}(t\mid a,X)$ in the derivations from the survival function setting with $H(s, t\mid a,X)$ and $\hat{H}(s, t\mid a,X)$, gives that the second term in \eqref{eq:empRMSTDecomp} is $o_p(1)$ in $L_2(P)$ by assumption \ref{ass1} and \ref{ass5} by similar arguments as in the survival function setting. For the third term in \eqref{eq:empRMSTDecomp}, we use the same strategy as in the survival function setting as described in \textcite{westling}, but now extended to the RMST setting. Define 
$$
K_2(s,t\mid a,x) = \frac{H(s,t\mid, a,x)}{S(s\mid a,x)} - \int_s^t\dd u,
$$
and let $\hat{K}_2$ be defined accordingly, with $\hat{S}$ in place of $S$ and $\hat{H} = \int \hat{S}\dd s$. Then, by the backward equation (\cite{gill}, Theorem 5)
\begin{align*}
    K_2(s,t\mid a,x) =& \int_s^t \frac{S(u\mid a,x)}{S(s\mid a,x)} \dd u - \int_s^t \dd u \\
    =& \int_s^t \left(1 - \int_s^u \frac{S(u\mid a,x)}{S(w\mid a,x)} \Lambda(\dd w\mid a,x) \right) - \int_s^t \dd u \\
    =& -\int_s^t \int_s^u \frac{S(u\mid a,x)}{S(w\mid a,x)} \Lambda(\dd w\mid a,x) \dd u \\
    =& -\int_s^t \int_w^t \frac{S(u\mid a,x)}{S(w\mid a,x)} \dd u \Lambda(\dd w\mid a,x) \\
    =& - \int_s^t \frac{H(w,t\mid a,x)}{S(w\mid a,x)} \Lambda(\dd w\mid a,x).
\end{align*}
Hence $K_2(\dd s, t\mid a,x) = \frac{H(s,t\mid a,x)}{S(s\mid a,x)} \Lambda(\dd s\mid a,x)$, and similarly for $\hat{K}_2$. Now, we can write the third term in \eqref{eq:empRMSTDecomp} as (again, dropping $\mathbb{1}(A=a)$)
\begin{align*}
    &\left(\int_0^t \frac{\hat{H}(s,t\mid a, X)\mathbb{1}(\tilde{T}\geq s)}{\hat{S}(s\mid a, X)\hat{g}(s\mid a, X)}\hat{\Lambda}(\dd s \mid a, X) - \int_0^t \frac{H(s,t\mid a, X)\mathbb{1}(\tilde{T}\geq s)}{S(s \mid a, X)g(s\mid a, X)} \Lambda(\dd s \mid a, X)\right)^2 \nonumber \\
    = & \left(\int_0^{t\wedge \tilde{T}} \frac{1}{\hat{g}(s\mid a, X)}\hat{K}_2(\dd s \mid a, X) - \int_0^{t\wedge \tilde{T}} \frac{1}{g(s\mid a, X)} K_2(\dd s \mid a, X)\right)^2. \\
    = & \left(\int_0^{t\wedge \tilde{T}} \left( \frac{1}{\hat{g}(s\mid a, X)} - \frac{1}{g(s\mid a, X)} \right) \hat{K}_2(\dd s \mid a, X) \right. \nonumber \\
    & \left. \phantom{\int_0^{t\wedge \tilde{T}} \int_0^{t\wedge \tilde{T}} \int_0^{t\wedge \tilde{T}} } + \int_0^{t\wedge \tilde{T}} \frac{1}{g(s\mid a, X)} \left[ \hat{K}_2(\dd s \mid a, X) - K_2(\dd s \mid a, X) \right]\right)^2.
\end{align*}
Hence, by the same arguments used for in the survival function setting, it follows the third term in \eqref{eq:empRMSTDecomp} is $o_p(1)$ in $L_2(P)$-norm, which concludes the proof.
\end{proof}

\subsection{Proof of Theorem \ref{asTheta}}
\subsubsection{Remainder term}
Let $\phi_{\Theta_l}(\nu)$ be the uncentered version of the EIF $\tilde{\psi}_{\Theta_l}$ at nuisance parameter $\nu$. We will use a decomposition of the remainder term $P\phi_{\Theta_l}(\hat{\nu}) - \Theta_l$ from \textcite{hines}. Observe that $\Theta_l = P\{\tau - \tau_l\}^2$. Then
\begin{align} 
    P\phi_{\Theta_l}(\hat{\nu}) - \Theta_l &= P\{ (\hat{\tau}_l - \varphi(\hat{\nu}))^2 - (\hat{\tau} - \varphi(\hat{\nu}))^2 - (\tau - \tau_l)^2\} \nonumber \\
    &= P\{(\hat{\tau}_l - \varphi(\hat{\nu}))^2 - (\hat{\tau} - \varphi(\hat{\nu}))^2 - \tau^2 - \tau_l^2 + 2\tau_l^2 \} \nonumber \\
    &= P\{ (\hat{\tau}_l - \tau_l)^2 - (\hat{\tau} - \tau)^2 + 2\hat{\tau}_l\tau_l - 2\hat{\tau}\tau + 2(\hat{\tau} - \hat{\tau}_l)\hat{\varphi} \} \nonumber \\
    &= \norm{\hat{\tau}_l - \tau_l}^2 - \norm{\hat{\tau} - \tau}^2 + P\{ 2\hat{\tau}_l\tau_l - 2\hat{\tau}_l\tau + 2(\hat{\tau} - \hat{\tau}_l)(\hat{\varphi} - \tau) \} \nonumber \\
    &= \norm{\hat{\tau}_l - \tau_l}^2 - \norm{\hat{\tau} - \tau}^2 + 2P\{ (\hat{\tau} - \hat{\tau}_l)(\hat{\varphi} - \tau) \} \nonumber \\
    &\leq o_p(n^{-1/2}) + 2KP(\varphi(\hat{\nu}) - \tau) \nonumber \\
    &\leq o_p(n^{-1/2}) \nonumber
\end{align}
for some $K\geq 0$, where the second and fifth equality is due to iterated expectation, the first inequality follows from assumption \ref{ass2} and \ref{ass3} the fact that $\hat{\tau}(X) - \hat{\tau}_l(X)$ is bounded almost surely and the second inequality follows from Lemma \ref{rem:ate}.

\subsubsection{Empirical process term}
We need to show that $\norm{\phi_{\Theta_l}(\hat{\nu}) - \phi_{\Theta_l}(\nu)} = o_p(1)$. Consider the decomposition given in \textcite{hines}
\begin{align*}
    \phi_{\Theta_l}(\hat{\nu}) - \phi_{\Theta_l}(\nu) &= (\hat{\tau}_l - \tau_l)^2 - (\hat{\tau} - \tau)^2 + 2(\varphi - \tau_l)(\tau_l - \hat{\tau}_l) - 2(\varphi - \tau)(\tau - \hat{\tau}) + 2(\hat{\varphi}- \varphi)(\hat{\tau} - \hat{\tau}_l) \\
    &= \sum_{i=1}^5 a_i 
\end{align*}
so that $\norm{\phi(\hat{\nu}) - \phi(\nu)} \leq \sum_{i=1}^5\norm{a_i}$. We will treat each term separately. \\ \\
$(a_1)$: From \ref{ass2} we have that $\norm{(\hat{\tau}_l - \tau_l)^2} = o_p(1)$ since $x \mapsto x^2$ is continuous. \\ \\
$(a_2)$:
Same argument as in $(a_1)$. \\ \\
$(a_3)$: Consider the survival case. Then the following bound holds almost surely:
\begin{align*}
    &(\varphi(O) - \tau_l(x))^2 \\
    =& \left(\tau(x) - \tau_l(x) - \left(\frac{\mathbb{1}(A=1)}{\pi(1\mid x)} - \frac{\mathbb{1}(A=0)}{\pi(0\mid x)} \right)\int_0^t \frac{S(t\mid A, x)}{S(s\mid A, x)S_c(s\mid A, x)} \dd M(s\mid A, x) \right)^2 \\ 
    \leq & 2(\tau(x) - \tau_l(x))^2 + 2\left( \eta^{-1} \int_0^t \frac{S(t\mid A, x)}{S(s\mid A, x)} \dd N(s) - \eta^{-1}\int_0^{t} \frac{S(t\mid A, x)}{S(s\mid A, x)}\mathbb{1}(\tilde{T} \geq s) \dd \Lambda(s\mid A, x) \right)^2 \\
     \leq & 2K^2 + 4(\eta^{-1})^2 + 4\left(\eta^{-1}\int_0^{t} \frac{S(t\mid A, x)}{S(s\mid A, x)} \dd \Lambda(s\mid A, x) \right)^2 \\ 
    =& 2K^2 + 4(\eta^{-1})^2 + 4(\eta^{-1})^2\left(1 - S(t\mid a, x) \right)^2 \\
    \leq & 2K^2 + 8\eta^{-2}
\end{align*}
where the first inequality follows \ref{ass1}, the second inequality from $\tau(x) - \tau_l(x)$ being bounded almost surely and that $\frac{S(t\mid a,x)}{S(s\mid a,x)} \leq 1$, $s \leq t$. The second equality is due the backward equation (theorem 5, \cite{gill}), realising that $\frac{S(t\mid a,x)}{S(s\mid a,x)} = \prodi_{]s, t]}\left(1 - d\Lambda(u\mid a, x)\right)$. 

Now consider the RMST setting. Start by observing that 
$$\frac{H(s,t\mid a, x)}{S(s\mid a, x)} = \int_s^t\frac{S(u \mid a, x)}{S(s\mid a, x)} \dd u \leq t - s \leq t$$ 
and 
\begin{align*}
    \int_0^t \frac{H(s,t\mid a, x)}{S(s\mid a, x)} \dd \Lambda(s\mid, a, x) = \int_0^t \int_0^u \frac{S(u \mid a, x)}{S(s\mid a, x)} \dd \Lambda(s\mid, a, x) \dd u = \int_0^t S(u\mid a,x) - 1 \dd u
\end{align*}
by the backward equation. Then by calculations similar to the once from the survival setting we have 
\begin{align*}
    &(\varphi(O) - \tau_l(x))^2 \\
    \leq & 2(\tau(x) - \tau_l(x))^2 + 2\left( \eta^{-1} \int_0^t \frac{H(s, t\mid A, x)}{S(s\mid A, x)} \dd N(s) - \eta^{-1}\int_0^{t} \frac{H(s,t\mid A, x)}{S(s\mid A, x)}\mathbb{1}(\tilde{T} \geq s) \dd \Lambda(s\mid A, x) \right)^2 \\
     \leq & 2K^2 + 4(\eta^{-1}t)^2 + 4\left(\eta^{-1}\int_0^{t} S(u\mid a, x) - 1 \dd u \right)^2 \\ 
    =& 2K^2 + 4\eta^{-2}t^2 + 4\eta^{-2}(t^2  + t^2) \\
    =& 2K^2 + 12\eta^{-2}t^2.
\end{align*}
Letting $C(\eta, t) = \max\{2K^2 + 8\eta^{-2}, 2K^2 + 12\eta^{-2}t^2\}$, \ref{ass2} gives that 
\begin{align*}
    \norm{a_3} \leq  \sqrt{C(\eta, t)}\norm{\tau_l - \hat{\tau}_l} = o_p(1).
\end{align*} \\ \\
$(a_4):$ Noticing that 
$$
(\varphi(O) - \tau(x))^2 = \left( - \left(\frac{\mathbb{1}(A=1)}{\pi(1\mid x)} - \frac{\mathbb{1}(A=0)}{\pi(0\mid x)} \right)\int_0^t \frac{S(t\mid A, x)}{S(s\mid A, x)S_c(s\mid A, x)} \dd M(s\mid A, x) \right)^2
$$
in the survival setting and 
$$
(\varphi(O) - \tau(x))^2 = \left( - \left(\frac{\mathbb{1}(A=1)}{\pi(1\mid x)} - \frac{\mathbb{1}(A=0)}{\pi(0\mid x)} \right)\int_0^t \frac{H(s,t\mid A, x)}{S(s\mid A, x)S_c(s\mid A, x)} \dd M(s\mid A, x) \right)^2
$$
in the RMST setting, the calculations from $(a_3)$ gives that
$$
    \norm{a_4} \leq  \sqrt{C(\eta, t)}\norm{\tau - \hat{\tau}} = o_p(1)
$$
with $C(\eta, t) = \max\{4\eta^{-2}, 8\eta^{-2}t^2\}$. \\\\
$(a_5):$ By Lemma \ref{emp:ate} and assumption \ref{ass6}, $\norm{a_5} = o_p(1)$.

\subsection{Proof of Theorem \ref{asGammaAndChi}}
\subsubsection{Remainder term related to $\hat{\Gamma}_j^{CF}$}
Let $\phi_{\Gamma_j}(\nu)$ be the uncentered version of the EIF $\tilde{\psi}_{\Gamma_j}$ at nuisance parameter $\nu$. Consider the decomposition
\begin{align*}
    & P\{ \phi_{\Gamma_j}(\hat{\nu}) - \phi_{\Gamma_j}(\nu) \} \\
    =& \E\left\{ [\varphi(\hat{\nu})(O) - \hat{\tau}_j(X)][X_j - \hat{E}_n^j(X_{-j})] - [\varphi(\nu)(O) - \tau_j(X)][X_j - E(X_j \mid X_{-j})] \right\} \\
    =& \E\left\{ [\varphi(\hat{\nu})(O) - \varphi(\nu)(O)]X_j + [\tau_j(X) - \hat{\tau}_j(X)]X_j \right. \\
    &- \left. \varphi(\hat{\nu})(O)\hat{E}_n^j(X_{-j}) + \hat{\tau}_j(X)\hat{E}_n^j(X_{-j})  + \varphi(\nu)(O)E(X_j \mid X_{-j}) - \tau_j(X)E(X_j \mid X_{-j}) \right\} \\
    =& E\left\{ [\hat{\varphi}(\nu)(O) - \varphi(\nu)(O)][X_j - \hat{E}_n^j(X_{-j})] \right. \\
    &- [\hat{\tau}_j(X) - \tau_j(X)][X_j - \hat{E}_n^j(X_{-j})] \\
    &-\left. [\hat{E}_n^j(X_{-j}) - E(X_j\mid X_{-j})][\varphi(\nu)(O) - \tau_j(X)]  \right\}.
\end{align*}
We note that 
\begin{align*}
\E\{\varphi(\nu)(O) - \tau_j(X)\} = \E\{ E(\tau(X)\mid X_{-j}) + E(M\mid A, X) - \tau_j(X) \} = 0
\end{align*}
by iterated expectation, where $M$ is the martingale integral in the expression of $\varphi$, which is itself a martingale conditional on $A$ and $X$. Thus, by iterated expectation and assumption \ref{assB}, the remainder term is given by
\begin{align*}
    & E\left\{ [\hat{\varphi}(\nu)(O) - \varphi(\nu)(O)][X_j - \hat{E}_n^j(X_{-j})] - [\hat{\tau}_j(X) - \tau_j(X)][X_j - \hat{E}_n^j(X_{-j})] \right\} \\
    \leq & \  \sqrt{\delta} | P \{  \varphi(\hat{\nu}) - \tau \} | + \norm{\hat{\tau}_j - \tau_j}\norm{\hat{E}_n^j - E(\cdot \mid X_{-j})} \\
    = & \ o_p(n^{-1/2}) 
\end{align*}
Here, the inequality is given by Cauchy-Schwarz, the triangle inequality and assumption \ref{assb2}. The equality is given by assumption \ref{ass3}, \ref{assb1} and Lemma \ref{rem:ate}.
\subsubsection{Empirical process term related to $\hat{\Gamma}_j^{CF}$}
Consider again the decomposition from the remainder term:
\begin{align*}
    \phi_{\Gamma_j}(\hat{\nu}) - \phi_{\Gamma_j}(\nu) =& [\hat{\varphi}(\nu)(O) - \varphi(\nu)(O)][X_j - \hat{E}_n^j(X_{-j})] \\
    &- [\hat{\tau}_j(X) - \tau_j(X)][X_j - \hat{E}_n^j(X_{-j})] \\
    &- [\hat{E}_n^j(X_{-j}) - E(X_j\mid X_{-j})][\varphi(\nu)(O) - \tau_j(X)].
\end{align*}
Thus, we need to bound each term in $L_2(P)$. For the first term we have
\begin{align*}
    \E\left\{[\hat{\varphi}(\nu)(O) - \varphi(\nu)(O)]^2[X_j - \hat{E}_n^j(X_{-j})]^2\right\} \leq \delta^2 \norm{\varphi(\hat{\nu}) - \varphi(\nu)}^2 = o_p(1) 
\end{align*}
by Lemma \ref{emp:ate} and assumption \ref{assb2}. Consistency of the second term follows from the consistency of $\hat{\tau}$ together with assumption \ref{assb2}, and consistency of the third term follows from consistency of $\hat{E}_n^j$ together with the bound of $P\{\varphi - \tau_j\}^2$ calculated in $(a3)$ in the empirical process section of the proof of Theorem \ref{asTheta}.

\subsubsection{Remainder term related to $\hat{\chi}_j^{CF}$}
We note that 
\begin{align*}
    &\E\left\{\left[ X_j - E(X_j\mid X_{-j}) \right]^2 \right\} \\
    =& \ \E\left\{ X_j^2 + E(X_j\mid X_{-j})^2 \right\} - 2\E\left\{X_jE(X_j\mid X_{-j}) \right\} \\
    =& \ \E\left\{ X_j^2 + E(X_j\mid X_{-j})^2 \right\} - 2\E\left\{E(X_j\mid X_{-j})^2 \right\} \\
    =& \ \E\left\{ X_j^2 - E(X_j\mid X_{-j})^2 \right\} 
\end{align*}
by iterated expectation. Hence
\begin{align*}
    & P\left\{ \phi_{\chi_j}(\hat{\nu}) - \phi_{\chi_j}(\nu) \right\} \\
    =&  \E\left\{ \phi_{\chi_j}(\hat{\nu}) - X_j^2 + E(X_j\mid X_{-j})^2 \right\} \\
    =&  \E\left\{ X_j^2 + \hat{E}_n^{j2}(X_{-j}) - 2X_j\hat{E}_n^{j2}(X_{-j}) - X_j^2 + E(X_j\mid X_{-j})^2 \right\} \\
    =&  \E\left\{ \left[ \hat{E}_n^{j}(X_{-j}) - E(X_j\mid X_{-j}) \right]^2 + 2\hat{E}_n^{j}(X_{-j})E(X_j\mid X_{-j}) - 2X_j\hat{E}_n^{j2}(X_{-j}) \right\} \\
    =& \norm{ \hat{E}_n^{j} - E(\cdot \mid X_{-j}) }^2 \\
    =& o_p(n^{-1/2}),
\end{align*}
where the fourth equality is due to iterated expectation and the last equality is given by assumption \ref{assb1}.

\subsubsection{Empirical process term related to $\hat{\chi}_j^{CF}$}
Consider 
\begin{align*}
    &\phi_{\chi_j}(\hat{\nu}) - \phi_{\chi_j}(\nu) \\
    =& \left[ \hat{E}_n^j(X_{-j}) - E(X_j\mid X_{-j}) \right]^2 + 2\left[ E(X_j\mid X_{-j}) - X_j \right]\left[ \hat{E}_n^j(X_{-j}) - E(X_j\mid X_{-j}) \right],
\end{align*}
where the consistency of $\hat{E}_n^j$ gives consistency in $L_2(P)$ of the first term. For the second term, note that 
$$E( [ E(X_j\mid X_{-j}) - X_j ]^2 \mid X_{-j} ) = \var(X_j\mid X_{-j})$$
and hence 
\begin{align*}
    & E\left\{ 4\left[ E(X_j\mid X_{-j}) - X_j \right]^2\left[ \hat{E}_n^j(X_{-j}) - E(X_j\mid X_{-j}) \right]^2 \right\} \\
    = & 4E\left\{ \var(X_j\mid X_{-j}) \left[ \hat{E}_n^j(X_{-j}) - E(X_j\mid X_{-j}) \right]^2 \right\} \\
    \leq & 4 K \norm{\hat{E}_n^j- E(\cdot\mid X_{-j})}^2 \\
    = & o_p(1)
\end{align*}
for some $K>0$, by assumption \ref{assb1} and boundedness of $\var(X_j\mid X_{-j})$, which gives the result.

\subsection{Consistency of cross-fitted variance estimators}\label{app:consVar}
\begin{proof}[\textbf{Proof of Lemma \ref{lem:SECF}}]
    The first claim, \eqref{lem:SECF1}, is given by the functional delta method (\cite{vaart}, ch. 25.7) and hence $n^{1/2}(\hat{\Psi}^{CF}- \Psi) \overset{D}{\rightarrow} \mathcal{N}(0, P\tilde{\psi}^2)$. For the claim \eqref{lem:SECF2} we note that it is suffices to show that
    $$
    \mathbb{P}_n^k \tilde{\psi}(\hat{\psi}_i^{CF}, \hat{\nu}_{i, -k})^2 \overset{P}{\longrightarrow} P\tilde{\psi}(\psi_i, \nu_i)^2
    $$
    for each $k$, since $K$ is assumed finite and not depending on $n$. We note that 
    \begin{align}\label{decomp:varCF1}
        \mathbb{P}_n^k \tilde{\psi}(\hat{\psi}_i^{CF}, \hat{\nu}_{i, -k})^2 
        = \hat{\psi}_i^{2, CF} + \mathbb{P}_n^k \varphi_i(\hat{\nu}_{i, -k})^2 - 2 \hat{\psi}_i^{CF} \mathbb{P}_n^k \varphi_i(\hat{\nu}_{i, -k})
    \end{align}
    and that $\hat{\psi}_i^{2, CF} \overset{P}{\rightarrow} \psi_i^2$ by the continuous mapping theorem, since $\hat{\psi}_i^{CF} \overset{P}{\rightarrow} \psi_i$ by the assumption of asymptotic linearity. Hence, by the continuous mapping theorem, the result follows if each of the $\mathbb{P}_n^k$-sums in the above display are consistent. Observe that
    \begin{align*}
        \mathbb{P}_n^k\varphi_i(\hat{\nu}_{i, -k}) - P\varphi(\nu) = (\mathbb{P}_n^k - P)(\varphi_i(\hat{\nu}_{i, -k}) - \varphi(\nu)) + P(\varphi_i(\hat{\nu}_{i, -k} - \varphi(\nu)) + (\mathbb{P}_n^k - P)\varphi_i(\nu).
    \end{align*}
    The first term above is $o_p(n^{-1/2})$ by Lemma 2 in the supplementary material of \textcite{kennedy2020}, since $\norm{\varphi(\hat{\nu}) - \varphi(\nu)} = o_p(1)$ by assumption. The second term is $o_p(1)$ since $P(\varphi_i(\hat{\nu}_{i, -k}) - \varphi(\nu)) \leq \norm{\varphi_i(\hat{\nu}_{i, -k}) - \varphi(\nu)} = o_p(1)$ by assumption and the third term is $o_p(1)$ by the law of large numbers, since $\E\{\varphi_i(\nu_i)(O_i)\} = \psi_i < \infty$. Thus $ \mathbb{P}_n^k \varphi_i(\hat{\nu}_{i, -k})$ converges to $P\varphi_i(\nu_i)$ in probability. For the sum $ \mathbb{P}_n^k \varphi_i(\hat{\nu}_{i, -k})^2$ in \eqref{decomp:varCF1}, we note that since $\norm{\varphi_i(\hat{\nu}_{i, -k}) - \varphi_i(\nu_i)} = o_p(1)$, by assumption, the continuous mapping theorem (for metric spaces, see e.g. \cite{vaart}, Theorem 18.11) gives that $\norm{\varphi_i(\hat{\nu}_{i, -k})^2 - \varphi_i(\nu_i)^2} = o_p(1)$. Combined with the fact that $\E\{\varphi_i(\nu_i)(O_i)^2\} < \infty$, we can use the same arguments given for the consistency of $\mathbb{P}_n^k \varphi_i(\hat{\nu}_{i, -k})$, to show that $\mathbb{P}_n^k \varphi_i(\hat{\nu}_{i, -k})^2$ converges in probability to $P\varphi_i(\nu_i)^2$. Collecting the results, the continuous mapping theorem now gives the following convergence for \eqref{decomp:varCF1}: 
    \begin{align*}
        \hat{\psi}_i^{2, CF} + \mathbb{P}_n^k \varphi_i(\hat{\nu}_{i, -k})^2 - 2 \hat{\psi}_i^{CF} \mathbb{P}_n^k \varphi_i(\hat{\nu}_{i, -k}) \overset{P}{\longrightarrow} \psi_i^2 + P \varphi_i(\nu_{i})^2 - 2 \psi_i P \varphi_i(\nu_{i}) = P\tilde{\psi}_i(\psi_i, \nu_i)^2,
    \end{align*}
    and hence, the claim, \eqref{lem:SECF2}, follows. 

    As for the second claim, the last claim, \eqref{lem:SECF3}, follows if      
    $$
    \mathbb{P}_n^k \tilde{\psi}(\hat{\psi}_1^{CF}, \hat{\psi}_2^{CF}, \hat{\nu}_{1, -k}, \hat{\nu}_{2, -k})^2 \overset{P}{\longrightarrow} P\tilde{\psi}(\psi_1, \psi_2, \nu_1, \nu_2)^2
    $$
    for each $k$, since $K$ is assumed finite and not depending on $n$. Observe that
    \begin{align*}
       &\mathbb{P}_n^k \tilde{\psi}(\hat{\psi}_1^{CF}, \hat{\psi}_2^{CF}, \hat{\nu}_{1, -k}, \hat{\nu}_{2, -k})^2 \\
       =& \mathbb{P}_n^k \frac{1}{\left(\hat{\psi}_2^{CF}\right)^2}\left(\varphi_1(\hat{\nu}_{1, -k}) - \hat{\psi}_1^{CF} - \frac{\hat{\psi}_1^{CF}}{\hat{\psi}_2^{CF}} \left(\varphi_2(\hat{\nu}_{2, -k}) - \hat{\psi}_2^{CF} \right) \right)^2 \\
       =& \frac{1}{\left(\hat{\psi}_2^{CF}\right)^2} \mathbb{P}_n^k \left(\varphi_1(\hat{\nu}_{1, -k}) - \hat{\psi}_1^{CF} \right)^2 + \left(\hat{\Psi}^{CF}\right)^2\mathbb{P}_n^k \left(\varphi_2(\hat{\nu}_{2, -k}) - \hat{\psi}_2^{CF} \right)^2 \\
       &\phantom{\frac{1}{\left(\hat{\psi}_2^{CF}\right)^2}} - 2\hat{\Psi}^{CF}\mathbb{P}_n^k \left(\varphi_1(\hat{\nu}_{1, -k}) - \hat{\psi}_1^{CF} \right) \left(\varphi_2(\hat{\nu}_{2, -k}) - \hat{\psi}_2^{CF} \right).    
    \end{align*}
    We will consider each term in the above display separately. In the proof of \eqref{lem:SECF2}, we showed consistency of the $\mathbb{P}_n^k$-sums in the first two terms, and the continuous mapping theorem gives that $\frac{1}{\left(\hat{\psi}_2^{CF}\right)^2} \overset{P}{\rightarrow} \frac{1}{\psi_2^2}$ and $\left(\hat{\Psi}^{CF}\right)^2 \overset{P}{\rightarrow} \Psi^2$. The continuous mapping theorem then gives consistency of the first two terms. For the last term, we use the decomposition
    \begin{align*}
    &2\hat{\Psi}^{CF}\mathbb{P}_n^k \left(\varphi_1(\hat{\nu}_{1, -k}) - \hat{\psi}_1^{CF} \right) \left(\varphi_2(\hat{\nu}_{2, -k}) - \hat{\psi}_2^{CF} \right) \\
    =&  2\hat{\Psi}^{CF}\left(\hat{\psi}_1^{CF} \hat{\psi}_2^{CF} - \hat{\psi}_1^{CF}\mathbb{P}_n^k \varphi_2(\hat{\nu}_{2,-k}) - \hat{\psi}_2^{CF}\mathbb{P}_n^k \varphi_1(\hat{\nu}_{1,-k}) + \mathbb{P}_n^k \varphi_1(\hat{\nu}_{1,-k})\varphi_2(\hat{\nu}_{2,-k}) \right).   
    \end{align*}
    Consistency of the three first terms inside the parenthesis are shown in the in the proof of \eqref{lem:SECF2} and hence, we only need to show consistency of the last term. Here, it suffices to show consistency of $\varphi_1(\hat{\nu}_{1,-k})\varphi_2(\hat{\nu}_{2,-k})$ together with $P\varphi_1(\nu_1)\varphi_2(\nu_2) < \infty$, from which the result follows by the same arguments used to show consistency of $\mathbb{P}_n^k\varphi_i(\hat{\nu}_{i,-k}).$ For the latter, Cauchy-Schwarz and the triangle inequality gives
    $$
    P\varphi_1(\nu_1)\varphi_2(\nu_2) \leq \norm{\tilde{\psi}_1 + \psi_1}\norm{\tilde{\psi}_2 + \psi_2} \leq \left(\norm{\tilde{\psi}_1} + \norm{\psi_1}\right)\left(\norm{\tilde{\psi}_2} + \norm{\psi_2} \right) < \infty.
    $$
    By Theorem 18.10 in \textcite{vaart} $(\varphi_1(\hat{\nu}_{1,-k}),\varphi_2(\hat{\nu}_{2,-k})) \overset{P}{\rightarrow} (\varphi_1(\nu_1),\varphi_2(\nu_2))$ since $\varphi_i(\hat{\nu}_{i,-k}) \overset{P}{\rightarrow} \varphi_i(\nu_i)$ by assumption. Then, the continuous mapping theorem (18.11 in \cite{vaart}) gives that $\varphi_1(\hat{\nu}_{1,-k})\varphi_2(\hat{\nu}_{2,-k}) \overset{P}{\rightarrow} \varphi_1(\nu_1)\varphi(\nu_2)$, and the result follows.
\end{proof}

\clearpage 

\section{Supplementary Material}

\section{Additional simulation results for $\Omega_1$}
\begin{figure}[h!] 
    \centering
    \begin{subfigure}[c]{0.9\linewidth}
    \centering
    \includegraphics[scale = 0.65]{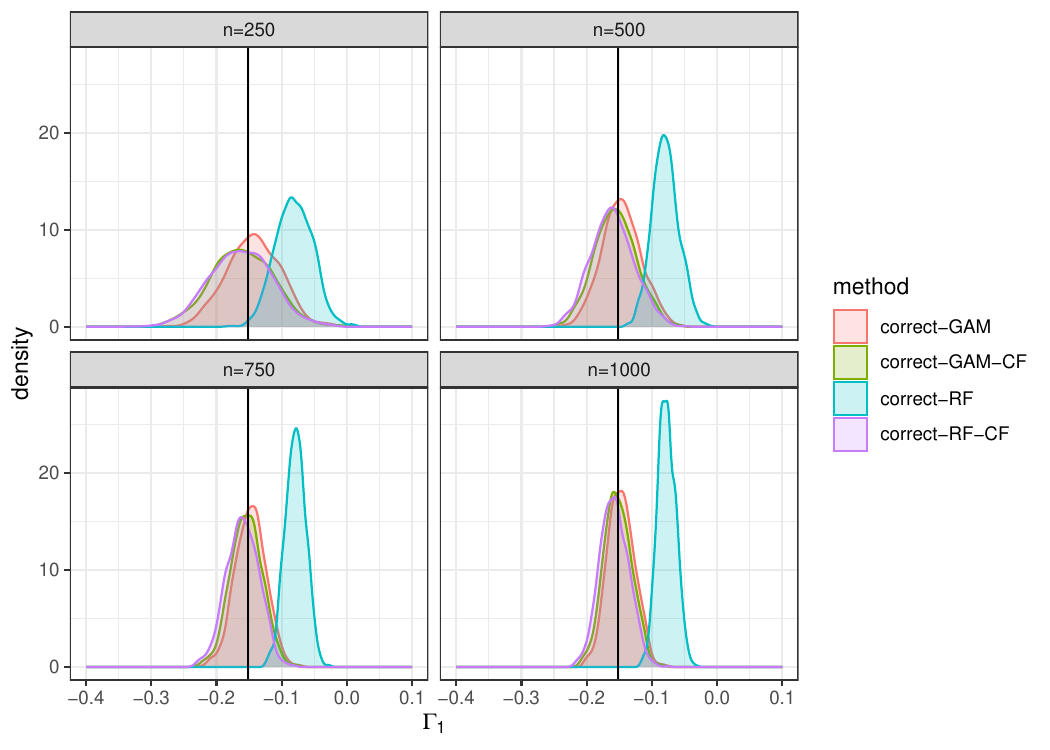}
    \caption{Correctly specified $\Lambda$, $\Lambda_c$ and $\pi$}
    \label{fig:simGammaCorrect}
    \end{subfigure}

    \begin{subfigure}[c]{0.9\linewidth}
    \centering
    \includegraphics[scale = 0.65]{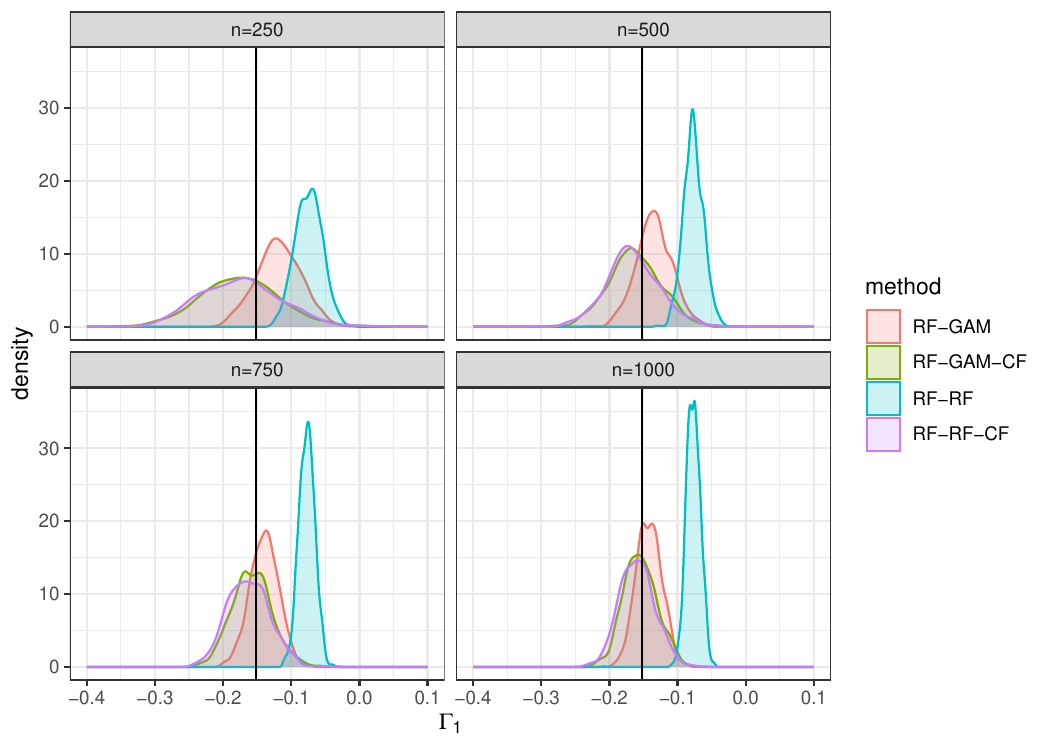}
    \caption{Flexible estimation of $\Lambda$, $\Lambda_c$ and $\pi$}
    \label{fig:simGammaRF}
    \end{subfigure}
    \caption{Sampling distribution of estimators of $\Gamma_1$ in the survival function setting with varying nuisance estimators, with and without cross-fitting, across sample sizes $n=250,500,750,1000$. The abbreviations of the methods are read as follows: A-B-C, where A corresponds to the nuisance estimators $\Lambda$, $\Lambda_c$ and $\pi$, B corresponds to the nuisance estimators $\hat{E}_n$ and $\hat{E}^j_n$, and C corresponds to whether or not cross-fitting was used. Here, \textit{correct} corresponds to correctly specified Cox and logistic regression, RF corresponds to Random Forest, and GAM corresponds to a generalized additive model.}
\end{figure}

\begin{figure}[h!] 
    \centering
    \begin{subfigure}[c]{0.9\linewidth}
    \centering
    \includegraphics[scale = 0.65]{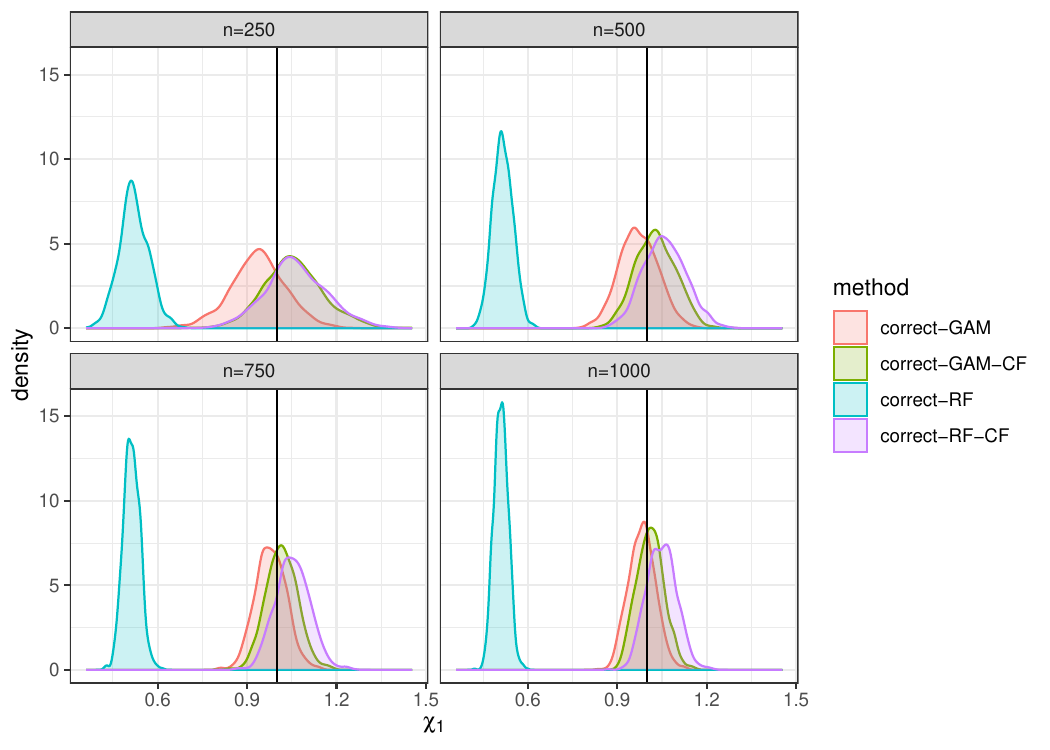}
    \caption{Correctly specified $\Lambda$, $\Lambda_c$ and $\pi$}
    \label{fig:simChiCorrect}
    \end{subfigure}

    \begin{subfigure}[c]{0.9\linewidth}
    \centering
    \includegraphics[scale = 0.65]{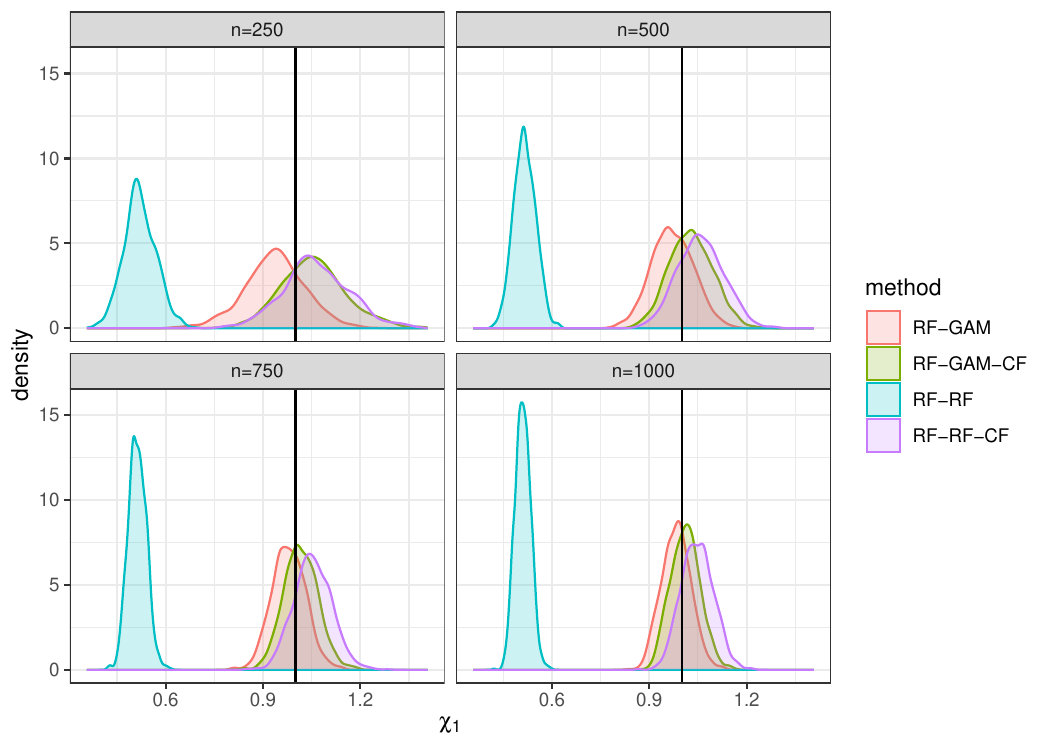}
    \caption{Flexible estimation of $\Lambda$, $\Lambda_c$ and $\pi$}
    \label{fig:simChiRF}
    \end{subfigure}
    \caption{Sampling distribution of estimators of $\chi_1$ in the survival function setting with varying nuisance estimators, with and without cross-fitting, across sample sizes $n=250,500,750,1000$. The abbreviations of the methods are read as follows: A-B-C, where A corresponds to the nuisance estimators $\Lambda$, $\Lambda_c$ and $\pi$, B corresponds to the nuisance estimators $\hat{E}_n$ and $\hat{E}^j_n$, and C corresponds to whether or not cross-fitting was used. Here, \textit{correct} corresponds to correctly specified Cox and logistic regression, RF corresponds to Random Forest, and GAM corresponds to a generalized additive model.}
\end{figure}

\end{document}